%% file: main.tex
\pdfoutput=1

\documentclass[sigconf]{acmart}

\renewcommand\footnotetextcopyrightpermission[1]{} 
\AtBeginDocument{%
  \providecommand\BibTeX{{%
    \normalfont B\kern-0.5em{\scshape i\kern-0.25em b}\kern-0.8em\TeX}}}

\usepackage{enumitem}
\input{./sections/macros.tex}

\begin{document}

\title{Expressive power of linear algebra query languages}
\author{Floris Geerts}
\affiliation{%
  \institution{University of Antwerp}
}
\email{floris.geerts@uantwerp.be}

\author{Thomas Mu\~noz}
\affiliation{%
  \institution{PUC Chile and IMFD Chile}
}
\email{tfmunoz@uc.cl}

\author{Cristian Riveros}
\affiliation{%
  \institution{PUC Chile and IMFD Chile}
}
\email{cristian.riveros@uc.cl}

\author{Domagoj Vrgo\v{c}}
\affiliation{%
  \institution{PUC Chile and IMFD Chile}
}
\email{dvrgoc@ing.puc.cl}


\begin{abstract}
Linear algebra algorithms often require some sort of iteration or recursion as is illustrated by standard algorithms for Gaussian elimination, matrix inversion, and transitive closure. A key characteristic shared by these 
algorithms is that they allow looping for a number of steps that is bounded by the matrix dimension. 
In this paper we extend the matrix query language \lang with this type of recursion, and show that this suffices to express  classical linear algebra algorithms. We study the expressive power of this language and show that it naturally corresponds to arithmetic circuit families, which are often said to capture linear algebra. Furthermore, we analyze several sub-fragments of our language, and show that their expressive power is closely tied to logical formalisms on semiring-annotated relations.
\end{abstract}

%
%
\maketitle

\section{Introduction}
\input{./sections/intro.tex}

\section{MATLANG}\label{sec:matlang}
\input{./sections/prelim.tex}
\section{Extending MATLANG with for loops}\label{sec:formatlang}
\input{./sections/forloops.tex}

\section{Algorithms in Linear Algebra}\label{sec:queries}
\input{./sections/algorithms.tex}
\section{Expressiveness of for loops}\label{sec:circuits}

\input{./sections/expr.tex}

\section{Restricting the power of for loops}\label{sec:restrict}
\input{./sections/restrictions.tex}



\section{Conclusions}\label{sec:conclude}
\input{./sections/concl.tex}


\bibliographystyle{ACM-Reference-Format}

\bibliography{nourlbiblio}


\newpage 
\appendix

\onecolumn
\section*{Appendix}
\input{./sections/appendix.tex}

\end{document}

%% file: sections/macros.tex
\usepackage{xspace}

\usepackage[textwidth=2cm,textsize=small]{todonotes}

\newcommand{\domagoj}[1]{\todo[inline, color=blue!15]{{\bf Domagoj:} #1}}
\newcommand{\cristian}[1]{\todo[inline, color=orange!30]{{\bf Cristian:} #1}}
\newcommand{\thomas}[1]{\todo[inline, color=green!30]{{\bf Thomas:} #1}}
\newcommand{\floris}[1]{\todo[inline, color=red!30]{{\bf Floris:} #1}}

\renewcommand{\domagoj}[1]{}
\renewcommand{\cristian}[1]{}
\renewcommand{\thomas}[1]{}
\renewcommand{\floris}[1]{}

\usepackage{algorithm} 
\usepackage[noend]{algpseudocode}

\usepackage{tikz}
\usepackage{pgfplots}

\pgfplotsset{compat=1.12}
\usetikzlibrary{arrows,decorations.pathmorphing,backgrounds,positioning,fit,calc,automata,matrix}
\usetikzlibrary{trees}
\usetikzlibrary{shapes}
\usetikzlibrary{chains}
\usetikzlibrary{patterns}
\usetikzlibrary{tikzmark}

\usepackage{xspace}
\usepackage[overload]{empheq}
\usepackage{mathtools}
\DeclarePairedDelimiter{\ceil}{\lceil}{\rceil}

\newcommand{\NN}{\mathbb{N}}

\newcommand{\RR}{\mathbb{R}}

\newcommand{\cS}{\mathcal{S}}

\newcommand{\cV}{\mathcal{V}}

\newcommand{\ssum}{\Sigma}
\newcommand{\sprod}{\Pi}

\newtheorem*{mylem*}{Lemma}

\newtheorem*{myclaim*}{Claim}
\newtheorem*{myprop*}{Proposition}

{}

\newcommand{\cG}{\mathcal{G}}

\newcommand{\sem}[2]{\llbracket #1 \rrbracket(#2)}

\newcommand{\Mnam}{\mathcal{M}}
\newcommand{\Mvar}{\mathcal{V}}
\newcommand{\Fun}{\mathcal{F}}

\newcommand{\llet}{\texttt{let } V = e_1 \texttt{ in } e_2}
\newcommand{\ones}{\mathbf{1}}
\newcommand{\diag}{\texttt{diag}}

\newcommand{\I}{\mathcal{I}}

\newcommand{\Sch}{\mathcal{S}}

\newcommand{\mtr}[1]{\textsf{Mat}[#1]}

\newcommand{\dom}{\mathcal{D}}
\newcommand{\conc}{\texttt{mat}}

\newcommand{\DD}{\texttt{Symb}}
\newcommand{\size}{\texttt{size}}

\newcommand{\ttype}{\texttt{type}_{\Sch}}

\newcommand{\type}{\texttt{type}}

\newcommand{\rak}{\texttt{RA}$_K^+$\xspace}
\newcommand{\lara}{\texttt{LARA}\xspace}
\newcommand{\lang}{\texttt{MATLANG}\xspace}
\newcommand{\langf}[1]{\texttt{MATLANG}[#1]}
\newcommand{\langfor}{\texttt{for}\text{-}\texttt{MATLANG}\xspace}
\newcommand{\langforf}[1]{\texttt{for}\text{-}\texttt{MATLANG}[#1]\xspace}
\newcommand{\langsum}{\texttt{sum}-\texttt{MATLANG}\xspace}
\newcommand{\langprod}{\texttt{FO}-\texttt{MATLANG}\xspace}
\newcommand{\langmprod}{\texttt{prod}-\texttt{MATLANG}\xspace}

\newcommand{\ffor}[3]{\texttt{for}\, #1,#2 \texttt{.}\, #3}

\newcommand{\initf}[4]{\texttt{for}\, #2,#3\!=\! #1 \texttt{.}\, #4}

\newcommand{\mmin}[1]{\texttt{min}(#1)}

\newcommand{\ccol}[2]{\texttt{col(}#1,#2\texttt{)}}
\newcommand{\red}[2]{\texttt{reduce(}#1,#2\texttt{)}}
\newcommand{\nneq}[2]{\texttt{neq(}#1,#2\texttt{)}}
\newcommand{\ccoleq}[2]{\texttt{col}_{\texttt{eq}}\texttt{(}#1,#2\texttt{)}}

\newcommand{\push}[2]{#1\texttt{.push}(#2)}
\newcommand{\pop}[1]{#1\texttt{.pop}}
\newcommand{\getsize}[1]{#1\texttt{.size}}
\newcommand{\gettop}[1]{#1\texttt{.top}}

\newcommand{\isplus}[1]{\texttt{isplus}\left( #1 \right)}
\newcommand{\isprod}[1]{\texttt{isprod}\left(#1 \right)}
\newcommand{\isone}[1]{\texttt{isone}\left(#1 \right)}
\newcommand{\isinput}[1]{\texttt{isinput}\left(#1 \right)}
\newcommand{\getfirst}[1]{\texttt{getfirst}\left(#1 \right)}
\newcommand{\getinput}[1]{\texttt{getinput}\left(#1 \right)}
\newcommand{\getroot}{\texttt{getroot}()}
\newcommand{\isnotlast}[2]{\texttt{not{\_}last}\left(#1,#2 \right)}
\newcommand{\nextgate}[2]{\texttt{next{\_}gate}\left(#1, #2 \right)}

\newcommand{\Iden}[1]{\texttt{Iden}(#1)}
\newcommand{\pondIden}[2]{\texttt{E}\left[ #1, #2 \right]}

\newcommand{\logspace}{{\sc LOGSPACE}}

\newcommand{\ddom}{\mathbb{D}}
\newcommand{\fdom}{\operatorname{dom}}
\newcommand{\att}{\mathbb{A}}
\newcommand{\tuples}{\mathbf{tuples}}
\newcommand{\supp}{\operatorname{supp}}
\newcommand{\cJ}{\mathcal{J}}
\newcommand{\cR}{\mathcal{R}}
\newcommand{\adom}{\mathbf{adom}}

\newcommand{\ksum}{\oplus}
\newcommand{\kprod}{\odot}
\newcommand{\bigksum}{\bigoplus}
\newcommand{\bigkprod}{\bigodot}
\newcommand{\kzero}{\mymathbb{0}}
\newcommand{\kone}{\mymathbb{1}}

\newcommand{\row}{\mathsf{row}}

\newcommand{\col}{\mathsf{col}}

\newcommand{\arae}{Q}

\newcommand{\ssem}[2]{\llbracket #1 \rrbracket_{#2}}

\newcommand{\hadprod}{\circ} 
\newcommand{\qhadprod}{\Pi^{\hadprod}}

\newcommand{\cA}{\mathcal{A}}
\newcommand{\arity}{\operatorname{arity}}

\DeclareMathAlphabet{\mymathbb}{U}{BOONDOX-ds}{m}{n}

%% file: sections/intro.tex
Linear algebra-based algorithms have become a key component in data analytic workflows. As such, there is a growing interest in the database community to integrate linear algebra functionalities into relational database management systems \cite{Jermaine/17/LAonRA,2019Boehm,LARA_Berlin_2016,JankovLYCZJG19,Khamis0NOS18}. In particular, from a query language perspective, several proposals have recently been put forward to unify relational algebra and linear algebra. Two notable examples of this are: \lara~\cite{HutchisonHS17}, a minimalistic language in which a number of atomic operations on 
associative tables are proposed, and \lang, a query language for 
matrices \cite{matlang-journal}.\looseness=-1

Both \lara and \lang have been studied by the database theory community, showing interesting connections to relational algebra and logic. For example, fragments of \lara are known to capture first-order logic with aggregation~\cite{BarceloH0S20}, and \lang has been recently shown to be equivalent to a restricted version of the (positive) relational algebra on $K$-relations, \rak~\cite{brijder2019matrices}, where $K$ denotes a semiring. On the other hand, 
some standard constructions in linear algebra 
are out of reach for these languages. For instance, it was shown that under standard complexity-theoretic assumptions, \lara can not compute the inverse of a matrix or its determinant~\cite{BarceloH0S20}, and operations such as the transitive closure of a matrix are known to be inexpressible in \lang~\cite{matlang-journal}. Given that these are fundamental constructs in linear algebra, one might wonder how to extend \lara or \lang in order to allow expressing such properties.

One approach would be to add these constructions explicitly to the language. Indeed, this was done for \lang in~\cite{matlang-journal}, and \lara in ~\cite{BarceloH0S20}. In these works, the authors have extended the core language with the trace, the inverse, the determinant, or the eigenvectors operators and study the expressive power of the result. However, one can argue that there is nothing special in these operators, apart they have been used historically in linear algebra textbooks and they extend the expressibility of the core language. The question here is whether these new operators form a sound and natural choice to extend the core language, or are they just some particular queries that we would like to support. 

In this paper we take a more principled approach by studying what are the atomic operations needed to define standard linear algebra algorithms. Inspecting any linear algebra textbook, one sees that most linear algebra procedures heavily rely on the use of for-loops in which iterations happen over the dimensions of the matrices involved. To illustrate this, let us consider the example of computing the transitive closure of a graph. This can be done using a modification of the Floyd-Warshall algorithm~\cite{cormen}, which takes as its input an $n\times n$ adjacency matrix $A$ representing our graph, and operates according to the following pseudo-code:
\begin{tabbing}
\quad\texttt{for}\=\,  $k = 1..n$ \texttt{do}\\
\> \texttt{for}\=\,  $i = 1..n$ \texttt{do}\\
\> \> \texttt{for}\=\,  $j = 1..n$ \texttt{do}\\
\> \> \> $A[i,j] := A[i,j] + A[i,k] \cdot A[k,j]$
\end{tabbing}
After executing the algorithm, all of the non zero entries signify an edge in the (irreflexive) transitive closure graph.
%
\cristian{I am not sure about this sentence. Depending on the order how is iterated, maybe you can compute more than $n$. Actually, if you don't add the identity to the original matrix, even the transitive closure will not work.}


By
examining standard linear algebra algorithms such as Gaussian elimination, $LU$-decomposition, computing the inverse of a matrix, or its determinant, we can readily see that this pattern continues. Namely, we observe that there are two main components to such algorithms: (i) the ability to iterate up to the matrix dimension; and (ii) the ability to access a particular position in our matrix. In order to allow this behavior in a query language, we propose to extend \lang with limited recursion in the form of for-loops, resulting in the language \langfor. To  simulate the two components of standard linear algebra algorithms in a natural way, we simulate a loop of the form \texttt{for}\, $i=1..n$ \texttt{do} by leveraging canonical vectors. In other words, we use the canonical vectors $b_1=(1,0,\ldots)$, $b_2=(0,1,\ldots)$, \ldots, to access specific rows and columns, and iterate over these vectors. In this way,
we obtain a language able to compute important linear algebra operators such as $LU$-decomposition, determinant, matrix inverse, among other things.

Of course, a natural question to ask now is whether this really results in a language suitable for linear algebra? We argue that the correct way to approach this question is to compare our language to arithmetic circuits, which have been shown to capture the vast majority of existing matrix algorithms, from basic ones such as computing the determinant and the inverse, to complex procedures such as discrete Fourier transformation, and Strassen's algorithm (see \cite{ShpilkaY10,allender} for an overview of the area), and can therefore be considered to effectively capture linear algebra. In the main technical result of this paper, we show that \langfor indeed computes the same class of functions over matrices as the ones computed by arithmetic circuit families of bounded degree.  As a consequence, \langfor inherits all expressiveness properties of circuits, and thus can simulate any linear algebra algorithm definable by circuits.

Having established that \langfor indeed provides a good basis for a linear algebra language, we move to a more fine-grained analysis of the expressiveness of its different fragments. For this, we aim to provide a connection with logical formalisms, similarly as was done by linking \lara and \lang to first-order logic with aggregates~\cite{BarceloH0S20,matlang-journal}. As we show, capturing different logics correspond to restricting how matrix variables are updated in each iteration of the for-loops allowed in \langfor. For instance, if we only allow to add some temporary result to a variable in each iteration (instead of rewriting it completely like in any programming language), we obtain a language, called \langsum, which is equivalent to \rak, directly extending an analogous result shown for \lang, mentioned earlier~\cite{brijder2019matrices}. We then study updating matrix variables based on another standard linear algebra operator, the Hadamard product, resulting in a fragment called \langprod, which we show to be equivalent to weighted logics~\cite{DrosteG05}. Finally, in \langmprod 
we 
update the variables based on the standard matrix product, and link this fragment 
to the ones discussed previously.  

\smallskip
\noindent
\textbf{Contribution and outline.} 
\begin{itemize}[leftmargin=0.5cm]
	\item After we recall \lang in Section~\ref{sec:matlang}, we show in Section~\ref{sec:formatlang}
	how for-loops can be added to \lang in a natural way. We also observe that
	\langfor strictly extends \lang. In addition, we discuss some design decisions behind the definition of \langfor, noting that our use of canonical vectors results in the availability of an order relation.
	
	\item In Section~\ref{sec:queries} we show that \langfor can compute important linear algebra algorithms in a natural way. We provide expressions in \langfor for LU decomposition (used to solve linear systems of equations), the determinant and matrix inversion.
	\item More generally, in Section~\ref{sec:circuits} we report our main technical contribution.
	 We show that every  uniform arithmetic circuits of polynomial degree correspond to a \langfor expression, and vice versa, when a \langfor expression has polynomial degree, then there is an equivalent uniform family of arithmetic circuits. As a consequence, \langfor inherits all expressiveness properties of such circuits.
	\item  Finally, in Section~\ref{sec:restrict} we generalize the semantics of \langfor to matrices with values in a semiring $K$, and show that two natural fragment of \langfor, \langsum, and \langprod, are equivalent to the (positive) relational algebra and weighted logics on binary $K$-relations, respectively. We also briefly comment on a minimal fragment of \langfor, based on \langmprod, that is able to compute matrix inversion.
\end{itemize}
Due to space limitations, most proofs are referred to the appendix.

\smallskip
\noindent
\textbf{Related work.} 
We already mentioned \lara~\cite{HutchisonHS17} and \lang~\cite{matlang-journal}
whose expressive power was further analyzed in~\cite{BarceloH0S20,brijder2019matrices,Geerts19,Geerts20}.
Extensions of \texttt{SQL} for matrix manipulations are reported in~\cite{Jermaine/17/LAonRA}. Most relevant
is~\cite{JankovLYCZJG19} in which a recursion mechanism is added to \texttt{SQL} which resembles for-loops.
The expressive power of this extension is unknown, however. Classical logics with aggregation~\cite{Hella:2001} and fixed-point logics with counting~\cite{GroheP17} can also be used for linear algebra. More generally, for the descriptive complexity of linear algebra we refer to~\cite{dghl_rank,holm_phd}. Most of these works require to encode real numbers inside relations, whereas we treat real numbers as atomic values. We refer to relevant papers related to arithmetic circuits and logical formalisms on semiring-annotated relations in the corresponding sections later in the paper.

%% file: sections/prelim.tex
We start by recalling the matrix query language \lang, introduced in \cite{matlang-journal}, which serves as our starting point.

\smallskip
\noindent
\textbf{Syntax.}\,  Let $\Mvar = \{V_1, V_2, \ldots\}$ be a countably infinite set of \textit{matrix variables} and $\Fun=\bigcup_{k>1}\Fun_k$ with
$\Fun_k$ a set of \textit{functions} of the  form $f:\RR^k \to \RR$, where $\RR$ denotes the set of real numbers. The syntax of $\lang$ expressions is defined by the following grammar\footnote{The original syntax also permits the operator $\llet$, which replaces every occurrence of $V$ in $e_2$ with the value of $e_1$. Since this is just syntactic sugar, we omit this operator. We also explicitly include matrix addition and scalar multiplication, although these can be simulated by pointwise function applications. Finally, we use transposition instead of conjugate transposition since we work with matrices over $\RR$.}:

\begin{tabular}{lcll}
$e$ & $::=$ & $V\in \Mvar$ & (matrix variable)\\
 & $|$ & $e^T$ & (transpose)\\ 
 & $|$ & $\ones(e)$ & (one-vector)\\ 
 & $|$ & $\diag(e)$ & (diagonalization of a vector)\\  
 & $|$ & $e_1 \cdot e_2$ & (matrix multiplication)\\   
 & $|$ & $e_1 + e_2$ & (matrix addition)\\   
 & $|$ & $e_1\times e_2$ & (scalar multiplication)\\
 & $|$ & $f(e_1,\ldots ,e_k)$ & (pointwise application of $f\in\Fun_k$).    
\end{tabular}
\vspace{1ex}


$\lang$ is parametrized by a collection of functions $\Fun$ but in the remainder of the paper we only make this dependence explicit, and write $\langf{\Fun}$, for some set $\Fun$ of functions, when these functions are crucial for some results to hold. When we simply write \lang, we mean that any function can be used (including not using any function at all).




\smallskip
\noindent
\textbf{Schemas and typing.}\,
To define the semantics of \lang\ expressions we need a notion of schema and well-typedness of expressions. A \lang\ \textit{schema} $\Sch$ is a pair $\Sch=(\Mnam,\size)$, where $\Mnam\subset \Mvar$ is a finite set of matrix variables, and $\size: \Mnam \mapsto \DD\times \DD$ is a function that maps each matrix variable in $\Mnam$ to a pair of \textit{size symbols}. The $\size$ function helps us determine whether certain matrix operations, such as matrix multiplication, can be performed for matrices adhering to a schema. 
We denote size symbols by Greek letters $\alpha,\beta,\gamma$. We also assume that $1\in \DD$. 
To help us determine whether a \lang\ expression can always be evaluated, we define the \textit{type} of an expression $e$, with respect to a schema $\Sch$, denoted by $\ttype(e)$, inductively as follows:
\begin{itemize}
\item $\ttype(V):= \size(V)$, for a matrix variable $V\in\Mnam$;
\item $\ttype(e^T):= (\beta,\alpha)$ if $\ttype(e)=(\alpha,\beta)$;
\item $\ttype(\ones(e)):= (\alpha,1)$ if $\ttype(e)=(\alpha,\beta)$;
\item $\ttype(\diag(e)):= (\alpha,\alpha)$, if $\ttype(e)=(\alpha,1)$;
\item $\ttype(e_1 \cdot e_2):= (\alpha,\gamma)$ if  $\ttype(e_1)=(\alpha,\beta)$, and $\ttype(e_2)=(\beta,\gamma)$;
\item $\ttype(e_1 + e_2):=(\alpha,\beta)$ if $\ttype(e_1)=\ttype(e_2)=(\alpha,\beta)$;
\item $\ttype(e_1\times e_2):=(\alpha,\beta)$ if $\ttype(e_1)=(1,1)$ and $\ttype(e_2)=(\alpha,\beta)$; and
\item $\ttype(f(e_1,\ldots ,e_k)):= (\alpha,\beta)$, whenever $\ttype(e_1) = \cdots = \ttype(e_k) := (\alpha,\beta)$ and $f\in\Fun_k$.
\end{itemize}
 We call an expression \textit{well-typed} according to the schema $\Sch$, if it has a defined type. 
A well-typed expression can be evaluated regardless of the actual sizes of the matrices assigned to matrix variables, as we describe next.

\smallskip
\noindent
\textbf{Semantics.}\, We use $\mtr{\RR}$ to denote the set of all real matrices and for 
$A\in\mtr{\RR}$, $\dim(A)\in\NN^2$ denotes its dimensions.
A (\lang) \textit{instance} $\I$ over a schema $\Sch$ is a pair $\I = (\dom,\conc)$, where $\dom : \DD \mapsto \NN$ assigns a value to each size symbol (and thus in turn  dimensions to each matrix variable), and $\conc : \Mnam \mapsto \mtr{\RR}$ assigns a concrete matrix to each matrix variable $V\in \Mnam$, such that $\dim(\conc(V)) = \dom(\alpha)\times \dom(\beta)$ if $\size(V) = (\alpha,\beta)$. That is, an instance tells us the dimensions of each matrix variable, and also the concrete matrices assigned to the variable names in $\Mnam$. We assume that $\dom(1) = 1$, for every instance $\I$. If $e$ is a well-typed expression according to $\Sch$, then we denote by $\sem{e}{\I}$ the matrix obtained by evaluating $e$ over $\I$, and define it as follows:
\begin{itemize}
\item $\sem{V}{\I} := \conc(V)$, for $V\in \Mnam$;
\item $\sem{e^T}{\I} := \sem{e}{\I}^T$, where $A^T$ is the transpose of matrix $A$;
\item $\sem{\ones(e)}{\I}$ is a $n\times 1$ vector with $1$ as all of its entries, where $\dim(\sem{e}{\I})=(n,m)$;
\item $\sem{\diag(e)}{\I}$ is a diagonal matrix with the vector $\sem{e}{\I}$ on its main diagonal, and zero in every other position;
\item $\sem{e_1\cdot e_2}{\I} := \sem{e_1}{\I} \cdot \sem{e_2}{\I}$;
\item $\sem{e_1+ e_2}{\I} := \sem{e_1}{\I} + \sem{e_2}{\I}$;
\item $\sem{e_1\times e_2}{\I} := a\times \sem{e_2}{\I}$ with $\sem{e_1}{\I}=[a]$; and
\item $\sem{f(e_1,\ldots ,e_k)}{\I}$ is a matrix $A$ of the same size as $\sem{e_1}{\I}$, and where $A_{ij}$ has the value $f(\sem{e_1}{\I}_{ij},\ldots ,\sem{e_k}{\I}_{ij})$.
\end{itemize}
\noindent
Although \lang\ forms a solid basis for a matrix query language, it is limited in expressive power. Indeed, \lang\ is subsumed by first order logic with aggregates that uses only three variables \cite{matlang-journal}. 
Hence,
 no \lang\ expression exists that can compute the transitive closure of a graph (represented by its adjacency matrix) or can compute the inverse of a matrix. Rather than extending \lang\ with specific linear algebra operators, such as matrix inversion, we  
next introduce a limited form of recursion in \lang.\looseness=-1

%% file: sections/forloops.tex
To extend \lang\ with recursion, we take inspiration from classical linear algebra algorithms, such as those described in \cite{num}. Many of these algorithms are based on \textit{for-loops} in which the termination condition for each loop is determined by the matrix dimensions. We have seen how the transitive closure of a matrix can be computed using for-loops in the Introduction. Here we add this ability to \lang, and show that the resulting language, called $\langfor,$ can compute properties outside of the scope of \lang. We see more advanced examples, such as Gaussian elimination and matrix inversion, later in the paper. 

\subsection{Syntax and semantics of \langfor} The syntax of \langfor is defined just as for \lang but with an extra rule in the grammar:
\medskip

\begin{tabular}{lcll}
 $\ffor{v}{X}{e}$ & (canonical for loop, with $v, X \in \Mvar$). 
\end{tabular}

\medskip
\noindent Intuitively, $X$ is a matrix variable which is iteratively updated according to the expression $e$. We simulate iterations of the form ``\texttt{for} $i\in [1..n]$'' by letting $v$ loop over the \textit{canonical vectors} $b_1^n,\ldots,b_n^n$ of dimension $n$. Here,
%
$b_1^n = [1\ 0 \cdots 0]^T$, $b_2^n = [0\ 1\ 0 \cdots 0]^T$, etc. When $n$ is clear from the context we simply write $b_1,b_2,\ldots$. In addition, the expression $e$ in the rule above may depend on $v$. 

We next make the semantics precise and start by
declaring the type of loop expressions.
Given a schema $\Sch$, the type of a \langfor expression $e$, denoted $\ttype(e)$, is defined inductively as in \lang but with following extra rule:
\begin{itemize}
\item $\ttype(\ffor{v}{X}{e}) := (\alpha,\beta)$, if \\
$\ttype(e)=\ttype(X) =(\alpha,\beta)$ and $\ttype(v) = (\gamma,1)$.
\end{itemize}
We note that $\Sch$ now necessarily includes $v$ and $X$ as variables and assigns size symbols to them.
We also remark that in the definition of the type of $\ffor{v}{X}{e}$, we require that $\ttype(X) = \ttype(e)$ as this expression updates the content of the variable $X$ in each iteration using the result of $e$. We further restrict the type of 
$v$ to be a vector, i.e., $\ttype(v)=(\gamma,1)$, since $v$ will be instantiated with canonical vectors.
%
A \langfor\ expression $e$ is well-typed over a schema $\Sch$ if its type is defined. 

For well-typed expressions we next define their semantics. This is done in an inductive way, just as for \lang. To define the semantics of $\ffor{v}{X}{e}$ over an instance $\I$, we need the following notation. Let $\I$ be an instance and $V\in \Mnam$. Then $\I[V := A]$ denotes an instance that coincides with $\I$, except that the value of the matrix variable $V$ is given by the matrix $A$. Assume that
$\ttype(v)= (\gamma,1)$, and $\ttype(e) = (\alpha,\beta)$ and $n := \dom(\gamma)$. Then, $\sem{\ffor{v}{X}{e}}{\I}$ is defined iteratively, as follows:
\begin{itemize}
\item Let $A_0 := \mathbf{0}$ be the zero matrix of size $\dom(\alpha)\times \dom(\beta)$.
\item For $i=1,\ldots n$, compute $A_i:= \sem{e}{\I[v := b^{n}_i, X:= A_{i-1}]}$.
\item Finally, set $\sem{\ffor{v}{X}{e}}{\I}:= A_{n}$.
\end{itemize}

For better understanding how \langfor  works, we next provide some  examples.
We start by showing that the one-vector and $\diag$ operators are redundant
in \langfor.
\begin{example}\label{ex:onevec}
We first show how the one-vector operator $\ones(e)$ can be expressed using \texttt{for} loops.
It suffices to consider the expression
$$e_{\ones}:=\ffor{v}{X}{X+v},$$
with $\ttype(v)=(\alpha,1)=\ttype(X)$ if $\ttype(e)=(\alpha,\beta)$. This expression is well-typed
and is of type $(\alpha,1)$. When evaluated over some instance $\I$ with $n=\dom(\alpha)$, $\sem{e_{\ones}}{\I}$ is defined as follows.
Initially, $A_0:=\mathbf{0}$. Then $A_i:=A_{i-1}+b_i^n$, i.e., the $i$th canonical vector is added to $A_{i-1}$.
Finally, $\sem{e_{\ones}}{\I}:=A_n$ and this now clearly coincides with $\sem{\ones(e)}{\I}$.\qed
\end{example}

%
\begin{example}\label{ex:diag}
We next show that the $\diag$ operator is redundant in \langfor.
Indeed, it suffices to consider the expression
$$e_{\mathsf{diag}}:=
\ffor{v}{X}{X + (v^T\cdot e) \times v\cdot v^T},$$ where $e$ is a \langfor\  expression of type $(\alpha,1)$. For this expression to be well-typed, $v$ has to be a vector variable of type $\alpha\times 1$ and $X$ a matrix variable of type $(\alpha,\alpha)$. Then, $\sem{e_{\mathsf{diag}}}{\I}$ is defined as follows.
Initially, $A_0$ is the zero matrix of dimension $n\times n$, where $n=\dom(\alpha)$. Then, in each iteration
$i\in[1..n]$, $A_{i}:=A_{i-1}+  ((b_i^n)^T\cdot\sem{e}{\I})\times (b_i^n\cdot (b_i^n)^T)$. In other words, $A_i$ is obtained by adding the matrix with value $(\sem{e}{\I})_i$ on position $(i,i)$ to $A_{i-1}$. Hence, $\sem{e_{\mathsf{diag}}}{\I}:=A_n=\sem{\diag(e)}{\I}$.\qed
 \end{example}

These examples illustrate that we can limit \langfor to consist of the following ``core'' operators: transposition, matrix multiplication and addition, scalar multiplication, pointwise function application, and for-loops. More specific, \langfor is defined by the following simplified syntax:
$$
e ::= V \ \mid \ e^T \!\ \mid \ e_1 \cdot e_2 \ \mid \ e_1 + e_2 \ \mid \ e_1\times e_2  \ \mid \  f(e_1,\ldots ,e_k) \ \mid \ \ffor{v}{X}{e}
$$
Similarly as for \lang, we write $\langforf{\Fun}$ for some set $\Fun$ of functions when these are required for the task at hand.



As a final example, we show that we can compute whether a graph contains a 4-$\textsf{clique}$ using \langfor.
\begin{example}\label{ex:fourcliques}
To test for $4$-cliques it suffices to consider the following expression with for-loops nested four times:
\begin{tabbing}
\texttt{for\,}\=$u,\,X_1.\ \ X_1 \ + $\\
\> \texttt{for\,}\=$v,\,X_2.\ \ X_2 \ +$ \\
\>\>\texttt{for\,}\=$w,\,X_3.\ \ X_3 \ +$ \\
\>\>\>\texttt{for\,}\=$x,\,X_4.\ \ X_4 \ +$ \\
\>\>\>\>$u^T\cdot V\cdot v \cdot u^T\cdot V\cdot w\cdot u^T\cdot V\cdot x \cdot $\\
\>\>\>\>$v^T\cdot V\cdot w \cdot v^T\cdot V\cdot x\cdot w^T\cdot V\cdot x \cdot g(u,v,w,x)$
\end{tabbing}
with $g(u,v,w,x)=f(u,v)\cdot f(u,w)\cdot f(u,x)\cdot f(v,w)\cdot f(v,x)\cdot f(w,x)$ and
$f(u,v)=1-u^T\cdot v$. Note that $f(b_i^n,b_j^n)=1$ if $i\neq j$ and $f(b_i^n,b_j^n)=0$ otherwise.
Hence, $g(b_i^n,b_j^n,b_k^n,b_\ell^n)=1$ if and only if all $i,j,k,l$ are pairwise different.
When evaluating the expression on an instance $\I$ such that $V$ is assigned to the adjacency 
matrix of a graph, the expression above evaluates to a non-zero value if and only if the graph
contains a four-clique.\qed
\end{example}
%
%

Given that \lang can not check for 4-cliques \cite{matlang-journal}, we easily obtain the following.

\begin{proposition}
\label{cor-ml-fml}
For any collection of functions $\Fun$, 
$\langf{\Fun}$ is properly subsumed by $\langforf{\Fun}$.
\end{proposition}

%
%
%

\subsection{Design decisions behind \langfor}

\input{./sections/design}

%% file: sections/design.tex
\noindent\textbf{Loop Initialization.} As the reader may have observed, in the semantics of for-loops we 
always initialize $A_0$ to the zero matrix~$\mathbf{0}$ (of appropriate dimensions). It is often convenient
to start the iteration given some concrete matrix  originating from the result of evaluation a \langfor\ expression $e_0$. To make this explicit, we write $\initf{e_0}{v}{X}{e}$ and its semantics is defined as above
with the difference that $A_0:=\sem{e_0}{\I}$. We observe, however, that $\initf{e_0}{v}{X}{e}$ can already
be expressed in \langfor. In other words, we do not loose generality by assuming an initialization of $A_0$ by $\mathbf{0}$.
The key insight is that in \langfor\ we can check during evaluation whether or not
the current canonical vector $b_i^n$ is equal to the $b_1^n$. This 
is due to the fact that for-loops iterate over the canonical vectors in a fixed order. We discuss this more in the next paragraph.
In particular, we can define a \langfor expression $\mmin$, which when evaluated on an instance, returns $1$ if its input vector is $b_1^n$, and returns $0$ otherwise. Given $\mmin$, consider now the
\langfor\ expression
 $$\ffor{v}{X}{\mmin{v}\cdot e(v,X/e_0) + (1-\mmin{v})\cdot e(v,X)},$$
 where we explicitly list $v$ and $X$ as matrix variables on which $e$ potentially depends on, and where
 $e(v,X/e_0)$ denotes the expression obtained by replacing every occurrence of $X$ in $e$ with $e_0$.
%
When evaluating this expression on an instance $\I$, $A_0$ is initial set to the zero matrix, in the first iteration (when  $v=b_1^n$ and thus $\mmin{v}=1$)
we have $A_1=\sem{e}{\I[v:=b_1^n,X:=\sem{e_0}{\I}]}$, and for consecutive iterations (when only the part related to $1-\mmin{v}$ applies) $A_i$ is updated as before. Clearly, the result of this evaluation is equal to
$\sem{\initf{e_0}{v}{X}{e}}{\I}$.

As an illustration, we consider the Floyd-Warshall algorithm given in the Introduction. 


\begin{example}\label{ex:floyd}
Consider the following expression:
\begin{tabbing}
$e_{FW} := $ \texttt{for\,}\=$v_k,\, X_1\!=\!A.\ \ X_1 \ + $\\
\> \texttt{for\,}\=$v_i, \, X_2.\ \ X_2 \ +$ \\
\>\>\texttt{for\,}\=$v_j,\, X_3.\ \ X_3 \ +$ \\
\>\>\>$(v_i^T\cdot X_1\cdot v_k \cdot v_k^T\cdot X_1\cdot v_j)\times v_i\cdot v_j^T$
\end{tabbing}
The expression $e_{FW}$ simulates the Floyd-Warshall algorithm by updating the matrix $A$, which is stored in the variable $X_1$. The inner sub-expression here constructs an $n\times n$ matrix that contains one in the position $(i,j)$ if and only if one can reach the vertex $j$ from $i$ by going through $k$, and zero elsewhere. If an instance $\I$ assigns to $A$ the adjacency matrix of a graph, then $\sem{e_{FW}}{\I}$ will be equal to the matrix produced by the algorithm given in the Introduction.
\qed
\end{example}

\noindent\textbf{Order.} By introducing for-loops we not only extend \lang\ with bounded recursion, we also introduce order information. Indeed, the semantics of the \texttt{for} operator assumes that the canonical vectors $b_1,b_2,\ldots$
are accessed in this order. It implies, among other things, that \langfor\ expressions are not permutation-invariant.
We can, for example, return the bottom right-most entry in a matrix. Indeed, consider the expression $e_{\mathsf{max}} := \ffor{v}{X}{v}$ which, for it to be well-typed, requires both $v$ and $X$ to be of type $(\alpha,1)$. Then, $\sem{e_{\mathsf{max}}}{\I}=b_n^n$, for $n=\dom(\alpha)$, simply because initially, $X=\mathbf{0}$, but $X$ will be overwritten by $b_1^n,b_2^n,\ldots,b_n^n$, in this order. Hence, at the end of the evaluation $b_n^n$ is returned.
To extract the bottom right-most entry from a matrix, we now simply use $e_{\mathsf{max}}^T\cdot V\cdot e_{\mathsf{max}}$.

Although the order is implicit in \langfor, we can explicitly use this order in \langfor expressions. More precisely, the order on canonical vectors is made accessible by
using the matrix:
\[
S_{\leq} = \begin{bmatrix}
1 & 1 & \cdots &  1 \\
0 & 1 & \cdots & 1\\
\vdots & \vdots & \ddots & 1 \\
0 & 0 & \cdots & 1 
\end{bmatrix}.
\] 
We observe that $S_{\leq}$ has the property that $b_i^T\cdot S_{\leq} \cdot b_j=1$, for two canonical vectors $b_i$ and $b_j$ of the same dimension, if and only if $i\leq j$. Otherwise, $b_i^T\cdot S_{\leq} \cdot b_j=0$. 
Interestingly, we can build the matrix $S_{\leq}$ with the following \langfor expression:
$$
e_{\leq}=\ffor{v}{X}{X + \left((X\cdot e_{\mathsf{max}}) + v \right)\cdot v^T + v\cdot e^T_{\mathsf{max}}},
$$
where $e_{\mathsf{max}}$ is as defined above. The intuition behind this expression is that by using the last canonical vector $b_n$, as returned by $e_{\mathsf{max}}$, we have access to the last column of $X$ (via the product $X\cdot e_{\mathsf{max}}$). We use this column such that after the $i$-th iteration, this column contains the $i$-th column of $S_{\leq}$. This is done by incrementing $X$ with $v\cdot e_{\mathsf{max}}^T$.
To construct $S_{\leq}$, in the $i$-th iteration we further increment $X$ with 
(i)~the current last column in $X$ (via $X\cdot e_{\mathsf{max}}\cdot v^T$) which holds
the $(i-1)$-th column of $S_{\leq}$; and (ii)~the current canonical vector (via $v\cdot v^T$). Hence, after iteration $i$, $X$ contains the first $i$ columns of $S_{\leq}$ and holds the $i$th column of $S_{\leq}$ in its last column. It is now readily verified that $X=S_{\leq}$ after the $n$th iteration.

It should be clear that if we can compute $S_{\leq}$ using $e_{\leq}$, then we can easily define the following predicates and vectors related with the order of canonical vectors:
\begin{itemize}
	\item $\mathsf{succ}(u,v)$ such that $\mathsf{succ}(b_i^n,b_j^n)=1$ if $i\leq j$ and $0$ otherwise. Similarly, we can define
	$\mathsf{succ}^+(u,v)$ such that  $\mathsf{succ}^+(b_i^n,b_j^n)=1$ if $i < j$ and $0$ otherwise;
	\item $\mathsf{min}(u)$ such that  $\mathsf{min}(b_i^n)=1$ if $i=1$ and $\mathsf{min}(b_i^n)=0$ otherwise; 
	\item $\mathsf{max}(u)$ such that  $\mathsf{max}(b_i^n)=1$ if $i=n$ and $\mathsf{min}(b_i^n)=0$ otherwise; and
	\item $e_{\mathsf{min}}$ and $e_{\mathsf{max}}$ such that $\sem{e_{\mathsf{min}}}{\I}=b_1^n$ and 
	$\sem{e_{\mathsf{max}}}{\I}=b_n^n$, respectively.
\end{itemize}
The definitions of these expressions 
are detailed in the appendix.

Having order information available results in \langfor to be quite expressive. We heavily rely on order information in the next sections to compute the inverse of matrices and more generally to simulate low complexity Turing machines and arithmetic circuits.

%% file: sections/algorithms.tex
One of our main motivations to introduce for-loops is to be able
to express classical linear algebra algorithms in a natural way. We have seen that \langfor is
quite expressive as it can check for cliques, compute the transitive closure, and can even
leverage a successor relation on canonical vectors. The big question is how expressive \langfor
actually is. We will answer this in the next section by connecting \langfor with 
arithmetic circuits of polynomial degree. Through this connection, one can move back and forth between \langfor and arithmetic circuits, and as a consequence, anything computable by such a circuit can be
computed by \langfor as well. When it comes to specific linear algebra algorithms, the detour via circuits
can often be avoided. Indeed, in this section we illustrate that \langfor is able to
compute LU decompositions of matrices. These decompositions form the basis of many other algorithms, such as solving linear systems of equations. We further show that \langfor is expressive enough to compute matrix inversion and the determinant. We recall that matrix inversion and determinant need to be explicitly added as separate operators in \lang~\cite{matlang-journal} and that the LARA language is unable to invert matrices under usual complexity-theoretic assumptions~\cite{BarceloH0S20}.

\subsection{LU decomposition}
%
%
A lower-upper (LU) decomposition factors a matrix $A$ as the product of a lower triangular matrix $L$ and upper triangular matrix $U$.  
This decomposition, and more generally LU decomposition with row pivoting (PLU),  underlies many linear algebra algorithms and 
we next show that \langfor can compute these decompositions.

\smallskip
\noindent
\textbf{LU decomposition by Gaussian elimination.} LU decomposition can be seen as a matrix form of Gaussian elimination in which the columns of $A$
are reduced, one by one, to obtain the matrix $U$. The reduction of columns of $A$ is achieved
as follows. Consider the first column $[A_{11},\ldots,A_{n1}]^T$ of $A$ and  define 
$c_1 := [0, \alpha_{21},\ldots, \alpha_{n1}]^T$ 
with $\alpha_{j1} := -\frac{A_{j1}}{A_{11}}$. Let $T_1:=I+ c_1\cdot b_1^T$ and consider
$T_1\cdot A$. That is, the $j$th row of $T_1\cdot A$ is obtained by multiplying the first row of $A$ by $\alpha_{j1}$ and adding it to the $j$th row of $A$. As a result, the first column of $T_1\cdot A$ is equal to $[A_{11},0,\ldots,0]^T$, i.e., 
all of its entries below the diagonal are zero.  One then iteratively performs a similar computation, using a matrix $T_i:=I+c_i\cdot b_i^T$, where $c_i$ now depends on the $i$th column in $T_{i-1}\cdots T_1\cdot A$. As a consequence, $T_i\cdot T_{i-1}\cdots T_1\cdot A$ is upper triangular
in its first $i$ columns. At the end of this process, $T_n\cdots T_1\cdot A=U$ where $U$ is the desired upper triangular matrix.
Furthermore, it is easily verified that each $T_i$ is invertible and by defining $L:=T_1^{-1}\cdot\cdots\cdot T_n^{-1}$ one obtains a lower triangular matrix satisfying $A=L\cdot U$. The above procedure is only successful when the denominators used in the definition of the vectors $c_i$ are non-zero. When this is the case we call a matrix $A$ \textit{LU-factorizable}. 

In case when such a denominator is zero in one of the reduction steps, one can remedy this situation by \textit{row pivoting}. That is, when the $i$th entry of the
$i$th row in $T_{i-1}\cdots T_1\cdot A$ is zero, one replaces the $i$th row by  $j$th row in this matrix, with $j>i$, provided that $i$the entry of the $j$th row is non-zero. If no such row exists, this implies that all elements below the diagonal are zero already in column $i$ and one can proceed with the next column. One can formulate this in matrix terms by stating that there exists a permutation matrix $P$, which pivots rows, such that $P\cdot A=L\cdot U$. Any matrix $A$ is LU-factorizable \textit{with pivoting}.

\smallskip
\noindent
\textbf{Implementing LU decomposition in \langfor.} 
We first assume that the input matrices are LU-factorizable. We deal with general matrices later on.
To implement the above procedure, we need to compute the vector $c_i$ for each column $i$. We do this in two steps. First, we extract from our input matrix its $i$th column and set all its upper diagonal entries to zero
by means of 
 the 
 expression:
$$\ccol{V}{y} := \ffor{v}{X}{\mathsf{succ}^+(y,v)\cdot(v^T\cdot V \cdot y)\cdot v + X}.$$
Indeed, when $V$ is assigned to a matrix $A$ and $y$ to $b_i$, we have that $X$ will be initially assigned
$A_0=\mathbf{0}$ and in consecutive iterations,  $A_j=A_{j-1}+ b_j^T\cdot A\cdot b_i$ if $j>i$ (because $\mathsf{succ}^+(b_i,b_j)=1$ if $j>i$) and $A_j=A_{j-1}$ otherwise (because $\mathsf{succ}^+(b_i,b_j)=0$ for $j\leq i$). 
The result of this evaluation is the desired column vector.
%
%
Using $\ccol{V}{y}$, we can now compute $T_i$ by the following expression:
$$\red{V}{y} := e_{\mathsf{Id}}+ f_/(\ccol{V}{y},-(y^T\cdot V\cdot y)\cdot \ones(y))\cdot y^T,$$
where $f_/:\RR^2\to\RR:(x,y)\mapsto x/y$ is the division function. 
When $V$ is assigned to $A$ and $y$ to $b_i$, $f_/(\ccol{A}{b_i},-(b_i^T\cdot A\cdot b_i)\cdot \ones(b_i))$ is equal to the vector $c_i$ used in the definition of $T_i$. To perform the reduction steps for all columns, we consider
the expression:
$$
e_{U}(V) :=  \left( \initf{e_{\mathsf{Id}}}{y}{X}{\red{X\cdot V}{y}\cdot X} \right) \cdot V.
$$
That is, when $V$ is assigned $A$, $X$ will be initially $A_0=I$, and then
$A_i=\red{A_{i-1}\cdot A}{b_i}=T_i\cdot T_{i-1}\cdots T_1\cdot A$, as desired.
We show in the appendix that, because we can obtain the matrices $T_i$ in \langfor and that these
are easily invertible, we can also construct an expression $e_L(V)$ which evaluates to $L$ when $V$ is assigned to
$A$. We may thus conclude the following.
\begin{proposition}\label{prop:gauss}
There exists $\langforf{f_/}$ expressions $e_L(V)$ and $e_U(V)$ such that
$\sem{e_L}{\I}=L$ and $\sem{e_U}{\I}=U$ form an LU-decomposition of $A$,
where $\conc(V)=A$ and $A$ is LU-factorizable.\qed
\end{proposition}
We remark that the proposition holds when division is added as a function in $\Fun$
in \langfor. When row pivoting is needed, we can also obtain a permutation matrix
$P$ such that $P\cdot A=L\cdot U$ holds by means of an expression in \langfor, provided
that we additionally allow the function $f_{>0}$, 
where $f_{>0}:\RR\to\RR$ is such that $f_{>0}(x):=1$ if $x>0$ and $f_{>0}(x):=0$ otherwise.

 \begin{proposition}\label{prop:palu}
There exist expressions $e_{L^{-1}P}(M)$ and $e_U(M)$ in $\langforf{f_/,f_{>0}}$  such that
$L^{-1}\cdot P=\sem{e_{L^{-1}P}}{\I}$ and $U=\sem{e_U}{\I}$, satisfy $L^{-1}\cdot P\cdot A=U$.\qed
 %
\end{proposition}
Intuitively, by allowing $f_{>0}$ we introduce a limited form of \texttt{if-then-else} in \langfor, which is needed
to continue reducing columns only when the right pivot has been found.

\subsection{Determinant and inverse}
Other key linear algebra operations include the computation of the determinant and
the inverse of a matrix (if the matrix is invertible). As a consequence of the expressibility
in $\langforf{f_/,f_{>0}}$ of LU-decompositions with pivoting, it can be shown that the determinant
and inverse can be expressed as well. 

However, the results
in the next section (connecting \langfor with arithmetic circuits) imply that the determinant
and inverse of a matrix can already be defined in $\langforf{f_/}$. So instead of using LU decomposition with pivoting for matrix inversion and computing the determinant, we provide an alternative solution.

More specifically, we rely on Csanky's algorithm for computing the inverse of a matrix~\cite{Csanky76}. This algorithm uses the characteristic
polynomial $p_A(x)=\mathsf{det}(xI-A)$ of a matrix. When expanded as a polynomial
$p_A(x)=\sum_{i=0}^{n} c_i x^i$ and it is known that $A^{-1}=\frac{-1}{\phantom{-1}c_n}\sum_{i=0}^{n-1}c_i A^{n-1-i}$
if $c_n\neq 0$. Furthermore, $c_0=1$, $c_n=(-1)^n\mathsf{det}(A)$ and the coefficients $c_i$ of $p_A(x)$
are known to satisfy the system of equations $S\cdot c=s$ given by:
$$
\left(\begin{matrix}
1 & 0 & 0 & \cdots & 0 & 0\\
S_1 & 2 & 0 & \cdots  &0 & 0\\
S_2 & S_1 & 3 & \cdots  &0 & 0\\
\vdots & \vdots & \vdots & \vdots & \vdots & 0\\
S_{n-1} & S_{n-2} & S_{n-3} & \cdots & S_1 & n\\
\end{matrix}\right)\cdot
\left(\begin{matrix}
c_1\\
c_2\\
c_3\\
\vdots\\
c_n\\
\end{matrix}\right)=\left(\begin{matrix}
S_1\\
S_2\\
S_3\\
\vdots\\
S_n\\
\end{matrix}\right),
$$
with $S_i=\mathsf{tr}(A^i)$. We show, in the appendix, that we can construct all ingredients of this system of equations in $\langforf{f_/}$. By observing that the matrix $S$ is a lower triangular matrix with non-zero elements on its diagonal, we can write it in the form $D_S+(S-D_{S})=D_S\cdot(I+D_S^{-1}\cdot (S-D_S))$ with $D_S$ the diagonal matrix consisting of the diagonal entries of $S$.
Hence $S^{-1}=(I+D_{S}^{-1}\cdot(S-D_{S}))^{-1}\cdot D_S^{-1}$. 
We remark $D_S^{-1}$ can simply be obtained by inverting the (non-zero) elements on the diagonal by means of $f_/$ in $\langforf{f_/}$. Furthermore, we observe that $(I+D_S^{-1}(S-D_S))^{-1}=\sum_{i=0}^{n}(D_S^{-1}(S-D_S))^i$ which is something
we can compute in $\langforf{f_/}$ as well. Hence, we can invert $S$ and obtain the vector $(c_1,\ldots,c_n)^T$ as $S^{-1}\cdot s$. To compute $A^{-1}$ it now suffices to compute
$\frac{-1}{\phantom{-1}c_n}\sum_{i=0}^{n-1}c_i A^{n-1-i}$. To find the determinant,
we compute $(-1)^nc_n$. All this can be done in $\langforf{f_/}$.
We may thus conclude:
\begin{proposition}\label{prop:inverse}
There are $\langforf{f_/}$ expressions $e_{\mathsf{det}}(V)$ and $e_{\mathsf{inv}}(V)$ such that
$\sem{e_{\mathsf{det}}}{\I}=\mathsf{det}(A)$, and  
$\sem{e_{\mathsf{inv}}}{\I}=A^{-1}$ when $\I$ assigns $V$
to $A$ and $A$ is invertible.
\qed
\end{proposition}

%% file: sections/expr.tex
In this section we explore the expressive power of $\langfor.$ Given that arithmetic circuits \cite{allender} capture most standard linear algebra algorithms \cite{Raz02,ShpilkaY10}, they seem as a natural candidate for comparison. Intuitively, an arithmetic circuit is similar to a boolean circuit \cite{aroraB2009}, except that it has gates computing the sum and the product function, and processes elements of $\RR$ instead of boolean values. To connect \langfor to arithmetic circuits we need a notion of uniformity of such circuits. After all, a \langfor expression can take matrices of arbitrary dimensions as input and we want to avoid having  different circuits for each dimension. To handle inputs of different sizes, we thus consider a notion of uniform families of arithmetic circuits, defined via a Turing machine generating a description of the circuit for each input size $n$.

What we show in the remainder of this section is that any function $f$ which operates on matrices, and is computed by a uniform family of arithmetic circuits of bounded degree, can also be computed by a \langfor expression, and vice versa. In order to keep the notation light, we will focus on 
 \langfor schemas over ``square matrices'' where each variable has type $(\alpha,\alpha),(\alpha,1),(1,\alpha)$, or $(1,1)$, although all of our results hold without these restrictions as well. In what follows, we will write $\langfor$ to denote $\langforf{\emptyset}$, i.e. the fragment of our language with no additional pointwise functions. We begin by defining circuits and then show how circuit families can be simulated by $\langfor.$

\subsection{From arithmetic circuits to \langfor}
Let us first recall the definition of arithmetic circuits. 
An \textit{arithmetic circuit} $\Phi$ over a set $X=\{x_1,\ldots,x_n\}$ of input variables is a directed
acyclic labeled graph. The vertices of $\Phi$ are called \textit{gates} and denoted by $g_1,\ldots,g_m$;
the edges in $\Phi$ are called \textit{wires}. The children of a gate $g$ correspond to all gates
$g'$ such that $(g,g')$ is an edge. The parents of $g$ correspond to all gates $g'$ 
such that $(g',g)$ is an edge. The \textit{in-degree}, or a \textit{fan-in}, of a gate $g$ refers to its number of children, and 
the \textit{out-degree} to its number of parents. We will not assume any restriction on the in-degree of a gate, and will thus consider circuits with unbounded fan-in. Gates with in-degree $0$ are called \textit{input gates}
and are labeled by either a variable in $X$ or a constant $0$ or $1$. All other gates
are labeled by either $+$ or $\times$, and are referred to as \textit{sum gates} or \textit{product gates}, respectively.
Gates with out-degree $0$ are called \textit{output gates}. When talking about arithmetic circuits, one usually focuses on circuits with $n$ input gates and a single output gate.\looseness=-1

The \textit{size} of $\Phi$, denoted by $|\Phi|$, is its number of gates and wires. The \textit{depth} of $\Phi$, denoted
by $\mathsf{depth}(\Phi)$, is the length of the longest directed path from any of its output gates to any of the input gates. The \textit{degree} of a gate is defined inductively: an input gate has degree~1, a sum gate has a degree equal to the maximum of degrees of its children, and a product gate has a degree equal to the sum of the degrees of its children. When $\Phi$ has a single output gate, the \textit{degree} of $\Phi$, denoted by $\mathsf{degree}(\Phi)$, is defined as the degree of its output gate. If $\Phi$ has a single output gate and its input gates take values from $\RR$, then $\Phi$ corresponds to a polynomial in $\RR[X]$ in a natural way. In this case, the {degree} of $\Phi$ equals the degree of the polynomial corresponding to $\Phi$.
%
If $a_1,\ldots ,a_n$ are values in $\RR$, then 
the result of the circuit on this input is the value computed by the corresponding polynomial, denoted by $\Phi(a_1,\ldots ,a_k)$.

In order to handle inputs of different sizes, we use the notion of uniform circuit families. An \textit{arithmetic circuit family} is a set of arithmetic circuits $\{\Phi_n\mid n=1,2,\ldots\}$ where $\Phi_n$ has $n$ input variables and a single output gate. An arithmetic circuit family is \textit{uniform} if there exists a \logspace-Turing machine,
which on input $1^n$, returns an encoding of the arithmetic circuit $\Phi_n$ for each $n$.
We observe that uniform arithmetic circuit families are necessarily of polynomial size. 
Another important parameter is the circuit depth. A circuit family is of logarithmic depth, whenever $\mathsf{depth}(\Phi_n)\in \mathcal{O}(log\, n)$. We  now show that \langfor subsumes uniform arithmetic circuit families that are of logarithmic depth. 


\begin{theorem}
\label{th-circuits-ml}
For any uniform arithmetic circuit family $\{\Phi_n\mid n=1,2,\ldots\}$ of logarithmic depth there is a \langfor schema $\Sch$ and an expression $e_\Phi$ using a matrix variable $v$, with $\ttype(v)=(\alpha,1)$ and $\ttype(e) = (1,1)$, such that for any input values $a_1,\ldots ,a_n$: 
\begin{itemize}
\item If $\I = (\dom,\conc)$ is a \lang\ instance such that $\dom(\alpha) = n$ and $\conc(v) = [a_1 \ldots a_n]^T$.
\item Then $\sem{e_\Phi}{\I} = \Phi_n(a_1,\ldots ,a_n)$.
\end{itemize}
\end{theorem}
It is important to note that the expression $e_\Phi$ does not change depending on the input size, meaning that it is uniform in the same sense as the circuit family being generated by a single Turing machine. The different input sizes for a \langfor instance are handled by the typing mechanism of the language.

\domagoj{New proof sketch.}
\textit{Proof sketch.} The proof of this Theorem, which is the deepest technical result of the paper, depends crucially on two facts: (i) that any polynomial time Turing machine working within linear space and producing linear size output, can be simulated via a \langfor\ expression; and (ii) that evaluating an arithmetic circuit $\Phi_n$ can be done using two stacks of  depth $n$.

Evaluating  $\Phi_n$ on input $(a_1,\ldots ,a_n)$ can be done in a depth-first manner by maintaining  two stacks: the gates-stack that tracks the current gate being evaluated, and the values-stack that stores the value that is being computed for this gate. The idea behind having two stacks is that whenever the number of items on the gates-stack is higher by one than the number of items on the values-stack, we know that we are processing a fresh gate, and we have to initialize its current value (to 0 if it is a sum gate, and to 1 if it is a product gate), and push it to the values-stack. We then proceed by processing the children of the head of the gates-stack one by one, and aggregate the results using sum if we are working with a sum gate, and by using product otherwise. 

In order  to access the information about the gate we are processing (such as whether it is a sum or a product gate, the list of its children, etc.) we use the uniformity of our circuit family. Namely, we know that we can generate the circuit $\Phi_n$ with a \logspace-Turing machine $M_\Phi$ by running it on the input $1^n$. Using this machine, we can in fact compute all the information needed to run the two-stack algorithms described above. For instance, we can construct a \logspace\ machine that checks, given two gates $g_1$ and $g_2$, whether $g_2$ is a child of $g_1$. Similarly, we can construct a machine that, given $g_1$ and $g_2$ tells us whether $g_2$ is the final child of $g_1$, or the one that produces the following child of $g_1$ (according to the ordering given by the machine $M_\Phi$). Defining these machines based of $M_\Phi$ is similar to the algorithm for the composition of two \logspace\ transducers, and is commonly used to evaluate arithmetic circuits \citep{allender}. 


To simulate the circuit evaluation algorithm that uses two stacks, in \langfor we can use a binary matrix of size $n\times n$, where $n$ is the number of inputs. The idea here is that  the gates-stack corresponds to the first $n-3$ columns of the matrix, with each gate being encoded as a binary number in positions $1,\ldots,n-3$ of a row. The remaining three columns are reserved for the values-stack, the number of elements on the gates stack, and the number of elements on the values stack, respectively. The number of elements is encoded as a canonical vector of size $n$. Here we crucially depend on the fact that the circuit is of logarithmic depth, and therefore the size of the two stacks is bounded by $n$ (apart from the portion before the asymptotic bound kicks-in, which can be hard-coded into the expression $e_\Phi$). Similarly, given that the circuits are of polynomial size, we can assume that gate ids can be encoded into $n-3$ bits.

This matrix is then updated in the same way as the two-stack algorithm. It processes gates one by one, and using the successor relation for canonical vectors determines whether we have more elements on the gates stack. In this case, a new value is added to the values stack ($0$ if the gate is a sum gate, and $1$ otherwise), and the process continues. Information about the next child, last child, or input value, are obtained using the expression which simulates the Turing machine generating this data about the circuit (the machines used never produce an output longer than their input). Given that the size of the circuit is polynomial, say $n^k$, we can initialize the matrix with the output gate only, and run the simulation of the two-stack algorithm for $n^k$ steps (by iterating $k$ times over size $n$ canonical vectors). After this, the value in position  $(1,n-2)$ (the top of the values stack) holds the final results. \qed

\smallskip
While Theorem \ref{th-circuits-ml} gives us an idea on how to simulate arithmetic circuits, it does not tell us which classes of functions over real numbers can be computed by \langfor expressions. In order to answer this question, we note that arithmetic circuits can be used to compute functions over real numbers. Formally, a circuit family $\{\Phi_n\mid n=1,2,\ldots\}$ computes a function $f:\bigcup_{n\geq 1} \mathbb{R}^n\mapsto\mathbb{R}$, if for any $a_1,\ldots a_n\in \mathbb{R}$ it holds that $\Phi_n(a_1,\ldots ,a_n) = f(a_1,\ldots ,a_n)$. To make the connection with \langfor\!, we need to look at circuit families of bounded degree. 

A circuit family $\{\Phi_n\mid n=1,2,\ldots\}$ is said to be of \textit{polynomial degree} if $\mathsf{degree}(\Phi_n)\in O(p(n))$, for some polynomial $p(n)$. Note that polynomial size circuit families are not necessarily of polynomial degree. An easy corollary of Theorem \ref{th-circuits-ml} tells us that all functions computed by uniform family of circuits of polynomial degree and logarithmic depth can be simulated using \langfor expressions. However, we can actually drop the restriction on circuit depth due to the result of Valiant et. al.~\cite{valiant1981fast} and Allender et. al. \cite{AllenderJMV98} which says that any function computed by a uniform circuit family of polynomial degree (and polynomial depth), can also be computed by a uniform circuit family of logarithmic depth. Using this fact, we can conclude the following:


\begin{corollary}
\label{cor-circ-ml}
For any function $f$ computed by a uniform family of arithmetic circuits of polynomial degree, there is an equivalent \langfor formula $e_f$.
\end{corollary}


Note that there is nothing special about circuits that have a single output, and both Theorem \ref{th-circuits-ml} and Corollary \ref{cor-circ-ml} also hold for functions  $f:\bigcup_{n\geq 1} \mathbb{R}^n\mapsto\mathbb{R}^{s(n)}$, where $s$ is a polynomial. Namely, in this case, we can assume that circuits for $f$ have multiple output gates, and that the depth reduction procedure of \cite{AllenderJMV98} is carried out for each output gate separately. Similarly, the construction underlying the proof of Theorem \ref{th-circuits-ml} can be performed for each output gate independently, and later composed into a single output vector.

\subsection{From \langfor to circuits}

Now that we know that arithmetic circuits can be simulated using \langfor expressions, it is natural to ask whether the same holds in the other direction. That is, we are asking whether for each \langfor expression $e$ over some schema $\Sch$ there is a uniform family of arithmetic circuits computing precisely the same result depending on the input size. 

In order to handle the fact that \langfor\ expressions can produce any matrix, and not just a single value, as their output, we need to consider circuits which have multiple output gates. Similarly, we need to encode matrix inputs of a \langfor\ expression in our circuits. We will write $\Phi(A_1,\ldots ,A_k)$, where $\Phi$ is an arithmetic circuit with multiple output gates, and each $A_i$ is a matrix of dimensions $\alpha_i\times \beta_i$, with $\alpha_i,\beta_i \in \{n,1\}$ to denote the input matrices for a circuit $\Phi$. We will also write $\texttt{type}(\Phi)=(\alpha,\beta)$, with $\alpha,\beta\in \{n,1\}$, to denote the size of the output matrix for $\Phi$. We call such circuits \textit{arithmetic circuits over matrices}. When $\{\Phi_n\mid n=1,2,\ldots\}$ is a uniform family of arithmetic circuits over matrices, we will assume that the Turing machine for generating $\Phi_n$ also gives us the information about how to access a position of each input matrix, and how to access the positions of the output matrix, as is usually done when handling matrices with arithmetic circuits \cite{Raz02}. The notion of degree is extended to be the sum of the degrees of all the output gates. With this  at hand, we can now show the following result.

\begin{theorem}
\label{th-ml-to-circuits}
Let $e$ be a \langfor expression over a schema $\Sch$, and let $V_1,\ldots ,V_k$ be the variables of $e$ such that $\ttype(V_i)\in \{(\alpha,\alpha), (\alpha,1), (1,\alpha), (1,1)\}$. Then there exists a uniform arithmetic circuit family over matrices $\Phi_n(A_1,\ldots ,A_k)$ such that:
\begin{itemize}
\item For any instance $\I = (\dom,\conc)$ such that $\dom(\alpha) = n$ and $\conc(V_i) = A_i$ it holds that:
\item $\sem{e}{\I} = \Phi_n(A_1,\ldots ,A_k)$.
\end{itemize}
\end{theorem}


It is not difficult to see that the proof of Theorem \ref{th-circuits-ml} can also be extended to support arithmetic circuits over matrices. In order to identify the class of functions computed by \langfor expressions, we need to impose one final restriction: than on the degree of an expression. Formally, the \textit{degree of \langfor expression $e$} over a schema $\Sch$, is the minimum of the degrees of any circuit family  $\{\Phi_n\mid n=1,2,\ldots\}$ that is equivalent to $e$. That is, the expression $e$ is of polynomial degree, whenever there is an equivalent circuit family for $e$ of a polynomial degree.  
For example, all \langfor expressions seen so far have polynomial degree.
With this definition, we can now identify the class of functions for which arithmetic circuits and \langfor formulas are equivalent. This is the main technical contribution of the paper. 

\begin{corollary}
\label{th-equivalence}
Let $f$ be a function with input matrices $A_1,\ldots ,A_k$ of dimensions $\alpha\times \beta$, with $\alpha,\beta \in \{n,1\}$. Then, $f$ is computed by a uniform circuit family over matrices of polynomial degree if and only if there is a \langfor expression of polynomial degree for $f$. 
\end{corollary}

Note that this result crucially depends on the fact that expressions in \langfor are of polynomial degree. Some \langfor expression are easily seen to produce results which are not polynomial. An example of such an expression is, for instance, $e_{\texttt{exp}} = \ffor{v}{X=A}{X\cdot X}$, over a schema $\Sch$ with $\ttype(v)= (\gamma,1)$, and $\ttype(X)=(1,1)$. Over an instance which assigns $n$ to $\gamma$ this expression computes the function $a^{2^n}$, for $A=[a]$. Therefore, a natural question to ask then is whether we can determine the degree of a \langfor expression. Unfortunately, as we show in the following proposition this question is in fact undecidable.

\begin{proposition}
\label{prop-undec}
Given a \langfor expression $e$ over a schema $\Sch$, it is undecidable to check whether $e$ is of polynomial degree.
\end{proposition}


Of course, one might wonder whether it is possible to define a syntactic subclass of \langfor expressions that are of polynomial degree and can still express many important linear algebra algorithms. We identify one such class in Section \ref{ss:sumML}, called \langsum, and in fact show that this class is powerful enough to capture relational algebra on (binary) $K$-relations.

\floris{What I am missing is what it is precisely that prevents us from doing any circuit in \langfor and what do these results actually imply. What cannot be done?}

\subsection{Supporting additional operators}

%
%
%

The equivalence of \langfor and arithmetic circuits we prove above assumes that circuits can only use the sum and product gates (note that even without the sum and the product function, \langfor\ can simulate these operations via matrix sum/product). However, both arithmetic circuits and expressions in $\langfor$ can be allowed to use a multitude of functions over $\RR$. The most natural addition to the set of functions is the division operator, which is crucially needed in many linear algebra algorithms, such as, for instance, Gaussian elimination, or $LU$ decomposition (recall Proposition \ref{prop:gauss}).
Interestingly, the equivalence in this case still holds, mainly due to a surprising result which shows that (almost all) divisions can in fact be removed for arithmetic circuits which allow sum, product, and division gates \cite{allender}.

More precisely, in \cite{strassen1973vermeidung,borodin1982fast,kaltofen1988greatest} it was shown that for any function of the form $f = g/h$, where $g$ and $h$ are relatively prime polynomials of degree $d$, if $f$ is computed by an arithmetic circuit of size $s$, then both $g$ and $h$ can be computed by a circuit whose size is polynomial in $s + d$. Given that we can postpone the division without affecting the final result, this, in essence, tells us that division can be eliminated (pushed to the top of the circuit), and we can work with sum-product circuits instead. The degree of a circuit for $f$, can then be defined as the maximum of degrees of circuits for $g$ and $h$. Given this fact, we can again use the depth reduction procedure of \cite{AllenderJMV98}, and extend Corollary 
\ref{th-equivalence} to circuits with division.
\begin{corollary}
\label{cor-division}
Let $f$ be a function taking as its input matrices $A_1,\ldots ,A_k$ of dimensions $\alpha\times \beta$, with $\alpha,\beta \in \{n,1\}$. Then, $f$ is computed by a uniform circuit family over matrices of polynomial degree that allows divisions, if and only if there is a $\langforf{f_/}$ expression of polynomial degree for $f$.
\end{corollary}

\floris{Again, should we link back to the previous section?}

An interesting line of future work here is to see which additional functions can be added to arithmetic circuits and \langfor formulas, in order to preserve their equivalence. Note that this will crucially depend on the fact that these functions have to allow the depth reduction of \cite{AllenderJMV98} in order to be supported.

\floris{The depth reduction thing is thus crucial for simulating by \langfor. This was mentioned in the proof, but as mentioned before, we may want to make more explicit what goes wrong otherwise, as to better provide an understanding of what can be done with matrices?}

\cristian{I don't follow this comment from Floris. However, we should say here that \langfor can always simulate arithmetic circuits with new operators if the depth of the family is poly-logarithmic. }

%% file: sections/restrictions.tex
We conclude the paper by zooming in on some special fragments of \langfor and in which matrices can take values from an arbitrary (commutative) semiring $K$. In particular, we first consider \langsum, in which iterations can only perform
additive updates, and show that it is equivalent in expressive power to the (positive)
relational algebra on $K$-relations. We then extend \langsum such that also updates involving pointwise-multiplication (Hadamard product) are allowed. The resulting fragment, \langprod, is shown to be equivalent in expressive power to weighted logics. Finally, we consider the fragment \langmprod in which updates involving sum and matrix multiplication, and possibly order information, is allowed. From the results in Section~\ref{sec:queries}, we infer that the latter fragment suffices to compute matrix inversion. An overview of the fragments and their relationships are depicted in Figure~\ref{thefigure}.

\subsection{Summation matlang and relational algebra}
\label{ss:sumML}
When defining $4$-cliques and in several other expressions we have seen so far, we 
only update $X$ by adding some matrix to it. This restricted form of for-loop proved useful throughout the paper, and we therefore introduce it as a special operator. That is, we define:
$$\Sigma v. e := \ffor{v}{X}{X + e}.$$
We define the subfragment of \langfor, called \langsum, to consist of the $\Sigma$ operator plus the ``core'' operators in \lang, namely, transposition, matrix multiplication and addition, scalar multiplication, and pointwise function applications.

One property of \langsum is that it only allows expressions of polynomial degree. Indeed, one can easily show that \langsum can only create matrix entries that are polynomial in the dimension $n$ of the expression. More precisely, we can show the following:
\begin{proposition}\label{prop:poly}
Every expression in \langsum is of polynomial degree.
\end{proposition}

Interestingly enough, this restricted version of for-loop already allows us to capture the \lang\ operators that are not present in the syntax of \langsum. More precisely, we see from Examples~\ref{ex:onevec} and~\ref{ex:diag} that the one-vector and $\diag$ operator are expressible in \langsum. Combined with the observation that the $4$-clique
expression of Example~\ref{ex:fourcliques} is in \langsum, the following result is immediate.
\begin{corollary}
\lang\ is strictly subsumed by \langsum.
\end{corollary}

%
%
%

What operations over matrices can be defined with \langsum that is beyond \lang? In~\cite{brijder2019matrices}, it was shown that \lang\ is strictly included in the (positive) relational algebra on $K$-relations, denoted by $\mathsf{RA}_{K}^+$~\cite{GreenKT07}.\footnote{The algebra used in~\cite{brijder2019matrices} differs slightly from the one given in~\cite{GreenKT07}. In this paper we work with the original algebra $\mathsf{RA}_{K}^+$ as defined in~\cite{GreenKT07}.} 
%
%
It thus seems natural to compare the expressive power of \langsum with $\mathsf{RA}_{K}^+$. 
The main result in this section is that \langsum\ and $\mathsf{RA}_{K}^+$ 
are equally expressive over binary schemas. 
To make this equivalence precise, we next give the 
definition of $\mathsf{RA}_{K}^+$~\cite{GreenKT07} and then show how to connect both formalisms.

Let $\ddom$ be a data domain and $\att$ a set of attributes. A relational signature is a finite subset of $\att$. A relational schema is a function $\cR$ on finite set of symbols $\fdom(\cR)$ such that $\cR(R)$ is a relation signature for each $R \in \fdom(\cR)$. To simplify the notation, from now on we write $R$ to denote both the symbol $R$ and the relational signature $\cR(R)$.
Furthermore, we write $R \in \cR$ to say that $R$ is a symbol of $\cR$. 
For $R \in \cR$, an $R$-tuple is a function $t: R \rightarrow \ddom$. We denote by $\tuples(R)$ the set of all $R$-tuples. Given $X \subseteq R$, we denote by $t[X]$ the restriction of $t$ to the set $X$.

A semiring $(K, \ksum, \kprod, \kzero, \kone)$ is an algebraic structure where $K$ is a non-empty set, $\ksum$ and $\kprod$ are binary operations over $K$, and $\kzero, \kone \in K$. Furthermore,  $\ksum$ and $\kprod$ are associative operations, $\kzero$ and $\kone$ are the identities of $\ksum$ and $\kprod$ respectively, $\ksum$ is a commutative operation, $\kprod$ distributes over $\ksum$, and $\kzero$ annihilates $K$ (i.e. $\kzero \kprod k = k \kprod \kzero = \kzero$). As usual, we assume that all semirings in this paper are commutative, namely, $\kprod$ is also commutative. We use $\bigksum_X$ or $\bigkprod_X$ for the $\ksum$- or $\kprod$-operation over all elements in $X$, respectively. Typical examples of semirings are the reals $(\RR, +, \times, 0,1)$, the natural numbers $(\NN, +, \times, 0,1)$, and the boolean semiring $(\{0,1\}, \vee, \wedge, 0, 1)$. 

Fix a semiring $(K, \ksum, \kprod, \kzero, \kone)$ and a relational schema $\cR$. A $K$-relation of $R \in \cR$ is a function $r: \tuples(R) \rightarrow K$ such that the support  $\supp(r) = \{t \in \tuples(R) \mid r(t) \neq \kzero\}$ is finite. 
A $K$-instance $\cJ$ of $\cR$ is a function that assigns relational signatures of $\cR$ to $K$-relations. Given $R \in \cR$, we denote by $R^\cJ$ the $K$-relation associated to $R$. Recall that $R^\cJ$ is a function and hence  $R^\cJ(t)$ is the value in $K$ assigned to $t$. 
Given a $K$-relation $r$ we denote by $\adom(r)$ the active domain of $r$ defined as $\adom(r) = \{t(a) \mid t \in \supp(r) \wedge a \in R\}$.
Then the active domain of an $K$-instance $\cJ$ of $\cR$ is defined as $\adom(\cJ) = \bigcup_{R \in \cR} \adom(R^\cJ)$. 

An $\mathsf{RA}_{K}^+$  expression $\arae$ over $\cR$ is given by the following syntax:
$$
\begin{array}{rcl}
\arae & := & R \ \mid \ \arae \cup \arae \ \mid \  \pi_X(\arae) \ \mid \  \sigma_X(\arae) \ \mid \ \rho_f(\arae) \ \mid \ \arae \bowtie \arae
\end{array}
$$
where $R \in \cR$, $X \subseteq \att$ is finite, and $f: X \rightarrow Y$ is a one to one mapping with $Y \subseteq \att$. One can extend the schema $\cR$ to any expression over $\cR$ recursively as follows: $\cR(R) = R$, $\cR(\arae \cup \arae') = \cR(\arae)$, $\cR(\pi_X(\arae)) = X$, $\cR(\sigma_X(\arae)) = \cR(\arae)$, $\cR(\rho_f(\arae)) = X$ where $f:X \rightarrow Y$, and $\cR(\arae \bowtie \arae') = \cR(\arae) \cup \cR(\arae')$ for every expressions $\arae$ and $\arae'$.
We further assume that any expression $\arae$ satisfies the following syntactic restrictions: $\cR(\arae') = \cR(\arae'')$ whenever $\arae = \arae' \cup \arae''$, $X \subseteq \cR(\arae')$ whenever $\arae = \pi_X(\arae')$ or $\arae = \sigma_X(\arae')$, and $Y = \cR(\arae')$ whenever $\arae = \rho_f(\arae')$ with $f: X \rightarrow Y$.

Given an $\mathsf{RA}_{K}^+$ expression $\arae$ and a $K$-instance $\cJ$ of $\cR$, we define the semantics $\ssem{\arae}{\cJ}$ as a $K$-relation of $\cR(\arae)$ as follows. For $X \subseteq \att$, let $\operatorname{Eq}_X(t) = \kone$ when $t(a) = t(b)$ for every $a, b \in X$, and $\operatorname{Eq}_X(t) = \kzero$ otherwise. For every tuple $t \in \cR(\arae)$:
$$
\begin{array}{ll}
\text{if $\arae = R$, then} & \ssem{\arae}{\cJ}(t) = R^\cJ(t) \\
\text{if $\arae = \arae_1 \cup \arae_2$, then} & \ssem{\arae}{\cJ}(t) = \ssem{\arae_1}{\cJ}(t) \ksum \ssem{\arae_2}{\cJ}(t)  \\
\text{if $\arae = \pi_X(\arae')$, then} & \ssem{\arae}{\cJ}(t) = \bigksum_{t': t'[X] = t} \ssem{\arae'}{\cJ}(t') \\
\text{if $\arae = \sigma_X(\arae')$, then} & \ssem{\arae}{\cJ}(t) = 
\ssem{\arae'}{\cJ}(t) \kprod \operatorname{Eq}_X(t)  \\
\text{if $\arae = \rho_f(\arae')$, then} & \ssem{\arae}{\cJ}(t) = 
\ssem{\arae'}{\cJ}(t \circ f)  
\\
\text{if $\arae = \arae_1 \bowtie \arae_2$, then} & \ssem{\arae}{\cJ}(t) =  \ssem{\arae_1}{\cJ}(t[Y]) \kprod  \ssem{\arae_2}{\cJ}(t[Z]),
\end{array}
$$
where $Y = \cR(\arae_1)$ and $Z = \cR(\arae_2)$. It is important to note that the $\bigksum$-operation in the semantics of $\pi_X(\arae')$ is well-defined given that the support of $\ssem{\arae'}{\cJ}$ is always finite. 

We are now ready for comparing \langsum with $\mathsf{RA}_{K}^+$ . First of all, we need to extend \langsum from $\RR$ to any semiring. Let $\mtr{K}$ denote the set of all $K$-matrices. 
Similarly as for \lang\ over $\RR$, given a \lang\ schema $\Sch$, a $K$-instance $\I$ over $\Sch$ is a pair $\I = (\dom,\conc)$, where $\dom : \DD \mapsto \NN$ assigns a value to each size symbol, and $\conc : \Mnam \mapsto \mtr{K}$ assigns a concrete $K$-matrix to each matrix variable. Then it is straightforward to extend the semantics of \lang, \langfor, and \langsum from $(\RR, +, \times, 0, 1)$ to $(K, \ksum, \kprod, \kzero, \kone)$ by switching $+$ with $\ksum$ and $\times$ with $\kprod$. 

The next step to compare \langsum with $\mathsf{RA}_{K}^+$  is to represent $K$-matrices as $K$-relations.
Let $\Sch=(\Mnam,\size)$ be a \lang\ schema. On the relational side
we have for each size symbol $\alpha\in\DD\setminus\{1\}$, attributes $\alpha$, $\row_\alpha$, and $\col_\alpha$ in $\att$. Furthermore, for each $V\in\Mnam$ and $\alpha \in \DD$ we denote
by $R_V$ and $R_\alpha$ its corresponding relation name, respectively. Then, given $\Sch$ we define the relational schema $\text{Rel}(\Sch)$ such that $\fdom(\text{Rel}(\Sch)) =  \{R_\alpha \mid \alpha\in\DD\} \cup \{R_V \mid V \in \Mnam\}$ where $\text{Rel}(\Sch)(R_\alpha) = \{\alpha\}$ and:
\[
\text{Rel}(\Sch)(R_V) = \begin{cases}
\lbrace\row_\alpha,\col_\beta \rbrace & \text{ if $ \size(V)=(\alpha,\beta)$} \\
\lbrace\row_\alpha \rbrace & \text{ if $ \size(V)=(\alpha,1)$} \\
\lbrace\col_\beta \rbrace  &
\text{ if $ \size(V)=(1,\beta)$} \\
\lbrace\rbrace & \text{ if $\size(V)=(1,1)$}.
\end{cases}
\]
Consider now a matrix instance $\I = (\dom,\conc)$ over $\Sch$.
Let $V\in\Mnam$ with $\size(V)=(\alpha,\beta)$ and let $\conc(V)$ be its corresponding $K$-matrix of dimension $\dom(\alpha)\times \dom(\beta)$.
To encode $\I$ as a $K$-instance in $\mathsf{RA}_{K}^+$, we use as data domain $\ddom = \mathbb{N} \setminus \{0\}$. Then we construct the $K$-instance $\text{Rel}(\I)$ such that for each $V\in\Mnam$ we define 
$R_V^{\text{Rel}(\I)}(t):=\conc(V)_{ij}$ whenever $t(\row_\alpha) = i \leq \dom(\alpha)$ and $t(\col_\beta) = j \leq \dom(\beta)$, and $\kzero$ otherwise. Furthermore, for each $\alpha \in \DD$ we define $R_\alpha^{\text{Rel}(\I)}(t):=\kone$ whenever $t(\alpha) \leq \dom(\alpha)$, and $\kzero$ otherwise. In other words, $R_\alpha$ and $R_\beta$ encodes the active domain of a matrix variable $V$ with $\size(V)=(\alpha,\beta)$. Given that the $\mathsf{RA}_{K}^+$ framework of \cite{GreenKT07} represents the ``absence'' of a tuple in the relation with $\kzero$, we need to separately encode the indexes in a matrix.
This is where $R_\alpha^{\text{Rel}(\I)}$ and $R_\beta^{\text{Rel}(\I)}$ are used for.
%
%
We are now ready to state the first connection between \langsum and $\mathsf{RA}_{K}^+$  by using the previous encoding. The proof of the proposition below is by induction on the structure of expressions.
\begin{proposition}\label{prop:sum_to_ara} 
	For each \langsum expression $e$ over schema $\Sch$ such that $\Sch(e)=(\alpha,\beta)$ with $\alpha\neq 1\neq\beta$, there exists an $\mathsf{RA}_{K}^+$  expression $\Phi(e)$ over relational schema $\text{Rel}(\Sch)$ such that $\text{Rel}(\Sch)(\Phi(e))=\{\row_\alpha,\row_\beta\}$ and 
	such that for any instance $\I$ over~$\Sch$,
	$$
	\sem{e}{\I}_{i,j}=\ssem{\Phi(e)}{\text{Rel}(\I)}(t)
	$$
	for tuple $t(\mathrm{row}_\alpha)=i$ and $t(\mathrm{col}_\beta)=j$. Similarly for when $e$ has schema $\Sch(e)=(\alpha,1)$, $\Sch(e)=(1,\beta)$ or $\Sch(e)=(1,1)$, then $\Phi(e)$ has schema $\text{Rel}(\Sch)(\Phi(e))=\{\mathrm{row}_\alpha\}$,
	$\text{Rel}(\Sch)(\Phi(e))=\{\mathrm{col}_\alpha\}$, or
	$\text{Rel}(\Sch)(\Phi(e))=\{\}$, respectively.
\end{proposition}

We now move to the other direction.
To translate $\mathsf{RA}_{K}^+$  into \langsum, we must restrict our comparison to $\mathsf{RA}_{K}^+$  over $K$-relations with at most two attributes. Given that linear algebra works over vector and matrices, it is reasonable to restrict to unary or binary relations as input. Note that this is only a restriction on the input relations and not on intermediate relations, namely, expressions can create relation signatures of arbitrary size from the binary input relations. 
Thus, from now we say that a relational schema $\cR$ is binary if $|R| \leq 2$ for every $R \in \cR$. We also make the assumption that there is an (arbitrary) order, denoted by $<$, on the attributes in $\att$. 
This is to identify which attributes correspond to rows and columns when moving to matrices. 
Then, given that relations will be  either unary or binary and there is an order on the attributes, we write $t = (v)$ or $t = (u,v)$ to denote a tuple over a unary or binary relation $R$, respectively, where $u$ and $v$ is the value of the first and second attribute with respect to $<$.

Consider a binary relational schema $\cR$. With each $R\in \cR$ we associate a matrix variable $V_R$ such that, if $R$ is a binary relational signature, then $V_R$ represents a (square) matrix, and, if not (i.e. $R$ is unary), then $V_R$ represents a vector. Formally, fix a symbol $\alpha \in \DD \setminus \{1\}$. Let $\text{Mat}(\cR)$ denote the \lang \ schema
$(\Mnam_\cR,\size_\cR)$ such that $\Mnam_\cR = \{ V_R \mid R \in \cR\}$ and $\size_\cR(V_R) = (\alpha, \alpha)$ whenever $|R| = 2$, and $\size_\cR(V_R) = (\alpha, 1)$ whenever $|R|=1$. 
Take now a $K$-instance $\cJ$ of $\cR$ and suppose that $\adom(\cJ) = \{d_1, \ldots, d_n\}$ is the active domain of $\cJ$ (the order over $\adom(\cJ)$ is arbitrary). Then we define the matrix instance $\text{Mat}(\cJ) = (\dom_\cJ,\conc_\cJ)$ such that $\dom_\cJ(\alpha) = n$, $\conc_\cJ(V_R)_{i,j} = R^{\cJ}((d_i, d_j))$ whenever $|R|=2$, and $\conc_\cJ(V_R)_{i} = R^{\cJ}((d_i))$ whenever $|R|=1$. 
Note that, although each $K$-relation can have a different active domain, we encode them as square matrices by considering the active domain of the $K$-instance. By again using an inductive proof on the structure of 
$\mathsf{RA}_{K}^+$ expressions, we obtain the following result.
\begin{proposition}\label{prop:ara_to_sum} 
	Let $\cR$ be a binary relational schema. For each $\mathsf{RA}_{K}^+$  expression $\arae$ over $\cR$  such that $|\cR(\arae)| = 2$, there exists a \langsum  expression $\Psi(\arae)$ over \lang \ schema $\text{Mat}(\cR)$ such that for any $K$-instance $\cJ$ with $\adom(\cJ) = \{d_1, \ldots, d_n\}$ over $\cR$,
	$$
	\ssem{\arae}{\cJ}((d_i, d_j))=\sem{\Psi(\arae)}{\text{Mat}(\cJ)}_{i,j}.
	$$
	Similarly for when $|\cR(\arae)| = 1$, or $|\cR(\arae)| = 0$ respectively.
\end{proposition} 

It is important to remark that the expression $\arae$ of the previous result can have intermediate expressions that are not necessary binary, given that the proposition only restricts that the input relation and the schema of $\arae$ must have arity at most two. We recall from~\cite{brijder2019matrices} that \lang\ corresponds to $\mathsf{RA}_{K}^+$ where intermediate expressions are at most ternary, and this underlies, e.g., the inability of \lang\ to check for $4$-cliques. In \langsum, we can deal with intermediate relations of arbitrary arity. In fact, each new attribute can be seen to correspond to an application of the $\Sigma$ operator. For example, in the $4$-clique expression, four $\Sigma$ operators are needed, in analogy to how
$4$-clique is expressed in $\mathsf{RA}_{K}^+$.

Given the previous two propositions we derive the following conclusion which is the first characterization of relational algebra with a (sub)-fragment of linear algebra.
\begin{corollary}
	\langsum and $\mathsf{RA}_{K}^+$  over binary relational schemas are equally expressive. 
\end{corollary}

As a direct consequence, we have that \langsum cannot compute matrix inversion. Indeed, using similar arguments as
in~\cite{matlang-journal}, i.e., by embedding $\mathsf{RA}_{K}^+$  in (infinitary) first-order logic with counting and by leveraging its locality, one can show that \langsum cannot compute the transitive closure of an adjacency matrix. By contrast, the transitive closure can be expressed by means of matrix inversion~\cite{matlang-journal}. We also note that the evaluation of the $\Sigma$ operator is
independent of the order in which the canonical vectors are considered. This is because $\oplus$ is commutative.
Hence, \langsum cannot express the order predicates mentioned in Section~\ref{sec:formatlang}.

\subsection{Hadamard product and weighted logics}\label{subsec:langprod}

\input{./sections/wl}

\subsection{Matrix multiplication as a quantifier}\label{subsec:langlinear}
In a similar way, we can consider a fragment in which sum and the usual product of matrices can be used
in for-loops. Formally, for an expression $e$ we define the operator:
$$
\sprod v.\,  e=\ffor {v}{X = I}{X\cdot e}.
$$
where $I$ is the identity matrix. We call \langmprod the subfragment of \langfor that consists of \langsum extended with $\sprod v$. It is readily verified that  $\qhadprod v$ can expressed in terms of $\sprod v$.
Furthermore, by contrast to the Hadamard product, matrix multiplication is a non-commutative operator. As a consequence, one can formulate expressions that are not invariant under the order in which the canonical vectors
are processed.

\begin{proposition}
	Every expression in \langprod can be defined in \langmprod. Moreover, there exists an expression that uses the $\sprod v$ quantifier that cannot be defined in \langprod.
\end{proposition}

What is interesting is that \langsum extended with $\sprod v$ suffices to compute the transitive closure,
 provided that we allow for the $f_{>0}$ function. Indeed, one can use the expression $e_{\mathsf{TC}}(V):=f_{>0}\bigl(\sprod v.\, (e_{\mathsf{Id}}+V)\bigr)$ for this purpose because
 $\sem{e_{\mathsf{TC}}}{\I}=f_{>0}\bigl((I+A)^n\bigr)$ when $\I$ assigns an $n\times n$ adjacency matrix $A$ to $V$, and non-zero entries in $(I+A)^n$ coincide with non-zero entries in  the transitive closure of $A$.
%
%
Furthermore, if we extend this fragment with access to the matrix $S_{<}$, defining the (strict) order on canonical vectors, then Csanky's matrix inversion algorithm becomes expressible (if $f_/$ is allowed).
 We leave the study of this fragment and, in particular, the relationship to full \langfor, for future work. 
 
Finally, in Figure~\ref{thefigure} we show a diagram of all the fragments of \langfor introduced in this section and their corresponding equivalent formalisms.
%
%
%

\begin{figure}
	
	\begin{tikzpicture}[->,>=stealth, semithick, auto, initial text= {}, initial distance= {3mm}, accepting distance= {4mm}, node distance=0.5cm, semithick]
	
	\node [rectangle, draw=black, fill=white, minimum height=4mm, minimum width=2cm, rounded corners] (ML) at (0, 0) {$\texttt{ML}$};
	
	\node [inner sep=0mm] (SML) at ($(ML) + (0,0.65)$) {$\texttt{sum}\text{-}\texttt{ML}\equiv \texttt{RA}^+_K$};
	
	\node [circle, radius=4mm,draw=black, fill=black, inner sep=0mm] (pCLIQUE) at ($(ML) + (1.4,0.1)$) {};
	\node [right of=pCLIQUE,inner sep=0mm, node distance=0.65cm] (CLIQUE) {$\textsc{4Clique}$};
	
	\begin{pgfonlayer}{background}
	\node (SMLc)[draw=black, inner sep=1.5mm, rounded corners,fit=(SML)(ML)(CLIQUE)] {};
	\end{pgfonlayer}

	\node [inner sep=0mm] (FOML) at ($(SML) + (0,0.7)$) {$\texttt{FO}\text{-}\texttt{ML} \equiv \texttt{WL}$};
	
	\node [circle, radius=4mm,draw=black, fill=black, inner sep=0mm] (pDP) at ($(pCLIQUE) + (1.7,0.2)$) {};
	\node [right of=pDP,inner sep=0.5mm, node distance=0.3cm] (DP) {$\textsc{DP}$};

	\begin{pgfonlayer}{background}
	\node (FOMLc) [draw=black, inner sep=1.5mm, rounded corners,fit=(SMLc)(ML)(FOML)(DP)] {};
	\end{pgfonlayer}

	\node [inner sep=0mm] (PML)  at ($(FOML) + (0,0.7)$) {$\texttt{prod}\text{-}\texttt{ML} + S_{<}$};
	
	\node [circle, radius=4mm,draw=black, fill=black, inner sep=0mm] (pINV) at ($(pDP) + (1,1)$) {};
	\node [right of=pINV,inner sep=0mm, node distance=0.35cm] (INV) {$\textsc{Inv}$};
	
	\node [circle, radius=4mm,draw=black, fill=black, inner sep=0mm] (pDET) at ($(pDP) + (1,0.2)$) {};
	\node [right of=pDET,inner sep=0mm, node distance=0.35cm] (DET) {$\textsc{Det}$};

	\begin{pgfonlayer}{background}
	\node (PMLc) [draw=black, inner sep=1.5mm, rounded corners,fit=(SMLc)(ML)(FOMLc)(PML)(INV)(DET)] {};
	\end{pgfonlayer}

	\node [inner sep=0mm] (forML) at ($(PML) + (1,0.7)$) {$\texttt{for}\text{-}\texttt{ML} \equiv \text{Arithmetic Circuits}$};
	
	\node [circle, radius=4mm,draw=black, fill=black, inner sep=0mm] (pPALU) at ($(pINV) + (1, 0.5)$) {};
	\node [right of=pPALU,inner sep=0mm, node distance=0.5cm] (PALU) {$\textsc{PLU}$};
	
	\begin{pgfonlayer}{background}
	\node [draw=black, inner sep=1.5mm, rounded corners,fit=(SMLc)(ML)(FOMLc)(PMLc)(forML)(PALU)] {};
	\end{pgfonlayer}
	
	\end{tikzpicture}
	
	\caption{Fragments of \langfor and their equivalences. The functions \textsc{4Clique}, \textsc{DP} (diagonal product), \textsc{Inv}, \textsc{Det}, and \textsc{PLU} decomposition are placed in their fragments.} \label{thefigure}
\vspace{-2ex}
\end{figure}
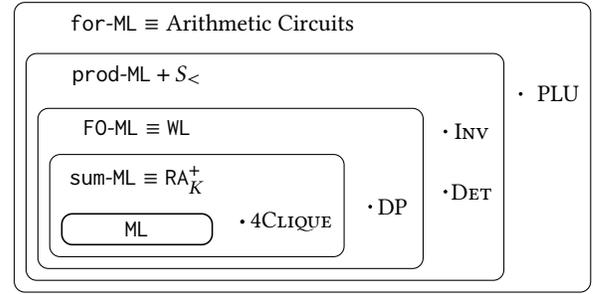

%% file: sections/wl.tex
Similarly to using sum, we can use other operations to update $X$ in the for-loop. The next natural choice is to consider products of matrices. In contrast to matrix sum, we have two options: either we can choose to use matrix product or to use the pointwise matrix product, also called the Hadamard product. We treat  matrix product in the next subsection and first explain here the connection of sum and Hadamard product operators to weighted logics.

For the rest of this section, fix a semiring $(K, \ksum, \kprod, \kzero, \kone)$. The Hadamard product over $K$-matrices can be defined as the pointwise application of $\kprod$ between two matrices of the same size. Formally, we define the expression $e \hadprod e'$ where $e, e'$ are expressions with respect to $\cS$ and $\ttype(e) = \ttype(e')$ for some schema $\Sch=(\Mnam,\size)$. Then the semantics of $e \hadprod e'$ is the pointwise application of $\kprod$, namely, $\sem{e \hadprod e'}{\I}_{ij} = \sem{e}{\I}_{ij} \kprod \sem{e'}{\I}_{ij}$ for any instance $\I$ of $\cS$. This enables us to define, similar as for  $\Sigma v$, the  pointwise-product quantifier $\qhadprod v$ as follows:
$$
\qhadprod v. \  e := \ffor{v}{X\!=\!\kone}{X \circ e}.
$$
where $\kone$ is a matrix with the same type as $X$ and all entries equal to the $\kone$-element of $K$ (i.e., we need to initialize $X$ accordingly with the $\kprod$-operator).
We cal \langprod  the subfragment of \langfor that consists of \langsum \ extended with $\qhadprod v$.

\begin{example}
	Similar to the trace of a matrix, a useful function in linear algebra is to compute the product of the values on the diagonal. 
	Using the $\qhadprod v$ operator, this can be easily expressed:
	 $$
	 e_{\mathsf{dp}}(V) := \qhadprod v. \ v^T\cdot V \cdot v.$$
\end{example}

Clearly, the inclusion of this new operator extends the expressive power to \langsum. For example,  $\sem{e_{\mathsf{dp}}}{\I}$ can be an exponentially large number in the dimension $n$ of the input.
By contrast, one can easily show that all expressions in \langsum can only return numbers polynomial in  $n$. That is, \langprod is more expressive than \langsum and $\mathsf{RA}_{K}^+$. 

\cristian{I add this example, which I believe help to say that this fragment express some interesting properties. Also I used it in the diagram and to show that there are formulas that cannot be defined in sum-matlang.}

To measure the expressive power of \langprod, we use weighted logics~\cite{DrosteG05} (WL) as a yardstick. Weighted logics extend monadic second-order logic from the boolean semiring to any semiring $K$. Furthermore, it has been used extensively to characterize the expressive power of weighted automata in terms of logic~\cite{droste2009handbook}. We use here the first-order subfragment of weighted logics to suit our purpose and, moreover, we extend its semantics over weighted structures (similar as in~\cite{GradelV17}).

A relational vocabulary $\Gamma$ is a finite collection of relation symbols such that each $R \in \Gamma$ has an associated arity, denoted by $\arity(R)$.
A $K$-weighted structure over $\Gamma$ (or just structure) is a pair $\cA = (A, \{R^\cA\}_{R \in \Gamma})$ such that $A$ is a non-empty finite set (i.e. the domain) and, for each $R \in \Gamma$, $R^\cA: A^{\arity(R)} \rightarrow K$ is a function that associates to each tuple in $A^{\arity(R)}$ a weight in $K$.

Let $X$ be a set of first-order variables. A $K$-weighted logic (WL) formula $\varphi$ over $\Gamma$ is defined by the following syntax:
$$
\begin{array}{rcl}
\varphi & := & x = y \ \mid \ R(\bar{x}) \ \mid \ \varphi \ksum \varphi \ \mid \ \varphi \kprod \varphi \ \mid \ \Sigma x. \varphi \ \mid \ \Pi x. \varphi
\end{array}
$$ 
where $x, y \in X$, $R \in \Gamma$, and $\bar{x} = x_1, \ldots, x_k$ is a sequence of variables in $X$ such that $k=\arity(R)$. As usual, we say that $x$ is a free variable of $\varphi$, if $x$ is not below $\Sigma x$ or $\Pi x$ quantifiers (e.g. $x$ is free in $\Sigma y. R(x,y)$ but $y$ is not). 
Given that $K$ is fixed, from now on we talk about structures and formulas without mentioning $K$ explicitly.  

An assignment $\sigma$ over a structure $\cA = (A, \{R^\cA\}_{R \in \Gamma})$ is a function $\sigma: X \rightarrow A$. Given $x \in X$ and $a \in A$, we denote by $\sigma[x \mapsto a]$ a new assignment such that $\sigma[x \mapsto a](y) = a$ whenever $x = y$ and $\sigma[x \mapsto a](y) = \sigma(y)$ otherwise. For $\bar{x} = x_1, \ldots, x_k$,  we write $\sigma(\bar{x})$ to say $\sigma(x_1),\ldots, \sigma(x_k)$. Given a structure $\cA = (A, \{R^\cA\}_{R \in \Gamma})$ and an assignment $\sigma$, we define the semantics $\ssem{\varphi}{\cA}(\sigma)$ of $\varphi$ as follows:
$$
\begin{array}{ll}
\text{if $\varphi := x = y$, then} & \ssem{\varphi}{\cA}(\sigma) = 
\left\{
\begin{array}{ll}
\kone & \text{if $\sigma(x) = \sigma(y)$} \\
\kzero & \text{otherwise}
\end{array}
\right. \\
\text{if $\varphi := R(\bar{x})$, then} & \ssem{\varphi}{\cA}(\sigma) = R^\cA(\sigma(\bar{x})) \\
\text{if $\varphi := \varphi_1 \ksum \varphi_2$, then} & \ssem{\varphi}{\cA}(\sigma) = \ssem{\varphi_1}{\cA}(\sigma) \ksum \ssem{\varphi_2}{\cA}(\sigma)  \\
\text{if $\varphi := \varphi_1 \kprod \varphi_2$, then} & \ssem{\varphi}{\cA}(\sigma) = \ssem{\varphi_1}{\cA}(\sigma) \kprod \ssem{\varphi_2}{\cA}(\sigma)  \\
\text{if $\varphi := \Sigma x. \, \varphi'$, then} & \ssem{\varphi}{\cA}(\sigma) =  \bigksum_{a \in A} \ssem{\varphi'}{\cA}(\sigma[x \mapsto a]) \\
\text{if $\varphi := \Pi x. \, \varphi'$, then} & \ssem{\varphi}{\cA}(\sigma) =  \bigkprod_{a \in A} \ssem{\varphi'}{\cA}(\sigma[x \mapsto a])
\end{array}
$$
When $\varphi$ contains no free variables, we omit $\sigma$ and write $\ssem{\varphi}{\cA}$ instead of $\ssem{\varphi}{\cA}(\sigma)$.

For comparing the expressive power of \langprod with WL, we have to show how to encode \lang\ instances into structures and vice versa. For this, we make two assumptions to put both languages at the same level: (1) we restrict structures to relation symbols of arity at most two and (2) we restrict instances to square matrices. The first assumption is for the same reasons as when comparing \langsum with $\mathsf{RA}_K^+$, and the second assumption is to have a crisp translation between both languages. Indeed, understanding the relation of \langprod with WL for non-square matrices is slightly more complicated and we leave this for future work. 

Let $\Sch=(\Mnam,\size)$ be a schema of square matrices, that is, there exists an $\alpha$ such that $\size(V) \in \{1, \alpha\} \times \{1,\alpha\}$ for every $V \in \Mnam$.
Define the relational vocabulary $\text{WL}(\Sch) = \{R_V \mid V \in \Mnam\}$ such that $\arity(R_V) = 2$ if $\size(V) = (\alpha, \alpha)$, $\arity(R_V) = 1$ if $\size(V) \in \{(\alpha,1), (1,\alpha)\}$, and $\arity(R_V) = 0$ otherwise.
Then given a matrix instance $\I = (\dom,\conc)$ over $\Sch$ define the structure $\text{WL}(\I) = (\{1, \ldots, n\}, \{R_V^{\I}\} )$ such that $\dom(\alpha) = n$ and $R_V^{\I}(i, j) = \conc(V)_{i,j}$ if $\size(V) = (\alpha, \alpha)$, $R_V^{\I}(i) = \conc(V)_{i}$ if $\size(V) \in \{(\alpha,1), (1,\alpha)\}$, and $R_V^{\I} = \conc(V)$ if $\size(V) = (1,1)$.

To encode weighted structures into matrices and vectors, the story is similar as for $\mathsf{RA}_K^+$. Let $\Gamma$ be a relational vocabulary where $\arity(R) \leq 2$. 
Define $\text{Mat}(\Gamma) = (\Mnam_\Gamma,\size_\Gamma)$ such that $\Mnam_\Gamma = \{ V_{R} \mid R \in \Gamma\}$ and $\size_\Gamma(V_{R})$ is equal to $(\alpha, \alpha), (\alpha, 1)$, or $(1,1)$ if $\arity(R)=2$, $\arity(R)=1$, or $\arity(R)=0$, respectively, for some $\alpha \in \DD$. Similarly, let $\cA = (A, \{R^{\cA}\}_{R \in \Gamma})$ be a structure with $A = \{a_1, \ldots, a_n\}$, ordered arbitrarily.
Then we define the matrix instance $\text{Mat}(\cA) = (\dom,\conc)$ such that $\dom(\alpha) = n$, $\conc(V_{R})_{i,j} = R^{\cA}(a_i, a_j)$ if $\arity(R)=2$, $\conc(V_{R})_{i} = R^{\cA}(a_i)$ if $\arity(R)=1$, and $\conc(V_{R}) = R^{\cA}$ otherwise.

Let $\Sch$ be a \lang\ schema of square matrices and $\Gamma$ a relational vocabulary of relational symbols of arity at most $2$. We can then show the equivalence of \langprod and WL as follows. 
\begin{proposition} \label{prop:wl}
Weighted logics over $\Gamma$ and \langprod over $\Sch$ have the same expressive power. More specifically,
\begin{itemize}
	\item for each \langprod expression $e$ over $\Sch$ such that $\Sch(e)=(1,1)$, there exists a WL-formula $\Phi(e)$ over $\text{WL}(\Sch)$ such that for every instance $\I$ of~$\Sch$, 
	$
	\sem{e}{\I} = \ssem{\Phi(e)}{\text{WL}(\I)}
	$.
	\item for each WL-formula $\varphi$ over $\Gamma$ without free variables, there exists a \langprod expression $\Psi(\varphi)$ such that for any structure $\cA$ over~$\text{Mat}(\Gamma)$,
	$
	\ssem{\varphi}{\cA}=\sem{\Psi(\varphi)}{\text{Mat}(\cA)}
	$.
\end{itemize}	
\end{proposition}

\floris{Three comments: (1) a reviewer may be puzzled about square matrix assumption; (2) is it clear why in the proposition sentences are considered (no free variables); and (3) can we say anything at all about its repercussions in terms of expressive power (apart from being equivalent)? }

%% file: sections/concl.tex
We proposed \langfor, an extension of \lang with limited recursion,
and showed that it is able to capture most of linear algebra due to its
connection to arithmetic circuits. We further revealed interesting connections
to logics on annotated relations. Our focus was on language design and
expressivity. An interesting direction for future work relates to efficient
evaluation of (fragments) of \langfor. A possible starting point is \cite{Christ_2013}
in which a general methodology for communication-optimal algorithms
for for-loop linear algebra programs is proposed.

%% file: sections/appendix.tex
\section{Preliminaries}
We first introduce some additional notations and describe simplifications that will be used later in the appendix.
\subsection{Definitions}
\input{./sections/app-def.tex}\label{app:def}

\subsection{Simplifications}\label{app:simp}
\input{./sections/app-simp.tex}

\section{Proofs of Section~\ref{sec:formatlang}}

\subsection{Order predicates}\label{app:order}
\input{./sections/app-order.tex}

\section{Proofs of Section~\ref{sec:queries}}
We next provide more details about how to perform LU-decomposition (without and with pivoting)
and to compute the determinant and inverse of a matrix.
\subsection{LU-decomposition}
\input{./sections/app-gauss.tex}

\subsection{LU-decomposition with pivoting}
\input{./sections/app-palu.tex}

\subsection{Determinant and inverse}\label{app:inverse}
\input{./sections/app-inverse.tex}
\section{Proofs of Section~\ref{sec:circuits}}

\subsection{Linear space functions}
\input{./sections/app-linspace-result.tex}

\subsection{Circuit evaluation}
\input{./sections/app-circuit-result.tex}

\subsection{From MATLANG to uniform ACs}
\input{./sections/app-lang-in-ac.tex}

\subsection{Undecidability}
\input{./sections/app-undec-result.tex}

\section{Proofs of Section~\ref{sec:restrict}}

\subsection{From \langsum to \rak}
\input{./sections/app-sum-to-ara.tex}

\subsection{From \rak to \langsum}
\input{./sections/app-ara-to-sum.tex}

\subsection{Weighted logics and \langprod}
\input{./sections/app-prod-and-wl.tex}

\subsection{Matrix inversion in \langmprod extended with order}\label{app:asset_order}
\input{./sections/app-asset-order.tex}

%% file: sections/app-def.tex
We sometimes want to iterate over $k$ canonical vectors. We define the following shorthand notation:
\begin{align*}
  \ffor{v_1,\ldots, v_k}{X}{e(X,v_1,\ldots, v_n)}:= &\ffor{v_1}{X_1}{X_1 +} \\
  &\hspace{1em}\initf{X_1}{v_2}{X_2}{X_2 + } \\
  &\hspace{2em}\initf{X_2}{v_3}{X_3}{X_3 + } \\
  &\hspace{8em}\ddots \\
  &\hspace{4em}\initf{X_{k-1}}{v_k}{X_k}{ e(X_k,v_1,\ldots, v_k)}.
\end{align*}
To reference $\ell$ different vector variables $X_1,\ldots,X_\ell$ in every iteration and update them in different ways we define:
\begin{multline*}
\ffor{v}{X_1,\ldots, X_\ell}{\left( e_1(X_1,v), e_2(X_2,v), \ldots, e_l(X_\ell,v) \right)} :=
\ffor{v}{X}{e_1(X\cdot e_{\mathsf{min}},v)\cdot (e_{\diag}(e_{\ones}(X^T))\cdot e_{\mathsf{min}})^T +\\ e_2(X\cdot e_{\mathsf{min} + 1},v)\cdot (e_{\diag}(e_{\ones}(X^T))\cdot e_{\mathsf{min} + 1})^T + \ldots + e_\ell(X\cdot e_{\mathsf{max}},v)\cdot (e_{\diag}(e_{\ones}(X^T))\cdot e_{\mathsf{max}})^T}
\end{multline*}
We note that for the latter expression to be semantically correct $v$ has to be of type $\gamma\times 1$, 
both $X_i$ and $e_i$ for $ i=1,\ldots,\ell$ have to be of type $\alpha\times 1$, 
and $X$ has to be of type $\alpha\times\beta$, where $\dom(\beta)=\ell$. Here
we use $e_{\diag}(e_{\ones}(X^T))$ to compute the $\beta\times\beta$ identity and ensure the typing of the
$e_{\mathsf{min} + i}$.
When evaluated on an instance $\I$,
$e_{\mathsf{min}}, e_{\mathsf{min} + i}$ evaluate to $b_1^{\dom(\beta)}$ and $b_{1+i}^{\dom(\beta)}$, 
respectively, and we show their defining expressions in section \ref{app:order}.
Similarly for $e_{\mathsf{max}}=b_n^{\dom(\beta)}$.
The combinations of both previous operators results in:
$$
\ffor{v_1,\ldots, v_k}{X_1,\ldots, X_\ell}{\left( e_1(X_1,v_1,\ldots, v_k), e_2(X_2,v_1,\ldots, v_k), \ldots, e_\ell(X_\ell,v_1,\ldots, v_k) \right)} :=\ffor{v_1,\ldots, v_k}{X}{e'(X,v_1,\ldots, v_k)}
$$
where 
\begin{align}
e'(X,v_1,\ldots,v_k):=&e_1(X\cdot e_{\mathsf{min}},v_1,\ldots,v_k)\cdot (e_{\diag}(e_{\ones}(X^T))\cdot e_{\mathsf{min}})^T \\
&+ e_2(X\cdot e_{\mathsf{min} + 1},v_1,\ldots,v_k)\cdot (e_{\diag}(e_{\ones}(X^T))\cdot e_{\mathsf{min} + 1})^T \\
&+ \ldots + e_\ell(X\cdot e_{\mathsf{max}},v_1,\ldots,v_k)\cdot (e_{\diag}(e_{\ones}(X^T))\cdot e_{\mathsf{max}})^T
\end{align}
It is clear that this expression iterates over $k$ canonical vectors and references $\ell$ independent vectors updating each of them in their particular way.

%% file: sections/app-simp.tex
When showing results based on induction of expressions in \langfor, it is often convenient to assume that function applications $f(e_1,\ldots,e_k)$ for $f\in\Fun_k$ are restricted to
the case when all expressions $e_1,\ldots,e_k$ have type $1\times 1$. This does not loose generality. Indeed,
for general function applications $f(e_1,\ldots,e_k)$, if we have $\ssum$, scalar product and function application on scalars (here denoted by $f_{1\times 1}$), we can simulate full function application, as follows:
 $$
f(e_1,\ldots, e_k) :=\Sigma v_i \Sigma v_j. f_{1\times 1}(v_i^T\cdot e_1\cdot v_j, \ldots ,v_i^T\cdot e_k\cdot v_j) \times v_i\cdot v_j^T.
$$

Furthermore, it also convenient at times to use the pointwise functions
$f_\odot^k:\RR^k\mapsto \RR:(x_1,\ldots,x_k)\mapsto x_1\times\cdots \cdot x_k$ and 
$f_\oplus^k:\RR^k\mapsto \RR:(x_1,\ldots,x_k)\mapsto x_1+\cdots + x_k$. In fact, it is readily observed that adding these functions does not extend the expressive power of \langfor:
\begin{lemma}
\label{lm-prod-sum}
We have that $\langforf{\emptyset} \equiv \langforf{\{f_\odot^k,f_\oplus^k \ | \ k\in \mathbf{N}\}}$.
\end{lemma}
In fact, this lemma also holds for the smaller fragments we consider.
We also observe that having $f_\odot^2:\RR^2\to\RR$ allows us to define scalar multiplication:
$$
e_1\times e_2 :=f_{\kprod}(\ones(e_2)^T\cdot e_1 \cdot \ones(e_2)^T, e_2).
$$
Conversely, $f_\odot^k$ can be expressed using scalar multiplication, as can be seen from our simulation of general function applications by pointwise function application on scalars.
Finally, a notational simplification is that when using scalars $a\in\RR$ in our expressions, we write sometimes
$a$ instead of $[a]$. For example,  $(1-e_{\ones}(v)^T\cdot v)$ stands for  $([1]-e_{\ones}(v)^T\cdot v)$.

%

%% file: sections/app-order.tex
We detail how order information on canonical vectors can be obtained in \langfor.
We provide explicit expressions for the operators mentioned in Section~\ref{sec:formatlang}
and furthermore, we also define expressions for operators that will be used in our proofs.

To begin with, we can easily obtain the last canonical vector using the expression 
$$
e_{\mathsf{max}} := \ffor{v}{X}{v}.
$$ 
In other words, we simply overwrite $X$ with the current canonical vector in each iteration.
Hence, at the end, $X$ is assigned to the last canonical vector.

%
As already mentioned in the main body of the paper,
to define an order relation for canonical vectors, we notice that the following matrix:
\[
S_{\leq} = \begin{bmatrix}
    1 & 1 & \cdots &  1 \\
    0 & \ddots & \ddots & \vdots \\
    \hdotsfor{3} & 1 \\
    0 & \cdots & \cdots & 1 
\end{bmatrix}.
\]
has the property that for two canonical vectors $b_i$ and $b_j$ of the same dimension, 
$$b_i^T\cdot S_{\leq} \cdot b_j=\begin{cases}1 & \text{if $i\leq j$}\\
0 &\text{otherwise}.
\end{cases}
$$
We observe that $S_{\leq}$ can be expressed in \langfor as 
follows:
$$
S_{\leq}:=\ffor{v}{X}{X + \bigl((X\cdot e_{\mathsf{max}}) + v \bigr)\cdot v^T + v\cdot e^T_{\mathsf{max}}},
$$
where $e_{\mathsf{max}}$ is as defined above. 
The intuition behind this expression is that by using the last canonical vector $b_n$, as returned by $e_{\mathsf{max}}$, we have access to the last column of $X$ (via the product $X\cdot e_{\mathsf{max}}$). We use this column such that after the $i$-th iteration, this column contains the $i$-th column of $S_{\leq}$. This is done by incrementing $X$ with $v\cdot e_{\mathsf{max}}^T$.
To construct $S_{\leq}$, in the $i$-th iteration we further increment $X$ with 
(i)~the current last column in $X$ (via $X\cdot e_{\mathsf{max}}\cdot v^T$) which holds
the $(i-1)$-th column of $S_{\leq}$; and (ii)~the current canonical vector (via $v\cdot v^T$). Hence, after iteration $i$, $X$ contains the first $i$ columns of $S_{\leq}$ and holds the $i$th column of $S_{\leq}$ in its last column. It is now readily verified that $X=S_{\leq}$ after the $n$th iteration.
%
%

By defining 
$$
\mathsf{succ}(u,v) := u^T\cdot S_{\leq} \cdot v,
$$
we obtain an order relation that allows us to discern whether one canonical vector comes before 
the other in the order given by $S_{\leq}$. If we want a strict order, we can just use the matrix
$S_< := S_{\leq} - e_{\mathsf{Id}}$, where $e_{\mathsf{Id}}$ is an expression in \langfor which returns the identity matrix (of appropriate dimension). Given this, we define
$$\mathsf{succ}^+(u,v) := u^T\cdot S_{<} \cdot v.$$
from which we can also derive 
$$
\mathsf{max}(u):=u^T\cdot e_{\mathsf{max}}.
$$
which is an expression that returns the last canonical vector.

Interestingly, we can also define the \textit{previous} relation between canonical vectors. 
For this, we require the following matrix:
\[
\mathsf{Prev} = \begin{bmatrix}
    0 & 1 & \cdots &  0 \\
    0 & \ddots & \ddots & \vdots \\
    \hdotsfor{3} & 1 \\
    0 & \cdots & \cdots & 0
\end{bmatrix},
\]
Using this matrix, we have that for a canonical vector $b_i$:
\[
\mathsf{Prev}\cdot b_i=\begin{cases}
               b_{i-1}, \text{ if } i > 1. \\
              \mathbf{0}, \text{ if } i = 1.
            \end{cases}
\]
where $\mathbf{0}$ is a vector of zeros of the same type as $b_i$. Notice also that $\ones(u)^T\cdot \mathsf{Prev} \cdot u$ is equal to zero, for a canonical vector $u$, if and only if $u = b_1$ is the first canonical vector, and zero otherwise.
Therefore the expression $\mathsf{min}(u)$ is defined as $$\mathsf{min}(u) := 1 - \ones(u)^T\cdot \mathsf{Prev} \cdot u,$$ and, when evaluated over canonical vectors, will result in $1$ if and only if $u=b_1$ is the first canonical vector.
To define the first canonical vector in the order given by \texttt{for}, we can then write:
$$e_{\mathsf{min}} := \ffor{v}{X}{X + \mathsf{min}(v)\times v},$$
Finally, we show that $\mathsf{Prev}$ can be defined using the following \langfor expression:
$$e_{\mathsf{Prev}}:= \ffor{v}{X}{X + \bigl((1 - \mathsf{max}(v))\times v\cdot e_{\mathsf{max}}^T - (X\cdot e_{\mathsf{max}})\cdot e_{\mathsf{max}}^T + (X\cdot e_{\mathsf{max}})\cdot v^T\bigr)}.$$
Here, $X$ is initialized as $\mathbf{0}$ and thus in the first iteration we put
 $b_1$ in the last column of $X$ (note that $X\cdot e_{\mathsf{max}}$ is also zero in the first iteration). Next, in iteration two, we add a matrix that has the stored vector $X\cdot e_{\mathsf{max}}$ (the previous canonical vector) in the column indicated by $v$ (the current canonical vector) and $v-X\cdot e_{\mathsf{max}}$ in the last column, to replace the vector stored. As a consequence, $b_2$ is now stored in the last column. In the last iteration, we have $b_{n-1}$ already in the last column, so no further update of $X$ is required.
 
To get the \textit{next} relation we simply do $e_{\mathsf{Next}} = e_{\mathsf{Prev}}^T$. We have that for a canonical vector $b_i$:
\[
{\mathsf{Next}}\cdot b_i=\begin{cases}
               b_{i+1}, \text{ if } i < n. \\
              \mathbf{0}, \text{ if } i = n.
            \end{cases}
\]
In this way, we also can obtain the following operators for a canonical vector $v$: 
$$\mathsf{prev}(v):=e_{\mathsf{Prev}}\cdot v.$$
$$\mathsf{next}(v):=e_{\mathsf{Next}}\cdot v.$$
More generally, we define 
\begin{align*}
    e_{\mathsf{getPrevMatrix}}(v)&:=\sprod w.  \mathsf{succ}(w,v)\times e_{\mathsf{Prev}} + (1 - \mathsf{succ}(w,v))\times e_{\mathsf{Id}}\\
    e_{\mathsf{getNextMatrix}}(v)&:=\sprod w. \mathsf{succ}(w,v)\times e_{\mathsf{Next}} + (1 - \mathsf{succ}(w,v))\times e_{\mathsf{Id}}
\end{align*}
expressions that, when $v$ is interpreted as canonical vector $b_i$, output $\mathsf{Prev}^i$ and $\mathsf{Next}^i$ respectively.
Note that
\[
\mathsf{Prev}^j\cdot b_i=\begin{cases}
               b_{i-j}, \text{ if } i > j. \\
              \mathbf{0}, \text{ if } i \leq j.
            \end{cases}
\]
and
\[
\mathsf{Next}^j\cdot b_i=\begin{cases}
               b_{i+j}, \text{ if } i + j \leq n. \\
              \mathbf{0}, \text{ if } i + j > n.
            \end{cases}
\]
Finally, define
$$
e_{\mathsf{min}+i}:=\underbrace{e_{\mathsf{getNextMatrix}}(\ldots e_{\mathsf{getNextMatrix}}}_{i \text{ times}}(e_{\mathsf{min}}))
$$
and
$$
e_{\mathsf{max}-i}:=\underbrace{e_{\mathsf{getPrevMatrix}}(\ldots e_{\mathsf{getPrevMatrix}}}_{i \text{ times}}(e_{\mathsf{max}}))
$$
We note that some these expressions were already used in Section~\ref{app:def}.

%% file: sections/app-gauss.tex
\newtheorem*{ALU}{Proposition~\ref{prop:gauss}}

We start with LU-decomposition without pivoting. We recall proposition \ref{prop:gauss}:
\begin{ALU}
  There exists $\langforf{f_/}$ expressions $e_L(V)$ and $e_U(V)$ such that
  $\sem{e_L}{\I}=L$ and $\sem{e_U}{\I}=U$ form an LU-decomposition of $A$,
  where $\conc(V)=A$ and $A$ is LU-factorizable.
\end{ALU}
\begin{proof}
	Let $A$ be an LU-factorizable matrix. We already explained how the expression 
	$e_U(V)$ is obtained in the main body of the paper, i.e., 
	$$
	e_{U}(V) :=  \left( \initf{e_{\mathsf{Id}}}{y}{X}{\red{X\cdot V}{y}\cdot X} \right) \cdot V.
	$$
	We recall that $e_U(A)=T_n\cdot\cdots\cdot T_1\cdot A$ with $L^{-1}=T_n\cdot\cdots\cdot T_1$. Let
	$$
	e_{L^{-1}}(V) :=  \initf{e_{\mathsf{Id}}}{y}{X}{\red{X\cdot V}{y}\cdot X}.
	$$
such that	$$
	e_{\mathsf{U}}(V) :=  e_{L^{-1}}(V) \cdot V.
	$$
	%
%
%
It now suffices to observe that, since $T_n=I$,
\begin{align*}
  L^{-1}&=(I-c_1\cdot b_1^T)\cdots (I-c_{n-1}\cdot  b_{n-1}^T) \\
  &=I-c_1\cdot b_1^T-\cdots - c_{n-1}\cdot b_{n-1}^T
\end{align*}
and hence,
\begin{align*}
  L&=(I+c_1\cdot b_1^T)\cdots (I+c_{n-1}\cdot b_{n-1}^T) \\
  &=I+c_1\cdot b_1^T+\cdots + c_{n-1}\cdot b_{n-1}^T.
\end{align*}
As a consequence, to obtain $L$ from $L^{-1}$ we just need to multiply every entry below the diagonal by $-1$. Since both  $L$ and $L^{-1}$ are lower triangular, this can done 
by computing $L=-1\times L^{-1} + 2\times I$. Translated into \langfor, this means that we can define
$$
e_{L}(V) :=  -1\times e_{L^{-1}}(V) + 2\times e_{\mathsf{Id}},
$$
which concludes the proof of the proposition.
\end{proof}

%% file: sections/app-palu.tex
\newtheorem*{PALU}{Proposition~\ref{prop:palu}}
We next consider LU-decomposition with pivoting. We recall proposition \ref{prop:palu}:

\begin{PALU}
  There exist expressions $e_{L^{-1}P}(V)$ and $e_U(V)$ in $\langforf{f_/,f_{>0}}$  such that
  $L^{-1}\cdot P=\sem{e_{L^{-1}P}}{\I}$ and $U=\sem{e_U}{\I}$, satisfy $L^{-1}\cdot P\cdot A=U$, where $\I$ is an instance such that $\conc(V)=A$. 
\end{PALU}
\begin{proof}
We assume that $f_{/}$ and $f_{>0}$ are in $\mathcal{F}$. Let $A$ be an arbitrary matrix.
%
By contrast to when $A$ is LU-factorizable, during the LU-decomposition process we may need row interchange (pivoting) in each step of the iteration. Let us assume that row interchange is needed immediately before 
step $k$, $1\leq k\leq n$. In other words, we now aim to reduce the $k$-th column of $A_k=T_{k-1}\cdots T_1\cdot A$, 
or $A_k=A$ if $k=1$, but now $A_k$ has a zero pivot, i.e., $(A_{k})_{kk}=0$. 
Let $P$ be the matrix that denotes the necessary row interchange. If we know
$P$, then 
to compute $T_k$ we need to perform $\red{P\cdot X\cdot A}{v}$ in this iteration,
where $\red{\cdot}$ is the expression in \langfor reducing a column, as explained in the main body of the paper.
Furthermore, we need to apply the permutation $P$ to the current result, resulting in the 
expression $\initf{I}{v}{X}{\red{P\cdot X\cdot A}{v}\cdot P\cdot X}$. We now remark that
$P$ is a permutation matrix of the  form $P = I - u\cdot u^T$ and it denotes an interchange (if multiplied by left) of rows $i$ and $j$ if $u=(b_{i}-b_{j})$. Note that we are performing a row interchange for column $k$ and thus $i=k$ and $j>k-1$. If no interchange is needed, $i=j=k$ and $P=I$.
Also note that when $k=n$ no interchange takes place. Furthermore, if no suitable $b_j$ can
be found, this implies that no interchange is required as well and we can move on to next column.

To find the vector $u$ in $P$, we can, for example, find the first entry $j\geq k$ in column $k$ of $A_k$ that holds a non-zero value. More generally, we can find the first entry in a vector $a$ that holds a non-zero value by using the function $f_{>0}$. Indeed, consider the following expression:
$$
\nneq{a}{u}:=\ffor{v}{X}{\left( 1-e_{\ones}(v)^T\cdot X \right) \times f_{>0}\left( ( v^T\cdot a )^2 \right)\times v + \mathsf{max}(v)\times\left( 1-e_{\ones}(v)^T\cdot X \right)\times \left( 1 - f_{>0}\left( ( v^T\cdot a )^2 \right) \right) \times u}
$$
Here, $\nneq{a}{u}$ receives two $n$ dimensional vectors $a$ and $u$ and outputs a 
canonical vector $b_j$ such that $a_j$ is the first non-zero entry of $a$, or $u$ if such non-zero value does not exist in $a$. We check for $f_{>0}((\cdot)^2)$ 
in case a negative number is tested. The above expression simply checks in each iteration
whether $X$ already holds a canonical vector. If so, then $X$ is not updated. Otherwise,
$X$ is replaced by the current canonical vector $b_j$ if and only if $b_j^T\cdot a$ is non-zero. Furthermore, when the final canonical vector is considered and $X$ does not hold
a canonical vector yet and $b_n^T\cdot a$ is zero, the vector $u$ is returned.

We use $\nneq{a}{u}$ to find a pivot for a specific column. Let us assume again that we
want to find a pivot in column $k$ of $A_k$. We can then first make all entries in that column, with indexes smaller or equal to $k$, zero, just as we did by means of $\ccol{\cdot}{\cdot}$ in the
definition of $\red{\cdot}$. Except, now we also need to make the $k$the entry zero as well.
Let us denote by $\ccoleq{\cdot}{\cdot}$ the operation $\ccol{\cdot}{\cdot}$, as defined in the main body of the paper, but using $\mathsf{succ}$ instead of $\mathsf{succ}^+$ (to include the $k$ entry). Given this, we can construct $P=I-u\cdot u^T$ as follows:
$$
e_{P_u}(A,u) := e_{\mathsf{Id}} - \left[ u - \nneq{ \ccoleq{A}{u} }{u} \right]\cdot \left[ u - \nneq{ \ccoleq{A}{u} }{u} \right]^T.
$$ 
From the explanations given above, it should be clear that $e_{P_u}(A,u)$ computes the necessary permutation matrix of $A_k$ for the column indicated by $u$, or $I$
if no permutation is needed, or if such permutation does not exist (so we skip the current column). Also, we have to modify the $\red{V}{y}$ operators, as follows:
$$
\red{V}{y}:= e_{\mathsf{Id}}+ f_{>0}\left( ( y^T\cdot V\cdot y)^2 \right)\times f_/(\ccol{V}{y},\left[ -(y^T\cdot V\cdot y)\times e_{\ones}(y) + \left( 1 - f_{>0}\left( ( y^T\cdot V\cdot y)^2 \right) \right)\times e_{\ones}(y) \right])\cdot y^T,
$$
so that when $V$ is interpreted by a matrix $B$ and $y=b_i$, it returns $I+c_ib_i^T$ if $B_{ii}$ is not zero. 
If $B_{ii}=0$ then we divide $\ccol{B}{b_i}$ by $e_{\ones}(b_i)$ (so we don't get \textit{undefined}), 
but we don't add $c_ib_i^T$ precisely because $B_{ii}=0$, and return the identity so nothing happens. We check 
for $f_{>0}((\cdot)^2)$ in case a negative number is tested.

\thomas{Extra line of explanation added.}

Finally, we define
$$
e_{L^{-1}P}(V):=\initf{e_{\mathsf{Id}}}{v}{X}{\red{e_{P_v}(X\cdot V,v)\cdot X\cdot V}{v}\cdot e_{P_v}(X\cdot V,v)\cdot X}
$$
and $e_{\mathsf{U}}(V):=e_{L^{-1}P}(V)\cdot V$ as the desired expressions.

As a final observation, in the definition of $e_{L^{-1}P}(V)$ 
we interlaced permutation matrices with the $T_i$'s. More specifically, 
$A_k=T_k\cdot P\cdot T_{k-1}\cdots T_1\cdot A$. We observe, however, that for $\ell\leq k-1$ and
$T_{\ell}=I-c_\ell\cdot b_\ell^T$, we have that  $b_\ell^T\cdot P=b_\ell$ because $b_\ell$ has zeroes in positions in the rows involved in the row exchange $P$. Also, note that  $P^2=I$ and thus 
$$P\cdot T_\ell\cdot P=P^2-P\cdot c_\ell\cdot b_\ell^T\cdot P=I-\widehat{c}_\ell\cdot b_l^T=\widehat{T}_\ell.$$
%
%
As a consequence,
$$
T_k\cdot P\cdot T_{k-1}\cdots T_1=T_k\cdot P\cdot T_{k-1}\cdot P^2\cdot T_{k-2}\cdot P^2\cdots P^2 \cdot T_1\cdot P^2=T_k\cdot (P\cdot T_{k-1}\cdot P)\cdots (P\cdot T_1\cdot P)\cdot P=\widehat{T}_{k-1}\cdots \widehat{T}_1\cdot P,
$$
and thus we may assume that $P$ occurs at the end. Hence, we obtain $L^{-1}\cdot P\cdot A=U$.
%
%
\end{proof}

%% file: sections/app-inverse.tex
\newtheorem*{INVERSE}{Proposition~\ref{prop:inverse}}
We next turn our attention to computing the inverse and determinant of a matrix.
To show Proposition \ref{prop:inverse} we first show that it holds when considering non-singular lower or upper triangular matrices.
\begin{lemma}\label{prop:upperlowerinverse}
There are $\langforf{f_/}$ expressions $e_{\mathsf{upperDiagInv}}(V)$ and $e_{\mathsf{lowerDiagInv}}(V)$
such that $\sem{e_{\mathsf{upperDiagInv}}}{\I}=A^{-1}$ when $\I$ assigns $V$
to an invertible upper triangular matrix $A$ and $\sem{e_{\mathsf{lowerDiagInv}}}{\I}=A^{-1}$ when $\I$ assigns $V$
to an invertible lower triangular matrix $A$.
\end{lemma}

\begin{proof} We start by considering the expression:
$$
e_{\mathsf{ps}}(V):= e_{\mathsf{Id}} + \ssum v.\sprod w. \left[ \mathsf{succ}(w,v)\times V + (1 - \mathsf{succ}(w,v))\times e_{\mathsf{Id}} \right].
$$
Here, $e_{\mathsf{ps}}(A)$ results in $I+A+A^2+\cdots + A^n$ for any matrix $A$. In the expression, the outer loop defines which power we compute. 
That is, when $v$ is the $i$th canonical vector, we compute $A^i$.
Computing $A^i$ is achieved via the inner product loop, which uses $\mathsf{succ}(w,v)$ to 
determine whether $w$ comes before or is $v$ in the ordering of canonical vectors.
When this is the case, we multiply the current result by $A$, and when $w$ is greater 
than $v$, we use the identity as not to affect the already computed result. We add the identity at the end.

Now, let $A$ be an $n\times n$ matrix that is upper triangular and let $D_A$ be the matrix consisting of the diagonal elements of $A$, i.e.,
\[
D_A = \begin{bmatrix}
    a_{11} & \cdots & \cdots &  0 \\
    0 & a_{22} & \cdots &  0 \\
    0 & \ddots & \vdots & \vdots \\
    \vdots & \cdots& \cdots & a_{nn}
\end{bmatrix}.
\]
We can compute $D_A$ by the expression:
$$
e_{\mathsf{getDiag}}(V) := \ssum v. (v^T\cdot V\cdot v) \times v\cdot v^T.
$$
Let $T=A-D_A$, then
$$
A^{-1}=\left[ D_A+T \right]^{-1}= \left[ D_A\left( I+D_A^{-1}T\right) \right]^{-1} = \left( I+D_A^{-1}T\right)^{-1}D_A^{-1}.
$$
We now observe that $D_{A}^{-1}$ simply consists of the inverses of the elements on the diagonal. This can be expressed, as follows:
$$
e_{\mathsf{diagInverse}}(V):=\ssum v. f_{/}(1,v^T\cdot V\cdot v)\times v\cdot v^T=\ssum v. f_{/}(1,v^T\cdot V\cdot v)\times v\cdot v^T,
$$
Where $f_{/}$ is the division function. In the last equality we take advantage of the fact that the diagonals of $A$ and $D_A$ are the same.

We now focus on the computation of $\left( I+D_A^{-1}\cdot T\right)^{-1}$. First, by construction, $D_A^{-1}\cdot T$ is strictly upper triangular and thus nilpotent, 
such that $\left( D_A^{-1}\cdot T\right)^n=0$, where $n$ is the dimension of $A$.
Recall the following algebraic identity 
$$(1+x)\left( \sum_{i=0}^{m}(-x)^i \right)=1-(-x)^{m+1}.$$
By choosing $m=n-1$ and applying it to $x=D_A^{-1}\cdot T$, we have
$$
\left(I+D_A^{-1}\cdot T \right)\left( \sum_{i=0}^{n-1}(-D_A^{-1}\cdot T)^i \right)=I- \left( -D_A^{-1}\cdot T\right)^n =I.
$$
Hence,
$$
\left(I+D_A^{-1}\cdot T \right)^{-1}=\sum_{i=0}^{n-1}(-D_A^{-1}\cdot T)^i=\sum_{i=0}^{n}(-D_A^{-1}\cdot T)^i.
$$
We now observe that
$$
e_{\mathsf{ps}}(-1\times D_A^{-1}\cdot T)=\sum_{i=0}^{n}(-D_A^{-1}\cdot T)^i=\left(I+D_A^{-1}\cdot T \right)^{-1},
$$
and thus 
$$
A^{-1}= e_{\mathsf{ps}}\left(-1\times \left[e_{\mathsf{diagInverse}}(A)(A-e_{\mathsf{getDiag}}(A))\right] \right)e_{\mathsf{diagInverse}}(A).
$$
Seeing this as an expression:
$$
e_{\mathsf{upperDiagInv}}(V):= e_{\mathsf{ps}}\left(-1\times \left[e_{\mathsf{diagInverse}}(V)(V-e_{\mathsf{getDiag}}(V))\right] \right)e_{\mathsf{diagInverse}}(V),
$$
we see that  when interpreting $V$ as an  upper triangular invertible matrix, 
$e_{\mathsf{upperDiagInv}}(A)$ evaluates to $A^{-1}$.

To deal with invertible lower triangular matrices $A$, we observe that  $\left(A^{-1}\right)^T=\left(A^T\right)^{-1}$ and $A^T$ is upper triangular.
Hence, it suffices to define
$$
e_{\mathsf{lowerDiagInv}}(V):= e_{\mathsf{upperDiagInv}}(V^T)^T.
$$
This concludes the proof of the lemma.
\end{proof}

We are now ready to prove proposition \ref{prop:inverse}. We recall:
\begin{INVERSE}
  There are $\langforf{f_/}$ expressions $e_{\mathsf{det}}(V)$ and $e_{\mathsf{inv}}(V)$ such that
  $\sem{e_{\mathsf{det}}}{\I}=\mathsf{det}(A)$, and  
  $\sem{e_{\mathsf{inv}}}{\I}=A^{-1}$ when $\I$ assigns $V$
  to $A$ and $A$ is invertible.
\end{INVERSE}

\begin{proof}
Let $A$ be an $n \times n$ matrix. As mentioned in the main body of the paper, we will implement Csanky's algorithm. Let $p_A(x):=\mathsf{det}(xI-A)$ denote  characteristic polynomial of $A$.  We write $p_A(x)=1 + \sum_{i=1}^n c_ix^i$ and let  $S_i:=\frac{1}{i+1}\mathsf{tr}(A^i)$ with $\mathsf{tr}(\cdot)$ the trace operator which sums up the diagonal elements of a matrix.
Then, the coefficients $c_1,\ldots,c_n$ are known to satisfy\footnote{We use a slightly different, but equivalent, system of equations than the one mentioned in the paper.} 
$$
\underbrace{\left(\begin{matrix}
1 & 0 & 0 & \cdots & 0 & 0\\
S_1 & 1 & 0 & \cdots  &0 & 0\\
S_2 & S_1 & 1 & \cdots  &0 & 0\\
\vdots & \vdots & \vdots & \vdots & \vdots & 0\\
S_{n-1} & S_{n-2} & S_{n-3} & \cdots & S_1 & 1\\
\end{matrix}\right)}_{S}\cdot
\underbrace{\left(\begin{matrix}
c_1\\
c_2\\
c_3\\
\vdots\\
c_n\\
\end{matrix}\right)}_{\bar c}=\underbrace{\left(\begin{matrix}
S_1\\
S_2\\
S_3\\
\vdots\\
S_n\\
\end{matrix}\right)}_{\bar b}
$$
and furthermore, $c_n=(-1)^n\mathsf{det}(A)$ and if $c_{n}\neq 0$, then
$$
A^{-1}=\frac{1}{c_n}\sum_{i=0}^{n-1}c_i A^{n-1-i},
$$
with $c_0=1$. It is now easy to see that we can compute the $S_i's$ in $\langfor$. Indeed, for
$i=1,\ldots,n$ we can consider
$$
e_{\mathsf{powTr}}(V,v):=\ssum w. w^T\cdot\left(e_{\mathsf{pow}}(V,v)\cdot V\right)\cdot w
$$
with 
$$
e_{\mathsf{pow}}(V,v):= \sprod w. (\mathsf{succ}(w,v)\times V+(1-\mathsf{succ}(w,v))\times e_{\mathsf{Id}}).
$$
We have that $e_{\mathsf{pow}}(A,b_j)=A^{j}$ and thus $e_{\mathsf{powTr}}(A,b_j)=\mathsf{tr}(A^{j})$. Define:
$$
e_{S}(V,v):=f_{/}(1, 1 + \ssum w. \mathsf{succ}(w,v))\times e_{\mathsf{powTr}}(V,v).
$$
Here $e_{S}(A,b_i)=S_i$. Note that $i+1$ is computed summing up to the dimension indicated by $v$, and adding 1.
We can now easily construct the vector $\bar b$ used in the system of equations by means of the expression:
$$
e_{\bar b}(V):=\ssum w.e_{S}(V,w)\times w.
$$
We next construct the matrix $S$. We need to be able to \textit{shift} a vector $a$ in $k$ positions, i.e.,
such that $(a_1,\ldots,a_n)\mapsto (0,\ldots,a_1,\ldots,a_{n-k})$. We use $e_{\mathsf{getNextMatrix}}$ 
defined in section \ref{app:order}, i.e., we define:
$$
e_{\mathsf{shift}}(a,v):=\ssum w.(w^T\cdot a)\times(e_{\mathsf{getNextMatrix}}(v)\cdot w)
$$
performs the desired shift when $u$ is assigned a vector $a$ and $v$ is $b_k$. 
The matrix $S$ is now obtained as follows:
$$
S(V):= e_{\mathsf{Id}} + \ssum v. e_{\mathsf{shift}}(e_{\bar b}(V), v)\cdot v^T
$$
We now observe that $S$ is lower triangular with nonzero diagonal entries. So,
Lemma~\ref{prop:upperlowerinverse} tells us that we can invert it, i.e.,
$e_{\mathsf{lowerDiagInv}}(S)=S^{-1}$. As a consequence,
$$
e_{\bar c}(V):=e_{\mathsf{lowerDiagInv}}(S(V))\cdot e_{\bar b}(V).
$$
outputs $\bar c$ when $V$ is interpreted as matrix $A$. Observe that we only use the division operator. We now have all coefficients of the characteristic polynomial of $A$.

We can now define
$$
e_{\mathsf{det}}(V):=\left( \left(\left( \sprod w. (-1)\times e_{\ones}(V)\right)^T\cdot e_{\mathsf{max}}\right) \times e_{\bar c}(V) \right)^T\cdot e_{\mathsf{max}},
$$
an expression that, when $V$ is interpreted as any matrix $A$, outputs $\mathsf{det}(A)$.
Here, $(\sprod w. (-1)\times e_{\ones}(V))$ is the $n$ dimensional vector with $(-1)^n$ in all of its entries.
Since $c_n=(-1)^n\mathsf{det}(A)$, we extract $(-1)^n(-1)^n\mathsf{det}(A)=\mathsf{det}(A)$ with $e_{\mathsf{max}}$.

For the inverse, we have that
$$
A^{-1}=\frac{1}{c_n}\sum_{i=0}^{n-1}c_i A^{n-1-i} = \frac{1}{c_n}A^{n-1} + \sum_{i=1}^{n-1}\frac{c_i}{c_n}A^{n-1-i}.
$$
We compute $\frac{1}{c_n}A^{n-1}$ as
$$
f_{/}(1, e_{\bar c}(A)^T\cdot e_{\mathsf{max}})\times e_{\mathsf{pow}}(A, e_{\mathsf{max}})
$$
and $\sum_{i=1}^{n-1}\frac{c_i}{c_n}A^{n-1-i}$ as
$$
\ssum v. f_{/}\left( e_{\bar c}(A)^T\cdot v, e_{\bar c}(A)^T\cdot e_{\mathsf{max}} \right)\times e_{\mathsf{invPow}}(A, v),
$$
where
$$
e_{\mathsf{invPow}}(V, v):= \sprod w. (1-\mathsf{max}(w)) \times \left[ (1 - \mathsf{succ}(w,v))\times V + \mathsf{succ}(w,v)\times e_{\mathsf{Id}} \right] + \mathsf{max}(w)\times e_{\mathsf{Id}}.
$$
Here, $e_{\mathsf{invPow}}(A, b_i)=A^{n-1-i}$ and $e_{\mathsf{invPow}}(A, b_n)=I$.
Note that we always multiply by $e_{\mathsf{Id}}$ in the last step.
To conclude, we define:
$$
e_{\mathsf{inv}}(V):= f_{/}(1, e_{\bar c}(V)^T\cdot e_{\mathsf{max}})\times e_{\mathsf{pow}}(V, e_{\mathsf{max}}) + \left[ \ssum v. f_{/}\left( e_{\bar c}(V)^T\cdot v, e_{\bar c}(V)^T\cdot e_{\mathsf{max}} \right)\times e_{\mathsf{invPow}}(V, v) \right],
$$
an expression that, when $V$ is interpreted as any invertible matrix $A$, computes $A^{-1}$.
\end{proof}

\thomas{Is the following okay? Delete or modify as you want.}

As an observation, here we only use operators $\ssum$ and $\sprod$ defined in section \ref{sec:restrict}. 
We also assume access to order.

%% file: sections/app-linspace-result.tex
We start by showing a crucial ingredient for making the correspondence between \langfor
and arithmetic circuits. More specifically, 
we show that any polynomial time Turing machine, working within linear space and producing linear space output, can be simulated in \langfor. 
For this proof and section only, we will denote the canonical vectors as
$\mathbf{e}_1, \ldots, \mathbf{e}_n$, since $b$ will be used to represent a value on a position of a tape.

We consider  deterministic Turing Machines (TM) $T$ consisting of $\ell$ read-only input tapes, 
denoted by $R_1,\ldots,R_\ell$,
a work tape, denoted by $W$, and a write-only output tape, denoted by $O$. The TM $T$ has a set $Q$ of $m$
states, denoted by $q_0,\ldots,q_m$. We assume that $q_0$ is the initial state and $q_m$ is the accepting state.
The input and tape alphabet are $\Sigma=\{0,1\}$ and $\Gamma=\Sigma\cup\{\rhd,\lhd\}$, respectively. The special 
symbol $\rhd$ denotes the beginning of each of the tapes, the symbol $\lhd$ denotes the end of the $\ell$ input tapes. 
The transition function $\Delta$ is defined as usual, i.e., 
$\Delta:Q\times \Gamma^{\ell+2} \to Q\times \Gamma^{2}\times \{\leftarrow,\sqcup,\rightarrow\}^{\ell+2}$ 
such that $\Delta(q,(a_1,\ldots,a_{\ell},b,c))=\bigl(q',(b',c'),(\mathsf{d}_1,\ldots,\mathsf{d}_{\ell+2})\bigr)$
with $\mathsf{d}_i\in \{\leftarrow,\sqcup,\rightarrow\}$, means that when $T$ is in state $q$ and the $\ell+2$ 
heads on the tapes read symbols $a_1,\ldots,a_{\ell},b,c$, respectively, then $T$ transitions to state $q'$,
writes $b',c'$ on the work and output tapes, respectively, at the position to which the work and output 
tapes' heads points at, and finally moves the heads on the tapes according 
$\mathsf{d}_1,\ldots,\mathsf{d}_{\ell+2}$. More specifically, $\leftarrow$  indicates a move to the left, 
$\rightarrow$ a move to the right, and finally, $\sqcup$ indicates that the head does not move.

We assume that $\Delta$ is defined such that it ensures that on none of the tapes, heads can move beyond 
the leftmost marker $\rhd$. Furthermore, the tapes $R_1,\ldots,R_\ell$ are treated as read-only and the heads 
on these tapes cannot move beyond the end markers $\lhd$. Similarly, $\Delta$ ensures that the output tape $O$ 
is write only, i.e., its head cannot move to the left.  We also assume that $\Delta$ does not change the 
occurrences of $\rhd$ or writes $\lhd$ on the work and output tape.

A configuration of $T$ is defined in the usual way. That is, a configuration of the input tapes is of the form
$\rhd w_1qw_2\lhd$ with $w_1,w_2\in\Sigma^*$ and represents that the current tape content is 
$\rhd w_1w_2\lhd$, $T$ is in state $q$ and the head is positioned on the first symbol of $w_2$. 
Similarly, configurations of the work and output tape are represented by $\rhd w_1qw_2$. 
A configuration of $T$ is consists of configurations for all tapes. Given two configurations 
$c_1$ and $c_2$, we say that $c_1$ yields $c_2$ if $c_2$ is the result of applying the transition 
function $\Delta$ of $T$ based on the information in $c_1$. As usual, we close this ``yields'' relation 
transitively.

Given $\ell$ input words $w_1,\ldots,w_\ell\in\Sigma^*$, we assume that the initial configuration of 
$T$ is given by
$\bigl(q_0\rhd  w_1\lhd,q_0\rhd w_2\lhd,\ldots, q_0\rhd w_\ell\lhd,q_0\rhd, q_0\rhd \bigr)$ and an 
accepting configuration is assumed to be of the form 
$\bigl(\rhd q_m w_1\lhd,\rhd q_m w_2\lhd,\ldots, \rhd q_m w_\ell\lhd,\rhd q_m,\rhd q_m w\bigr)$ for some
$w\in\Sigma^*$. We say that $T$ computes the function $f:(\Sigma^*)^{\ell}\to\Sigma^*$ if for every
$w_1,\ldots,w_\ell\in\Sigma^*$, the initial configuration yields (transitively) an accepting 
configuration such that the configuration on the output tape is
given by $\rhd q_m f(w_1,\ldots,w_\ell)$.

We assume that once $T$ reaches an accepting configuration it stays indefinitely in that configuration 
(i.e., it loops). We further assume that $T$ only reaches an accepting configuration when all its input
words have the same size. Furthermore, when all inputs have the same size, $T$ will reach an accepting 
configuration.

We say that $T$ is a \textit{linear space machine} when it reaches an accepting configuration 
on inputs of size $n$ by using $\mathcal{O}(n)$ space on its work tape and additionally needs 
$\mathcal{O}(n^k)$ steps to do so. A \textit{linear input-output function} is a function of the form 
$f=\bigcup_{n\geq 0} f_n:(\Sigma^n)^\ell\to\Sigma^n$. In other words, for every $\ell$ words of the same 
size $n$, $f$ returns a word of size $n$. We say that a linear input-output function is a 
\textit{linear space input-output function} if
there exists a linear space machine  $T$ that for every $n\geq 0$, on input $w_1,\ldots,w_\ell\in\Sigma^n$ 
the TM $T$ has
$f_n(w_1,\ldots,w_\ell)$ on its the output tape when (necessarily) reaching an accepting configuration.

\begin{proposition} \label{prop:transducer}
Let $f=\bigcup_{n\geq 0}f_n:(\Sigma^n)^\ell\to \Sigma^n$ be a linear space input-ouput function 
computed by a linear space  machine $T$ with $m$ states, $\ell$ input tapes, which consumes 
$\mathcal{O}(n)$ space and runs in $\mathcal{O}(n^{k-1})$ time on inputs of size $n$. 
There exists (i)~a $\mathsf{MATLANG}$ 
schema $\mathcal{S}=(\mathcal{M},\textsf{size})$ where $\mathcal{M}$ consists matrix 
variables\footnote{We also need a finite number of auxiliary variables, these will be specified 
in the proof.} 
$Q_1,\ldots,Q_m,R_1,\ldots,R_\ell,H_1,\ldots,H_\ell,W_1,\ldots,W_s,H_{W_1},\ldots,H_{W_s},O,H_O, v_1,\ldots,v_{k}$ 
with $\mathsf{size}(V)=\alpha\times 1$ for all $V\in\mathcal{M}$; and (ii)~a $\mathsf{MATLANG}$ 
expression $e_f$ over $\mathcal{S}$ such that for the instance 
$\I=(\mathcal{D},\textsf{mat})$ over $\mathcal{S}$ with $\mathcal{D}(\alpha)=n$ and 
$$\mathsf{mat}(R_i)=\mathsf{vec}(w_i)\in \mathbb{R}^n\text{, for $i\in[\ell]$ and all other matrix variables instantiated with the zero vector in $\mathbb{R}^n$} $$
for words $w_1,\ldots,w_\ell\in\Sigma^n$ and such that $\mathsf{vec}(w_i)$ is the $n\times 1$-vector 
encoding the word $w_i$, we have that  $\mathsf{mat}(O)=\mathsf{vec}(f_n(w_1,\ldots,w_n))\in\mathbb{R}^n$ 
after evaluating $e_f$ on $\I$.
\end{proposition}
\begin{proof}
	The expression $e_f$ we construct will simulate the TM $T$. To have some more control on the space 
	and time consumption of $T$, let us first assume that $n$ is large enough, say larger than $n\geq N$, 
	such that $T$ runs in $sn$ space and $cn^{k-1}\leq n^k$ time for constants $s$ and $c$. We deal with $n<N$ later on.

\floris{To connect with my earlier comment in B1. I guess we assume that $n$ is also large enough to update $m+\ell+2s+2$independent vectors??}
\thomas{I updated the definition of the operator so there should be no problem now.}
To simulate $T$ we need to encode states, tapes and head positions. The matrix variables in 
$\mathcal{M}$ mentioned in the proposition will take these roles. More specifically, the variables 
$R_1,\ldots,R_\ell$ will hold the input vectors, $W_1,\ldots,W_s$ will hold the contents of the work
tape, where $s$ is the constant mentioned earlier, and $O$ will hold the contents of the output tape. 
The vectors corresponding to the work and output tape are initially set to the zero vector. 
The vector for the input tape $R_i$ is set to $\mathsf{vec}(w_i)$, for $i\in[\ell]$.

 With each tape we associate a matrix variable encoding the position of the head. More specifically, 
 $H_1,\ldots,H_\ell$ correspond to the input tape heads,
$H_{W_1},\ldots, H_{W_s}$ are the heads for the work tape, and $H_O$ is the head of the output tape. 
All these vectors are initialised with the zero vector. Later on, these vectors will be zero except 
for a single position, indicating the positions in the corresponding tapes the heads point to. 
For those positions $j$, $1<j<n$, the head vectors will carry value $1$.  When $j=1$ or $n$ and when 
it concerns positions for the input tape, the head vectors can carry value $1$ or $2$. We need to treat 
these border cases separately
because we only have $n$ positions available to store the input words, whereas the actual input tapes 
consist of $n+2$ symbols because of $\rhd$ and $\lhd$. So when, for example, $H_1$ has a $1$ in its first
entry, we interpret it as the head is pointing to the first symbol of the input word $w_1$. When $H_1$
has a $2$ in its first position, we interpret it as the head pointing to $\rhd$. The end marker $\lhd$ is
dealt with in the same way, by using value $1$ or $2$ in the last position of $H_1$. We use this encoding
for all input tapes, and also for the work tape $W_1$ and output tape $O$ with the exception that no end 
marker $\lhd$ is present.


To encode the states, we use the variables $Q_1,\ldots,Q_m$. We will ensure that when $T$ is in state 
$q_i$ when
 $\mathsf{mat}(Q_i)=[1,0,\ldots,0]^T\in\mathbb{R}^n$, otherwise $\mathsf{mat}(Q_i)$ is the zero 
 vector in $\mathbb{R}^n$.	

Finally, the variables $v_1,\ldots,v_{k}$ represent $k$ canonical vectors  which are used to iterate 
in for-loops. By iterating over then, we can perform $n^{k}$ iterations, 
which suffices for simulating the $\mathcal{O}(n^{k-1})$ steps used by $T$ to reach an accepting configuration. 

With these matrix variables in place, we start by defining $e_f$. It will consists of two subexpressions
$e_f^{\geq N}$, for dealing with $n\geq N$, and $e_f^{<N}$, for dealing with $n<N$. We explain the expression
$e_f^{\geq N}$ first.

In our  expressions we use subexpressions which we defined before in section \ref{app:order}. These subexpression 
require some auxiliary variables, as detailed below. As a consequence, $e_f$ will be an expressions 
defined over an extended schema $\mathcal{S}'$. Hence, the instance $\I$ in the statement of the Proposition 
is  an instance $\I'$ of $\mathcal{S}'$ which
coincides with $\I$ on $\mathcal{S}$ and in which the auxiliary matrix variables are all instantiated with 
zero vectors or matrices, depending on their size.

Now, we specify the finite auxiliary variables involved in the \langfor expression. These arise
when computing the following \langfor expressions defined 

\begin{itemize}
	\item $e_{\mathsf{Prev}}(z,Z,z',Z')$, and expression over auxiliary variables $z$, $z'$, $Z$ and $Z'$ with 
	$\mathsf{size}(z)=\mathsf{size}(z')=\mathsf{size}(Z)=\alpha\times 1$ and 
	$\mathsf{size}(Z')=\alpha\times\alpha$. On input $\I'$ with 
	$\mathsf{mat}(z)=\mathsf{mat}(z')=\mathsf{mat}(Z)$ the zero column vector of dimension $n$, 
	and $\mathsf{mat}(Z')$ the zero $n\times n$ matrix,
	$\sem{e_{\mathsf{Prev}}}{\I'}$ returns the $n\times n$ matrix $\mathsf{Prev}$ such that 
	$$\mathsf{Prev}\cdot \mathbf{e}_i:=\begin{cases} 
	\mathbf{e}_{i-1} & \text{if $i>1$}\\
	\mathbf{0} & \text{if $i=1$}.
	\end{cases}
	$$
	\item $e_{\mathsf{Next}}(z,Z,z',Z')$, and expression over auxiliary variables $z$, $z'$, $Z$ and $Z'$ 
	with $\mathsf{size}(z)=\mathsf{size}(z')=\mathsf{size}(Z)=\alpha\times 1$ and 
	$\mathsf{size}(Z')=\alpha\times\alpha$. On input $\I'$ with 
	$\mathsf{mat}(z)=\mathsf{mat}(z')=\mathsf{mat}(Z)$ the zero column 
	vector of dimension $n$, and $\mathsf{mat}(Z')$ the zero $n\times n$ matrix,
	$\sem{e_{\mathsf{Next}}}{\I'}$ returns the $n\times n$ matrix $\mathsf{Next}$ such that 
	$$\mathsf{Next}\cdot \mathbf{e}_i:=\begin{cases} 
	\mathbf{e}_{i+1} & \text{if $i<n$}\\
	\mathbf{0} & \text{if $i=n$}.
	\end{cases}
	$$
	\item $\textsf{min}(v,z,Z,z',Z)$ with auxiliary variables $z$, $z'$, $Z$ and $Z'$ as before, 
	and $v$ is one of the (vector) variables in $\mathcal{M}$. For an $n\times 1$ vector $\mathbf{v}$, 
	on input $\I'[v\gets \mathbf{v}]$	$$\sem{\mathsf{min}}{\I'[v\gets\mathbf{v}]}:=\begin{cases} 1 & \text{if $\mathbf{v}=\mathbf{e}_1$}\\
		0 & \text{otherwise}.
		\end{cases}$$

	\item $\textsf{max}(v,z,Z,z',Z)$ with auxiliary variables $z$, $z'$, $Z$ and $Z'$ as before, and 
	and $v$ is one of the (vector) variables in $\mathcal{M}$. For an $n\times 1$ vector $\mathbf{v}$, 
	on input $\I'[v\gets \mathbf{v}]$
	
	$$\sem{\mathsf{max}}{\I'[v\gets\mathbf{v}]}:=\begin{cases} 1 & \text{if $\mathbf{v}=\mathbf{e}_{n}$}\\
		0 & \text{otherwise}.
		\end{cases}$$
	\item $e_{\mathsf{min}}(z,Z,z',Z',z'',Z'')$, an expressions with
	auxiliary variables $z$, $z'$, $z''$, $Z$, $Z'$ and $Z''$ with 
	$\mathsf{size}(z)=\mathsf{size}(z')=\mathsf{size}(z'')=\mathsf{size}(Z)=\mathsf{size}(Z'')=\alpha\times 1$ 
	and $\mathsf{size}(Z')=\alpha\times\alpha$. On input $\I'$ with 
	matrix variables instantiated with zero vectors (or matrix for $Z'$),
 	 $\sem{e_{\mathsf{min}}}{\I'}=\mathbf{e}_1$. 
	\item $e_{\mathsf{max}}(z,Z,z',Z',z'',Z'')$, an expressions with
	auxiliary variables $z$, $z'$, $z''$, $Z$, $Z'$ and $Z''$ with 
	$\mathsf{size}(z)=\mathsf{size}(z')=\mathsf{size}(z'')=\mathsf{size}(Z)=\mathsf{size}(Z'')=\alpha\times 1$ 
	and $\mathsf{size}(Z')=\alpha\times\alpha$. On input $\I'$ with 
	matrix variables instantiated with zero vectors (or matrix for $Z'$),
 	 $\sem{e_{\mathsf{max}}}{\I'}=\mathbf{e}_n$. 	 		
\end{itemize}
We thus see that we only need $z,z',z'',Z,Z',Z''$ as auxiliary variables and these can be re-used 
whenever $e_f$ calls these functions. From now one, we omit the auxiliary variables from the description 
of $e_f$.

Let us first define $e_f^{\geq N}$. Since we want to simulate $T$ we need to be able to check which 
transitions of $T$ can be applied based on a current configuration. More precisely,
suppose that we want to check whether $\delta(q_i,(a_1,\ldots,a_{\ell},b,c))$ is applicable, then we 
need to check whether $T$ is in state $q_i$, we can do this by checking 
$\mathsf{min}(Q_i)$, and whether the heads on the tapes read symbols $a_1,\ldots,a_{\ell},b,c$. We 
check the latter by the following expressions.
For the input tape $R_i$ we define
$$
\mathsf{test\_inp}^i_b:=\begin{cases}
(1-\mathsf{min}(1/2\cdot H_i))\cdot(1-\mathsf{max}(1/2\cdot H_i))\cdot(1- R_i^T\cdot H_i) & \text{if $b=0$}\\
(1-\mathsf{min}(1/2\cdot H_i))\cdot(1-\mathsf{max}(1/2\cdot H_i))\cdot(R_i^T\cdot H_i) & \text{if $b=1$}\\
\mathsf{min}(1/2\cdot H_i) & \text{if $b=\rhd$}\\
\mathsf{max}(1/2\cdot H_i) & \text{if $b=\lhd$},\\
\end{cases}
$$
which returns $1$ if and only if either $b\in\{0,1\}$ is the value in $\mathsf{mat}(R_i)$ at the 
position encoded by $\mathsf{mat}(H_i)$, or when $b=\rhd$ and $\mathsf{mat}(H_i)$ is the vector 
$(2,0,\ldots,0)\in\mathbb{R}^n$, or when $b=\lhd$ and $\mathsf{mat}(H_i)$ is the vector 
$(0,0,\ldots,2)\in\mathbb{R}^n$. Similarly, for the output tape we define
$$
\mathsf{test\_out}_b:=\begin{cases}
(1-\mathsf{min}(1/2\cdot H_O))\cdot(1- O^T\cdot H_O) & \text{if $b=0$}\\
(1-\mathsf{min}(1/2\cdot H_O))\cdot(O^T\cdot H_O) & \text{if $b=1$}\\
\mathsf{min}(1/2\cdot H_O) & \text{if $b=\rhd$},\\
\end{cases}
$$
and for the work tapes $W_1,\ldots,W_s$ we define
$$
\mathsf{test\_work}^i_b:=\begin{cases}
(1-\mathsf{min}(1/2\cdot H_{W_i}))\cdot(1- W_i^T\cdot H_{W_i})) & \text{if $b=0$}\\
(1-\mathsf{min}(1/2\cdot H_{W_i}))\cdot (W_i^T\cdot H_{W_i}) & \text{if $b=1$}\\
\mathsf{min}(1/2\cdot H_{W_i}) & \text{if $b=\rhd$ and $i=1$}.\\
\end{cases}
$$
We then combine all these expressions into a single expression for $q_i\in Q$, 
$a_1,\ldots,a_\ell,b,c\in\Gamma$:
$$
\mathsf{isconf}_{q_i,a_1,\ldots,a_\ell,b,c}:=
\mathsf{min}(Q_i)\cdot \left(\prod_{j=1}^{\ell} \mathsf{test\_inp}_{a_j}^j\right)
\cdot\left(\sum_{j=1}^s \mathsf{test\_work}_b^j\right)\cdot \mathsf{test\_out}_{c}.
$$
This expression will return $1$ if and only if the vectors representing the tapes, 
head positions and states are such that $\mathsf{Q_i}$ is the first canonical vector 
(and thus $T$ is in state $q_i$), the heads point to entries in the tape vectors storing 
the symbols $a_1,\ldots,a_{\ell}, b,c$ or they point to the first (or last for input tapes) 
positions but have value $2$ (when the symbols are $\rhd$ or $\lhd$). 

To ensure that at the beginning of the simulation of $T$ by $e_f^{\geq N}$ we correctly encode 
that we are in the initial configuration, we thus need to initialise all vectors 
$\mathsf{mat}(H_1),\mathsf{mat}(H_2),\ldots, \mathsf{mat}(H_\ell), \mathsf{mat}(H_{W_1}),\mathsf{mat}(H_O)$ 
with the vector $(2,0,0,\ldots,0)\in\mathbb{R}$ since all heads read the symbol $\rhd$. Similarly, 
we have to initialise $\mathsf{Q_1}$ with the first canonical vector since $T$ is in state $q_0$.

We furthermore need to be able to correctly adjust head positions. We do this by means of the predecessor 
and successor expressions described above. 
A consequence of our encoding is that we need to treat the border cases (corresponding to $\rhd$ and 
$\lhd$) differently. More specifically, for the input tapes $R_i$ and heads $H_i$ we define 
$$
\mathsf{move\_inp}^i_{\mathsf{d}}:=
\begin{cases}
2\times \mathsf{min}(H_i)\times H_i + 1/2\times\mathsf{max}(1/2\times H_i)\times H_i  + (1-\mathsf{min}(H_i))(1-\mathsf{max}(1/2 \times H_i))\times e_{\mathsf{Prev}}\cdot H_i  
& \text{if $\mathsf{d}=\leftarrow$}\\
2\times \mathsf{max}(H_i)\times H_i + 1/2\times\mathsf{min}(1/2\times H_i)\times H_i  + (1-\mathsf{min}(1/2\times H_i))(1-\mathsf{max}(H_i))\times e_{\mathsf{Next}}\cdot H_i  
 & \text{if $\mathsf{d}=\rightarrow$}\\
H_i & \text{if $\mathsf{d}=\sqcup$}. 
\end{cases}
$$
In other words, we shift to the previous (or next) canonical vector when $\mathsf{d}$ is $\leftarrow$ 
or $\rightarrow$, respectively, unless we need to move to or from the position that will hold $\rhd$ 
or $\lhd$. In those case we readjust $\mathsf{mat}(H_i)$ (which will either $(1,0,\ldots,0)$, $(2,0,\ldots,0)$, 
$(0,\ldots,0,1)$ or $(0,\ldots,0,2)$) by either dividing or multiplying with $2$. In this way we can 
correctly infer whether or not the head points to the begin and end markers. For the output tape we 
proceed in a similar way, but only taking into account the begin marker and recall that we do not have 
moves to the left:
$$
\mathsf{move\_outp}_{\mathsf{d}}:=
\begin{cases}
1/2\times\mathsf{min}(1/2\times H_O)\times H_O  + (1-\mathsf{min}(1/2\times H_O))\times e_{\mathsf{Next}}\times H_O  
 & \text{if $\mathsf{d}=\rightarrow$}\\
H_O & \text{if $\mathsf{d}=\sqcup$}. 
\end{cases}
$$
Since we represent the work tape by $s$ vectors $W_1,\ldots,W_s$ we need to ensure that only one 
of the head vectors $H_{W_i}$ has a non-zero value and that by moving left or right, we need to 
appropriately update the right head vector. We do this as follows. We first consider the work tapes 
$W_i$ for $i\neq 1,s$ and define
$$
\mathsf{move\_work}^i_{\mathsf{d}}:=
\begin{cases}
	-\mathsf{min}(H_{W_i})\times H_{W_i} + (1-\mathsf{min}(H_{W_i}))\times e_{\mathsf{Prev}}\cdot H_{W_i} + \mathsf{min}(H_{W_{i+1}})\times e_{\mathsf{max}} & \text{if $\mathsf{d}=\leftarrow$}\\
		-\mathsf{max}(H_{W_i})\times H_{W_i} + (1-\mathsf{max}(H_{W_i}))\times e_{\mathsf{Next}}\cdot H_{W_i} + \mathsf{max}(H_{W_{i-1}})\times e_{\mathsf{min}} & \text{if $\mathsf{d}=\rightarrow$}\\
	H_{W_i} & \text{if $\mathsf{d}=\sqcup$}. 	
\end{cases}
$$
In other words, we set the $H_{W_i}$ to zero when a move brings us to either $W_{i-1}$ or $W_{i+1}$, we
move the successor or predecessor when staying within $W_i$, or initialise $H_{W_i}$ with the first or 
last canonical vector when moving from $W_{i-1}$ to $W_i$ (right move) or from $W_{i+1}$ to $W_i$ (left move).
For $i=s$ we can ignore the parts in the previous expression that involve $W_{s+1}$ (which does not exist):
$$
\mathsf{move\_work}^s_{\mathsf{d}}:=
\begin{cases}
	-\mathsf{min}(H_{W_s})\times H_{W_i} + (1-\mathsf{min}(H_{W_s}))\times e_{\mathsf{Prev}}\cdot H_{W_s}  & \text{if $\mathsf{d}=\leftarrow$}\\
		-\mathsf{max}(H_{W_s}) \times H_{W_s} + (1-\mathsf{max}(H_{W_s}))\times e_{\mathsf{Next}}\cdot H_{W_s} + \mathsf{max}(H_{W_{s-1}})\times e_{\mathsf{min}} & \text{if $\mathsf{d}=\rightarrow$}\\
	H_{W_s} & \text{if $\mathsf{d}=\sqcup$}. 	
\end{cases}
$$
For $i=1$, we can ignore the part involving $W_{0}$ (which does not exist) but have to take $\rhd$ 
into account:
$$
\mathsf{move\_work}^1_{\mathsf{d}}:=
\begin{cases}
	2\times \mathsf{min}(H_{W_1})\times H_{W_i} + (1-\mathsf{min}(H_{W_1}))\times e_{\mathsf{Prev}}\cdot H_{W_1} + \mathsf{min}(H_{W_{2}})\times e_{\mathsf{max}} & \text{if $\mathsf{d}=\leftarrow$}\\
		1/2\times\mathsf{min}(1/2\times H_{W_1})\times H_{W_1} + (1-\mathsf{max}(1/2\times H_{W_1}))\times e_{\mathsf{Next}}\cdot H_{W_1}  & \text{if $\mathsf{d}=\rightarrow$}\\
	H_{W_1} & \text{if $\mathsf{d}=\sqcup$}. 	
\end{cases}
$$
A final ingredient for defining $e_f^{\geq N}$ are expressions which update the work and output tape.
To define these expression, we need the position and symbol to put on the tape. For the output tape we define
$$
\mathsf{write\_outp}_b:=\begin{cases}
\mathsf{min}(1/2\times H_O)\times O & \text{if $b=\rhd$}\\
(1-\mathsf{min}(1/2\times H_O))\times\left((1-O^T\cdot H_O)\times O + (O^T\cdot H_O)\times (O-H_O)\right) &\text{if $b=0$}\\
(1-\mathsf{min}(1/2\times H_O))\times\left((1-O^T\cdot H_O)\times (O+H_O) + (O^T\cdot H_O)\times O\right) &\text{if $b=1$}\\
\end{cases}
$$
and similarly for the work tapes $i\neq 1$:
$$
\mathsf{write\_work}_b^i:=\begin{cases}
W_i & \text{if $b=\rhd$}\\
(1-W_i^T\cdot H_{W_i})\times W_i + (W_i^T\cdot H_{W_i})\times (W_i-H_{W_i}) &\text{if $b=0$}\\
(1-W_i^T\cdot H_{W_i})\times (W_i+H_{W_i}) + (W_i^T\cdot H_{W_i})\times W_i &\text{if $b=1$},
\end{cases}
$$
and for  $W_1$ we have to take care again of the begin marker:
$$
\mathsf{write\_work}_b^1:=\begin{cases}
\mathsf{min}(1/2\times H_{W_1})\times W_1 & \text{if $b=\rhd$}\\
(1-\mathsf{min}(1/2\times H_{W_1})\times\left((1-W_1^T\cdot H_{W_1})\times W_1 + (W_1^T\cdot H_{W_1})\times (W_1-H_{W_1})\right) &\text{if $b=0$}\\
(1-\mathsf{min}(1/2\times H_{W_1})\times\left((1-W_1^T\cdot H_{W_1})\times (W_1+H_{W_1}) + (W_1^T\cdot H_{W_1})\times W_1\right) &\text{if $b=1$}.
\end{cases}
$$

%
%

We are now finally ready to define $e_f^{\geq N}$ :
\begin{multline*}
e_f^{\geq N}:= \mathsf{for\,} v_1,\ldots,v_{k},Q_1,\ldots,Q_m,H_1,\ldots,H_\ell,W_1,\ldots,W_s, H_{W_1},\ldots,H_{W_s},O,H_O . \\
(e_{Q_1},\ldots,e_{Q_m},e_{H_1},\ldots,e_{H_\ell},e_{W_1},\ldots,e_{W_s},e_{H_{W_1}},\ldots,e_{H_{W_s}},e_{O}, e_{H_O}).
\end{multline*}
We refer to section \ref{app:def} for the definition of this form of the for-loop. The expressions used are (we use $\star$ below to mark irrelevant information in the transitions):
 \allowdisplaybreaks
\begin{align*}
	e_{Q_1}&:=\left(\prod_{j=1}^{k} \textsf{min}(v_i)\right)\times e_{\mathsf{min}}
	+ \sum_{\substack{(q_i,a_1,\ldots,a_\ell,b,c)\\
	\Delta(q_i,a_1,\ldots,a_\ell,b,c)=(q_1,\star)}} \!\!\!\!\!\!\!\!\! \mathsf{isconf}_{q_i,a_1,\ldots,a_\ell,b,c}\times e_{\mathsf{min}} \\
	e_{Q_j}&:=\sum_{\substack{(q_i,a_1,\ldots,a_\ell,b,c)\\
	\Delta(q_i,a_1,\ldots,a_\ell,b,c)=(q_j,\star)}} \!\!\!\!\!\!\!\!\! \mathsf{isconf}_{q_i,a_1,\ldots,a_\ell,b,c}\times e_{\mathsf{min}}
	 \quad \text{for $j\neq 1$}\\
	e_{H_i}&:=2\left(\prod_{j=1}^{k} \textsf{min}(v_i)\right)\times e_{\mathsf{min}}
	+\sum_{\substack{(q,a_1,\ldots,a_\ell,b,d)\\
	\Delta(q,a_1,\ldots,a_\ell,b,c)=(\star,\mathsf{d_i},\star)}}\!\!\!\!\!\!\!\!\! \mathsf{isconf}_{q,a_1,\ldots,a_\ell,b,c}\times\mathsf{move\_inp}^i_{\mathsf{d}_i}\\
	e_{H_{W_i}}&:=2\left(\prod_{j=1}^{k} \textsf{min}(v_i)\right)\times e_{\mathsf{min}}
	+\sum_{\substack{(q,a_1,\ldots,a_\ell,b,d)\\
	\Delta(q,a_1,\ldots,a_\ell,b,c)=(\star,\mathsf{d_{\ell+1}},\star)}}\!\!\!\!\!\!\!\!\! \mathsf{isconf}_{q,a_1,\ldots,a_\ell,b,c}\times\mathsf{move\_work}_{\mathsf{d}_{\ell+1}}^i\\
	e_{H_O}&:=2\left(\prod_{j=1}^{k} \textsf{min}(v_i)\right)\times e_{\mathsf{min}}
	+\sum_{\substack{(q,a_1,\ldots,a_\ell,b,d)\\
	\Delta(q,a_1,\ldots,a_\ell,b,c)=(\star,\mathsf{d}_{\ell+2})}}\!\!\!\!\!\!\!\!\! \mathsf{isconf}_{q,a_1,\ldots,a_\ell,b,c}\times\mathsf{move\_outp}_{\mathsf{d}_{\ell+2}}\\
	e_{W_i}&:=\sum_{\substack{(q,a_1,\ldots,a_\ell,b,d)\\
	\Delta(q,a_1,\ldots,a_\ell,b,c)=(\star,b',c',\star)}}\!\!\!\!\!\!\!\!\! \mathsf{isconf}_{q,a_1,\ldots,a_\ell,b,c}\times\mathsf{write\_work}_{b'}^i\\
	e_{O}&:=\sum_{\substack{(q,a_1,\ldots,a_\ell,b,d)\\
	\Delta(q,a_1,\ldots,a_\ell,b,c)=(\star,b',c',\star)}}\!\!\!\!\!\!\!\!\! \mathsf{isconf}_{q,a_1,\ldots,a_\ell,b,c}\times\mathsf{write\_outp}_{c'}.
\end{align*}

The correctness of $e_f^{\geq N}$ should be clear from the construction (one can formally verify this by
induction on the number of iterations). We next explain how the border cases $n<N$ can be dealt with.
For each $n<N$ and every possible input words
$w_1,\ldots,w_\ell$ of size $n$, we define a \langfor expression which checks whether
$\mathsf{mat}(R_i)=\mathsf{vec}(w_i)$ for each $i\in[\ell]$. This can be easily done since $n$ 
can be regarded as a constant. For example, to check whether $\mathsf{mat}(R_i)=[0,1,1]^T$ we simply write
$$
(1- R_i^T\cdot e_{\mathsf{min}})\times (R_i^T\cdot e_{\mathsf{Next}}\cdot e_{\mathsf{min}})\times (1- R_i^T\cdot e_{\mathsf{Next}}\cdot e_{\mathsf{Next}}\cdot e_{\mathsf{min}})\times (1- e_{\ones}(R_i)^T\cdot e_{\mathsf{Next}}\cdot e_{\mathsf{Next}}\cdot e_{\mathsf{Next}}\cdot e_{\mathsf{min}})
$$
which will evaluate to $1$ if and only if $\mathsf{mat}(R_i)=[0,1,1]^T$. We note that the final factor is in 
place to check that the dimension of $\mathsf{mat}(R_i)$ is three.
  We denote by
$e_{n,w}^i$ the expression which evaluates to $1$ if and only if $\mathsf{mat}(R_i)=\mathsf{vec}(w)$
for $|w|=n$.
We can similarly
write any word $w$ of fixed size in the matrix variable $O$. For example, suppose that $w=101$
then we write 
$$
O+ e_{\mathsf{min}}+  e_{\mathsf{Next}}\cdot e_{\mathsf{Next}}\cdot e_{\mathsf{min}}.
$$
We write $e_{n,w}$ be the expression which write $w$ of size $|w|=n$ in matrix variable $O$.
Then, the expressions
$$
e_{n,w_1,\ldots,w_n,w}:=e_{n,w_1}^1\cdot\cdots\cdot e_{nw_{\ell}}^\ell\cdot e_{n,w}
$$
will write $w$ in $O$ if and only if $\mathsf{mat}(R_i)=\mathsf{vec}(w_i)$ for $i\in[\ell]$.
We now simply take the disjunction over all words 
$w_1,\ldots,w_\ell\in\Sigma^n$ and $w=f_n(w_1,\ldots,w_\ell)\in\Sigma^n$:
$$
e_n:=\sum_{w_1,\ldots,w_\ell\in\Sigma^n} e_{n,w_1,\ldots,w_\ell,f_n(w_1,\ldots,w_\ell)},
$$
which correctly evaluates $f_n$. We next take a further disjunction by letting ranging from 
$n=0,\ldots, N-1$:
$$
e_f^{<N}:=\sum_{n=0}^{N-1} e_n
$$
Since every possible input is covered and only a unique expression 
$ e_{n,w_1,\ldots,w_\ell,f_n(w_1,\ldots,w_\ell)}$ will be triggered $e_f^{<N}$ will correctly
evaluate $f$ on inputs smaller than $N$.

Our final expression $e_f$ is now given by
$$
e_f:=e_f^{<N} + \mathsf{dim\_is\_greater\_than_N}\times e_f^{\geq N}
$$
where $\mathsf{dim\_is\_greater\_than_N}$ is the expression
$e_{\ones}(R_i)^T\cdot\underbrace{e_{\mathsf{Next}}\cdot\cdots\cdot e_{\mathsf{Next}}}_{\text{$N$ times}}$ 
which will evaluate to $1$ if an only if the input dimension is larger or equal than $N$.
\end{proof}

%% file: sections/app-circuit-result.tex
\newtheorem*{CIRCUITTOML}{Theorem~\ref{th-circuits-ml}}

We prove theorem \ref{th-circuits-ml}:
\begin{CIRCUITTOML}
  For any uniform arithmetic circuit family $\{\Phi_n\mid n=1,2,\ldots\}$ of logarithmic depth there is a \langfor schema $\Sch$ and an expression $e_\Phi$ using a matrix variable $v$, with $\ttype(v)=(\alpha,1)$ and $\ttype(e) = (1,1)$, such that for any input $a_1,\ldots ,a_n$ to the circuit $\Phi_n$:
  \begin{itemize}
  \item If $\I = (\dom,\conc)$ is a \lang\ instance such that $\dom(\alpha) = n$ and $\conc(v) = [a_1 \ldots a_n]^T$
  \item Then $\sem{e}{\I} = \Phi_n(a_1,\ldots ,a_n)$.
  \end{itemize}
\end{CIRCUITTOML}

\begin{proof}

Let $\Phi_n$ be a circuit with $n$ input gates and such that it can be computed by a $L-$uniform arithmetic circuit of log-depth. Each gate of the circuit that encodes $f$ has an $\texttt{id}\in\lbrace 0,1 \rbrace^n$. From now on, when we write $g$ for a gate of the circuit, we mean the $\texttt{id}$ encoding $g$.
Let $n^k$ be a polynomial such that the number of wires $W(n)\leq n^k$ for $n$ big enough. Further, we assume that $2W(n)\leq n^k$. We need this because the for-matlang simulation of the circuit is in a depth first search way, so $2W(n)$ wires will be traversed.
Then we have that:
\begin{itemize}
	\item the number of gates is bounded by $n^k$.
	\item we need at most $k\log (n)$ bits to store the $id$ of a gate.
	\item the depth of the circuit is at most $k'\log (n)$ for some $k'$.
\end{itemize}

So, let $n_0$ and $k$ such that $\forall n\geq n_0:$

\begin{align*}[right=\empheqrbrace (\star)]
    2W(n)&\leq n^k \\
	k \ceil{\log (n)} &\leq n-3 \\
	k' \ceil{\log(n)} &\leq n
\end{align*}

We know $n_0$ and $k$ exist. Let $n\geq n_0$. Towards the end, we will deal with the case when $n<n_0$.

Let $g$ be a gate. The children of $g$ are denoted by $g_1,\ldots, g_l$.
\begin{center}
\begin{tikzpicture}[level distance=1.5cm,
  level 1/.style={sibling distance=1.5cm},
  every node/.style = {
  	shape=circle,
    draw,
    align=center,
    top color=white,
    bottom color=white
    }]
  \node {\( g \)}
    child {node { \( g_1 \) }}
    child {node { \( \cdots \) }}
    child {node { \( g_l \) }};
\end{tikzpicture}
\end{center}

For example, a circuit that encodes the function $f(a_1,a_2,a_3,a_4)=a_1a_2 +a_3a_4$ is 

\begin{center}
\begin{tikzpicture}[level distance=1.5cm,
  level 1/.style={sibling distance=3cm},
  level 2/.style={sibling distance=1.5cm},
  every node/.style = {
  	shape=circle,
    draw,
    align=center,
    top color=white,
    bottom color=white
    }
  ]
  \node { \( + \) }
    child { node { \( \times \) }
      child { node { \( a_1 \) } }
      child {node { \( a_2 \) } }
    }
    child { node { \( \times \) }
      child { node { \( a_3 \) } }
      child { node { \( a_4 \) } }
    };
\end{tikzpicture}
\end{center}

We can simulate the polynomial $x^2+xy$ by doing $f(A)$ where $A=[x \hspace{1ex} x \hspace{1ex} x \hspace{1ex} y]^T$. The main idea is to traverse the circuit top down in a depth first search way and store visited gates in a stack and its corresponding current values in another stack, and aggregate in the iterations according to the gate type.

For a stack $S$, the operations are standard:

\begin{itemize}
	\item $\push{S}{s}$: pushes $s$ into $S$.
	\item $\pop{S}$: pops the top element.
	\item $\getsize{S}$: the length of the stack.
	\item $\gettop{S}$: the top element in the stack.
\end{itemize}

For the pseudo-code, $\cG$ and $\cV$ denote stacks of gates and values, respectively. The property that holds during the simulation is that the value in $\cV[i]$ is the value that $\cG[i]$ currently outputs. The algorithm ends with $\cG=\left[ g_{\texttt{root}}\right]$ and $\cV=\left[ v_{\texttt{root}}\right]$ after traversing the circuit, and returns $v_{\texttt{root}}$

During the evaluation algorithm there will be two possible configurations of $\cG$ and $\cV$.

\begin{enumerate}
	\item $\getsize{\cG} = \getsize{\cV} + 1$: this means that $\gettop{\cG}$ is a gate that we visit for the first time and we need to initialize its value.
	
	\item $\getsize{\cG} = \getsize{\cV}$: here $\gettop{\cV}$ is the value of evaluating the circuit in gate $\gettop{\cG}$. Therefore, we need to aggregate the value $\gettop{\cV}$ to the parent gate of $g$.
\end{enumerate}

We assume the circuit has input gates, $+, \times$-gates and allow constant $1$-gate.

The idea is to traverse the circuit top down in a depth first search way. For example, in the circuit $f(a_1,a_2,a_3,a_4)=a_1a_2 +a_3a_4$ above, we would initialize the output gate value as $0$ because it is a $+$ gate, so $\cG=\lbrace +\rbrace$, $\cV=\lbrace 0\rbrace$. Then stack the left $\times$ gate to $\cG$, stack its initial value (i.e. $1$) to $\cV$. Now stack $a_1$ to $\cG$ and its value (i.e. $a_1$) to $\cV$. Since we are on an input gate we pop the gate and value pair off of $\cG$ and $\cV$ respectively, aggregate $a_1$ to $\gettop{\cV}$ and continue by stacking the $a_2$ gate to $\cG$. We pop $a_2$ off of $\cV$ (and its gate off of $\cG$) and aggregate its value to $\gettop{\cV}$. We pop and aggregate the value of the left $\times$ gate to $\gettop{\cV}$ (the root value). Then continue with the right $\times$ gate branch similarly.

For the pseudo-code, we supply ourselves with the following functions:

\begin{itemize}
	\item[--] $\isplus{g}$: true if and only if $g$ is a $+$-gate.
	\item[--] $\isprod{g}$: true if and only if $g$ is a $\times$-gate.
	\item[--] $\isone{g}$: true if and only if $g$ is a $1$-gate.
	\item[--] $\isinput{g}$: true if and only if $g$ is an input gate.
	\item[--] $\getfirst{g}$: outputs the first child of $g$.
	\item[--] $\getinput{g}$: outputs $A[i]$ when $g$ is the $i$-th input.
	\item[--] $\isnotlast{g_1}{g_2}$: true if and only if $g_2$ is not the last child gate of $g_1$.
	\item[--] $\nextgate{g_1}{g_2}$: outputs the next child gate of $g_1$ after $g_2$.
	\item[--] $\getroot$: outputs the root gate of the circuit.
\end{itemize}

The corresponding $\lbrace 0,1 \rbrace^n\rightarrow\lbrace 0,1 \rbrace^n$ functions are:

\begin{itemize}
	\item[--] $\isplus{g}$: $1$ if and only if $g$ is a $+$-gate.
	\item[--] $\isprod{g}$: $1$ if and only if $g$ is a $\times$-gate.
	\item[--] $\isone{g}$: $1$ if and only if $g$ is a $1$-gate.
	\item[--] $\isinput{g}$: $1$ if and only if $g$ is an input gate.
	\item[--] $\getfirst{g}$: outputs the $\texttt{id}$ of the first child of $g$.
	\item[--] $\getinput{g}$: outputs canonical vector $b_i$, where the $i$-th input gate of $\Phi_n$ is encoded by $g$.
	\item[--] $\isnotlast{g_1}{g_2}$: $1$ if and only if $g_2$ is not the last child gate of $g_1$.
	\item[--] $\nextgate{g_1}{g_2}$: outputs the $\texttt{id}$ of the next child gate of $g_1$ after $g_2$.
	\item[--] $\getroot$: outputs the $\texttt{id}$ of the root gate of the circuit.
\end{itemize}

The previous functions are all definable by an $L$-transducer and can be defined from the $L$-transducer of $f$. Then, by proposition \ref{prop:transducer}, for each of these functions there is a \langfor expression that simulates them.

Now, we give the pseudo-code of the top-down evaluation. We define the functions $Initialize$ (algorithm \ref{alg:init_code}), $Aggregate$ (algorithm \ref{alg:agg_code}) and $Evaluate$ (algorithm \ref{alg:eval_code}). The main algorithm is $Evaluate$.

\begin{algorithm}
\caption{Initialize (pseudo-code)}\label{alg:init_code}
\begin{algorithmic}[1]
\Function{Initialize}{$\cG, \cV, A$}\Comment{The stacks and input. Here, $\getsize{\cG} =  \getsize{\cV} + 1$}
	\If{$\isplus{\gettop{\cG}}$}
		\State $\push{\cV}{0}$
		\State $\push{\cG}{\getfirst{\gettop{\cG}}}$
	\ElsIf{$\isprod{\gettop{\cG}}$}
		\State $\push{\cV}{1}$
		\State $\push{\cG}{\getfirst{\gettop{\cG}}}$
	\ElsIf{$\isone{\gettop{\cG}}$}
		\State $\push{\cV}{1}$
	\ElsIf{$\isinput{\gettop{\cG}}$}
		\State $\push{\cV}{A\left[ \getinput{\gettop{\cG}} \right]}$
	\EndIf
	\State \textbf{return} $\cG, \cV$
\EndFunction
\end{algorithmic}
\end{algorithm}

\begin{algorithm}
\caption{Aggregate (pseudo-code)}\label{alg:agg_code}
\begin{algorithmic}[1]
\Function{Aggregate}{$\cG, \cV$}\Comment{Here, $\getsize{\cG} =  \getsize{\cV}$}
	\State $g = \pop{\cG}$
	\State $v = \pop{\cV}$
	\If{$\isplus{\gettop{\cG}}$}
		\State $\gettop{\cV} = \gettop{\cV} + v$
	\ElsIf{$\isprod{\gettop{\cG}}$}
		\State $\gettop{\cV} = \gettop{\cV} \cdot v$
	\EndIf
	\If{$\isnotlast{\gettop{\cG}}{g}$}
		\State $\push{\cG}{\nextgate{\gettop{\cG}}{g}}$
	\EndIf
	\State \textbf{return} $\cG, \cV$
\EndFunction
\end{algorithmic}
\end{algorithm}

\begin{algorithm}
\caption{Evaluate (pseudo-code)}\label{alg:eval_code}
\begin{algorithmic}[1]
\Function{Evaluate}{$A$}\Comment{Input $ n\times 1$ vector $A$. Here, $\cG$ and $\cV$ are empty}
	\State $\push{\cG}{\getroot}$
	\While{$\getsize{\cG}\neq 1$ or $\getsize{\cV}\neq 1$}
		\If{$\getsize{\cG}\neq \getsize{V}$}
			\State $(\cG,\cV) := \texttt{Initialize}(\cG,\cV,A)$
		\Else
			\State $(\cG,\cV):= \texttt{Aggregate}(\cG,\cV)$
		\EndIf
	\EndWhile
	\State \textbf{return} $\gettop{\cV}$
\EndFunction
\end{algorithmic}
\end{algorithm}

The $Evaluate$ algorithm gives us the output of the circuit. Note that after each iteration it either holds that $\getsize{\cG} =  \getsize{\cV} + 1$ or $\getsize{\cG} =  \getsize{\cV}$. Furthermore, when we start we have $\getsize{\cG}=1$ and $\getsize{\cV}=0$. The condition $\getsize{\cG}= 1$ and $\getsize{\cV}=1$ holds only when we have traversed all the circuit, and the value in $\gettop{\cV}$ is the value that the root of the circuit outputs after its computation.

Next, we show how to encode this algorithm in $\langfor$.

Let $n_0\in\mathbb{N}$ be big enough for ($\star$) to hold and let $n\geq k$. Hence, the number of gates (values) is bounded by $n^k$ and we need $k\log (n)$ bits to encode the id of each gate.

To simulate the two stacks $\cG$ and $\cV$ we keep a matrix $X$ of dimensions $n \times n$.

\begin{itemize}
	\item Column $n$ will store a canonical vector that marks the top of stack $V$ (values).
	\item Column $n-1$ will store a canonical vector that marks the top of stack $G$ (gates).
	\item Column $n-2$ is the stack of values where $X[1, n-2]$ is the bottom of the stack.
	\item Columns $1$ to $n-3$ are the stack of gates.
\end{itemize}

If we have $j$ gates in the stack and currently $\getsize{\cG}=\getsize{\cV}$ then $X$ would look like:

\[
X = \begin{bmatrix}
    \texttt{id}_1 & v_1 & 0 & 0 \\
    \texttt{id}_2 & v_2 & 0 & 0 \\
    \vdots & \vdots & \vdots & \vdots \\
    \texttt{id}_j & v_j & 1 & 1 \\
    0 & 0 & 0 & 0 \\
    \vdots & \vdots & \vdots & \vdots \\
     0 & 0 & 0 & 0
\end{bmatrix}.
\]

Since $n\geq n_0$, $(\star)$ holds and thus we never use more than $n-3$ bits to encode an $\texttt{id}$. Also, $j\leq n$ given that we never keep more gates than the depth of the tree. As a consequence, we never keep more than $n$ values either.

\thomas{The following remark is necessary because of dimensions (check typing of START). I don't know if the reverse part is correct, but it makes sense to me since that way zeroes to the right actually mean nothing in the binary number (if it is reversed).}

An important detail is that the $\texttt{ids}$ of the gates are encoded as $\texttt{id}_r000$ for it to have dimension $n$, where $\texttt{id}_r$ is the corresponding binary number in reverse.

We make a series of definitions to make the notation more clear. Refer to section \ref{app:order} for more information
about these expressions.

Let $b_i$ be the $i$-th canonical vector. $\mathsf{Next}$ and $\mathsf{Prev}$ denote the successor and predecessor matrices respectively, such that

\[
  			\mathsf{Next}\cdot b_i=\begin{cases}
               b_{i+1} \text{ if } i\leq n \\
               \mathbf{0} \text{ otherwise }
            \end{cases}
\]

\[
  			\mathsf{Prev}\cdot b_i=\begin{cases}
               b_{i-1} \text{ if } i\geq n \\
               \mathbf{0} \text{ otherwise }
            \end{cases}
\]

We write expressions $e_{\mathsf{min}}$ for the first canonical vector and $e_{\mathsf{max}}$ for the last canonical vector. For any $i$ we write 
\begin{align*}
	e_{\mathsf{min}+\mathsf{i}} &= \mathsf{Next}^i\cdot e_{\mathsf{min}} \\
	e_{\mathsf{max}+\mathsf{i}} &= \mathsf{Prev}^i\cdot e_{\mathsf{max}}
\end{align*}

We use the extra $\lbrace 0,1 \rbrace^n\rightarrow\lbrace 0,1 \rbrace^n$ functions that have a $\langfor$ translation:

\[
  			\mathsf{min}(e)=\begin{cases}
               1 \text{ if } e=e_{\mathsf{min}} \\
               0 \text{ otherwise }
             \end{cases}
\]

\[
  			\mathsf{max}(e)=\begin{cases}
               1 \text{ if } e=e_{\mathsf{max}} \\
               0 \text{ otherwise }
             \end{cases}
\]

\[
  			\mathsf{succ}(b_i,b_j)=\begin{cases}
               1 \text{ if } i\leq j \\
               0 \text{ otherwise }
             \end{cases}
\]

When used in $\langfor$ these functions output $[0]$ and $[1]$.

Now 
\begin{align*}
	e_{V}&:=e_{\mathsf{max}-2} \\
	e_{G_{top}}&:=e_{\mathsf{max}-1} \\
	e_{V_{top}}&:=e_{\mathsf{max}}
\end{align*}

For a canonical vector, let $$\Iden{b_i}:=\ssum v. \mathsf{succ}(v,b_i)\cdot (v\cdot v^T).$$ This matrix has ones in the diagonal up to position $i$ marked by $e_{i}$. We define the following sub-matrices of $X$:
\begin{align*}
	V_{top} &:= X\cdot e_{V_{top}} \\
	V &:= \Iden{V_{top}} \cdot X \cdot e_V \\
 	G_{top} &:=X\cdot e_{G_{top}} \\
 	G &:= \Iden{G_{top}}\cdot X \cdot \Iden{e_{\mathsf{max}-3}}
\end{align*}

For example, if we are in a step where $\getsize{\cG}=\getsize{\cV} + 1$ then

\[
X = \begin{bmatrix}
    \texttt{id}_1 & v_1 & 0 & 0 \\
    \texttt{id}_2 & v_2 & 0 & 0 \\
    \vdots & \vdots & \vdots & \vdots \\
    \texttt{id}_{j-1} & v_{j-1} & 0 & 1 \\
    \texttt{id}_j & 0 & 1 & 0 \\
    0 & 0 & 0 & 0 \\
    \vdots & \vdots & \vdots & \vdots \\
     0 & 0 & 0 & 0
\end{bmatrix}, 
G = \begin{bmatrix}
    \texttt{id}_1 & 0 & 0 & 0 \\
    \texttt{id}_2 & 0 & 0 & 0\\
    \vdots & \vdots & \vdots & \vdots  \\
    \texttt{id}_{j-1} & 0 & 0 & 0\\
    \texttt{id}_j & 0 & 0 & 0\\
    0 & 0 & 0 & 0 \\
    \vdots & \vdots & \vdots & \vdots \\
     0 & 0 & 0 & 0
\end{bmatrix}, 
V = \begin{bmatrix}
    v_1  \\
    v_2 \\
    \vdots   \\
    v_{j-1} \\
    0 \\
    0 \\
    \vdots \\
     0 
\end{bmatrix}, 
G_{top} = \begin{bmatrix}
    0  \\
    0 \\
    \vdots   \\
    0 \\
    1 \\
    0 \\
    \vdots \\
     0 
\end{bmatrix}, 
V_{top} = \begin{bmatrix}
    0  \\
    0 \\
    \vdots   \\
    1 \\
    0 \\
    0 \\
    \vdots \\
     0 
\end{bmatrix}
\]

Here, $V$ is a vector encoding the stack of values in $X$ and $G$ is a matrix encoding the stack of gates in $X$.
Note that what is \textit{over} the top of the stacks is always set to zero due to $\Iden{G_{top}}$ and $\Iden{V_{top}}$.
Also, note that $G$ is of the same size as $X$. We sometimes omit the zeroes due to simplicity.

To set the initial state (algorithm \ref{alg:eval_code} line 2) we define the $\langfor$ expression: $$\text{START}:= e_{\mathsf{min}}\cdot \getroot^T + e_{\mathsf{min}}\cdot e_{G_{top}}^T.$$
For the initialize step, we define the $\langfor$ expressions: INIT${\_}$PLUS (algorithm \ref{alg:init_code}, lines 2, 3, 4), INIT${\_}$PROD (algorithm \ref{alg:init_code}, lines 5, 6, 7), CONST (algorithm \ref{alg:init_code}, lines 8, 9) and INPUT (algorithm \ref{alg:init_code}, lines 10, 11):

\begin{align*}
	\text{INIT{\_}PLUS} &:= \isplus{G^T\cdot G_{top}}\times \left[ G + (\mathsf{Next}\cdot G_{top}) \cdot \getfirst{G^T\cdot G_{top}}^T  + \mathsf{Next}\cdot G_{top}\cdot e_{G_{top}}^T +V\cdot e_{V}^T + \mathsf{Next}\cdot V_{top}\cdot e_{V_{top}}^T \right] \\
	\text{INIT{\_}PROD} &:= \isprod{G^T\cdot G_{top}}\times \left[ G + (\mathsf{Next}\cdot G_{top}) \cdot \getfirst{G^T\cdot G_{top}}^T + \mathsf{Next}\cdot G_{top}\cdot e_{G_{top}}^T +(V + \mathsf{Next}\cdot v_{top})\cdot e_{V}^T + \mathsf{Next}\cdot V_{top}\cdot e_{V_{top}}^T \right] \\
	\text{CONST} &:= \isone{G^T\cdot G_{top}}\times \left[ G + (V + \mathsf{Next}\cdot V_{top})\cdot e_{V}^T + \mathsf{Next}\cdot V_{top}\cdot e_{V_{top}}^T \right] \\
	\text{INPUT} &:= \isinput{G^T\cdot G_{top}}\times \left[ G + \left(V + \left( v^T \cdot \mathsf{Next}\cdot V_{top} \cdot \getinput{G^T\cdot G_{top}}^T \right)\right)\cdot e_{V}^T + \mathsf{Next}\cdot V_{top}\cdot e_{V_{top}}^T \right]
\end{align*} 
Where $v$ is the matrix variable stated in the theorem, the one associated with the input $A$ of the circuit.
Here, $G^T\cdot G_{top}$ is to get the current id in the top of the stack. In INIT${\_}$PLUS we get the current stack $G$, we add $\mathsf{Next}\cdot G_{top} \cdot \getfirst{G^T\cdot G_{top}}^T$ which is an $n\times n$ matrix with the first child of $G^T\cdot G_{top}$ in the next row. Then $\mathsf{Next}\cdot G_{top}\cdot e_{G_{top}}^T$ adds $\mathsf{Next}\cdot G_{top}$ to the $n-1$ column to mark the gate we added as the top. Next, we do the same with the values by adding $V\cdot e_{V} + \mathsf{Next}\cdot V_{top}\cdot e_{V_{top}}^T$.

The $\langfor$ expression equivalent to algorithm \ref{alg:init_code} is $$\text{INIT}:=\text{INIT{\_}PLUS}+\text{INIT{\_}PROD}+\text{CONST}+\text{INPUT}.$$

The idea is to return the matrix for the next iteration. Recall that here $\getsize{\cG}=\getsize{\cV} + 1$. So, when the operation is INPUT or CONST, if we start with

\[
\begin{bmatrix}
    \texttt{id}_1 & v_1 & 0 & 0 \\
    \texttt{id}_2 & v_2 & 0 & 0 \\
    \vdots & \vdots & \vdots & \vdots \\
    \texttt{id}_{j-1} & v_{j-1} & 0 & 1 \\
    \texttt{id}_j & 0 & 1 & 0 \\
    0 & 0 & 0 & 0 \\
    \vdots & \vdots & \vdots & \vdots \\
     0 & 0 & 0 & 0
\end{bmatrix}, \text{ then we return }
\begin{bmatrix}
    \texttt{id}_1 & v_1 & 0 & 0 \\
    \texttt{id}_2 & v_2 & 0 & 0 \\
    \vdots & \vdots & \vdots & \vdots \\
    \texttt{id}_{j-1} & v_{j-1} & 0 & 0 \\
    \texttt{id}_j & v_j & 1 & 1 \\
    0 & 0 & 0 & 0 \\
    \vdots & \vdots & \vdots & \vdots \\
     0 & 0 & 0 & 0
\end{bmatrix}.
\]

When the operation is INIT{\_}PLUS or INIT{\_}PROD, if we start with 

\[
\begin{bmatrix}
    \texttt{id}_1 & v_1 & 0 & 0 \\
    \texttt{id}_2 & v_2 & 0 & 0 \\
    \vdots & \vdots & \vdots & \vdots \\
    \texttt{id}_{j-1} & v_{j-1} & 0 & 1 \\
    \texttt{id}_j & 0 & 1 & 0 \\
    0 & 0 & 0 & 0 \\
    0 & 0 & 0 & 0 \\
    \vdots & \vdots & \vdots & \vdots \\
     0 & 0 & 0 & 0
\end{bmatrix}, \text{ then we return }
\begin{bmatrix}
    \texttt{id}_1 & v_1 & 0 & 0 \\
    \texttt{id}_2 & v_2 & 0 & 0 \\
    \vdots & \vdots & \vdots & \vdots \\
    \texttt{id}_{j-1} & v_{j-1} & 0 & 0 \\
    \texttt{id}_j & v_j & 0 & 1 \\
    \texttt{id}_{j+1} & 0 & 1 & 0 \\
    0 & 0 & 0 & 0 \\
    \vdots & \vdots & \vdots & \vdots \\
     0 & 0 & 0 & 0
\end{bmatrix}.
\]

For the aggregate expression (algorithm \ref{alg:agg_code}) we do the following. Let $$\pondIden{b_i}{c}=\ssum v. (v^T\cdot b_i)\cdot c\cdot v\cdot v^T + (1-v^T\cdot b_i)\cdot v \cdot v^T,$$ namely, it is the identity with $c$ in position $(i,i)$.

We define the expressions: AGG${\_}$PLUS (algorithm \ref{alg:agg_code}, lines 4, 5), AGG${\_}$PROD (algorithm \ref{alg:agg_code}, lines 6, 7),  IS${\_}$NOT${\_}$LAST (algorithm \ref{alg:agg_code}, lines 8, 9), IS${\_}$LAST and POP:

\begin{align*}
	\text{POP} &:= \Iden{\mathsf{Prev}\cdot G_{top}}\cdot G + \mathsf{Prev}\cdot V_{top}\cdot e_{V_{top}}^T  \\
	\text{AGG{\_}PLUS} &:= \isplus{G^T \cdot \left( P \cdot G_{top}\right)} \times \left[ \left( \Iden{\mathsf{Prev}\cdot V_{top}} \cdot V + \left( V^T \cdot V_{top} \right)\left( \mathsf{Prev}\cdot V_{top} \right)\right) \cdot e_{V}^T \right] \\
	\text{AGG{\_}PROD} &:= \isprod{G^T \cdot \left( P \cdot G_{top}\right)} \times \left[ \left( \pondIden{\mathsf{Prev}\cdot V_{top}}{V^T \cdot V_{top}} \cdot \Iden{\mathsf{Prev}\cdot V_{top}} \cdot V \right) \cdot e_{V}^T \right] \\
	\text{IS{\_}NOT{\_}LAST} &:= \isnotlast{G^T \cdot \left( P \cdot G_{top}\right)}{G^T \cdot G_{top}} \times \left[  G_{top} \cdot \nextgate{G^T \cdot \left( \mathsf{Prev}\cdot G_{top} \right) }{G^T \cdot G_{top}}^T + G_{top}\cdot e_{G_{top}}^T \right] \\
	\text{IS{\_}LAST} &:= \left( 1 - \isnotlast{G^T \cdot \left( P \cdot G_{top}\right)}{G^T \cdot G_{top}} \right)\times \left[ \left( \mathsf{Prev}\cdot G_{top} \right) \cdot e_{G_{top}}^T \right]
\end{align*}

The $\langfor$ expression equivalent to algorithm \ref{alg:agg_code} is $$\text{AGG}:=\text{POP} + \text{AGG{\_}PLUS}+\text{AGG{\_}PROD}+\text{IS{\_}NOT{\_}LAST}+\text{IS{\_}LAST}.$$

The $Evaluate$ method (algorithm \ref{alg:eval_code}) is defined as follows:

\begin{align*}
	\text{EVAL}&[v]= \\
	&e_{\mathsf{min}}^T \cdot \big\lbrace \ffor{X}{v_1, \ldots, v_k}: \\
	&\left( \sprod_{i=1}^k \mathsf{min}(v_i)\right) \times START + \\
	&\left( 1- \sprod_{i=1}^k \mathsf{min}(v_i)\right) \times \left( \left(1 - \mathsf{min}(G_{top})\cdot \mathsf{min}(V_{top}) \right) \times \left[ \left( 1 - G_{top}^T\cdot V_{top} \right) \times \text{INIT} + \left(  G_{top}^T\cdot V_{top} \right) \times \text{AGG} \right] + \mathsf{min}(G_{top})\times \mathsf{min}(V_{top})\times X\right) \\ 
	&\big\rbrace\cdot e_{V}
\end{align*}

Note that the $ \texttt{for}$-expression does the evaluation. The final output is in $X[1,max-2]$, we extract this value by multiplying the final result as $e_{\mathsf{min}}^T\cdot [\texttt{for}(\ldots )]\cdot e_{V}$.

Finally, we need to take care of all $n<n_0$, where $(\star)$ does not necessarily hold. For any $i$, let: $$\text{Eval}[i,A]:= \text{ the } 1\times 1 \text{ matrix with the value of the polynomial } \Phi_n(A) \text{ when } n=i.$$

Then we define: 
$$
\Phi_n(a_1,\ldots, a_i)=\ssum_{i=0}^{n_0-1}\left( e_{\mathsf{min}+\mathsf{i}}^T\cdot (e_{\mathsf{diag}}(e_{\ones}(v))\cdot e_{\mathsf{max}}) \right) \times \text{EVAL}[i,v] + \left( (\mathsf{Next}^{n_0}\cdot e_{\mathsf{min}})^T\cdot e_{\ones} (v) \right)\times \text{EVAL}[v].
$$ 
Above, $(e_{\mathsf{min}+\mathsf{i}}^T\cdot e_{\mathsf{max}})$ checks if the dimension is equal to $i$ (we multiply
by the $n\times n$ identity $e_{\mathsf{diag}}(e_{\ones}(v))$ to ensure typing), 
and $(\mathsf{Next}^{n_0}\cdot e_{\mathsf{min}})^T\cdot e_{\ones} (e_{\mathsf{min}})$ checks if the 
dimension is greater or equal than $n_0$.

\end{proof}

%% file: sections/app-lang-in-ac.tex


We consider circuits over matrices (multiple output gates). We will write 
$\Phi(A_1,\ldots ,A_k)$, where $\Phi$ is an arithmetic circuit with multiple output gates, and each 
$A_i$ is a matrix of dimensions $\alpha_i\times \beta_i$, with $\alpha_i,\beta_i \in \{n,1\}$ to denote 
the input matrices for a circuit $\Phi$. We will also write $\ttype(\Phi)=(\alpha,\beta)$, with 
$\alpha,\beta\in \{n,1\}$, to denote the size of the output matrix for $\Phi$. 
When $\{\Phi_n\mid n=1,2,\ldots\}$ is a uniform family of 
arithmetic circuits over matrices, we will assume that the Turing machine for generating $\Phi_n$ also 
gives us the information about how to access a position of each input matrix, and how to access the 
positions of the output matrix, as is usually done when handling matrices with arithmetic 
circuits \cite{Raz02}. The notion of degree is extended to be the sum of the degrees of all 
the output gates. The former will be denoted as $\Phi_{n}[i,j]$ when $\ttype(\Phi)=(n,n)$, 
$\Phi_{n}[i,1]$ when $\ttype(\Phi)=(n,1)$, $\Phi_{n}[1,j]$ when $\ttype(\Phi)=(1,n)$ and 
$\Phi_{n}$ when $\ttype(\Phi)=(1,1)$. Also, when we write $a \oplus b$ we mean 

\begin{center}
\begin{tikzpicture}[level distance=1.5cm,
  level 1/.style={sibling distance=1.5cm},
  every node/.style = {
  	shape=circle,
    draw,
    align=center,
    top color=white,
    bottom color=white
    }]
  \node {\( + \)}
    child {node { \( a \) }}
    child {node { \( b \) }};
\end{tikzpicture}
\end{center}
When we write $\bigoplus_{l=1}^n a_l$ we mean 

\begin{center}
\begin{tikzpicture}[level distance=1.5cm,
  level 1/.style={sibling distance=1.5cm},
  every node/.style = {
  	shape=circle,
    draw,
    align=center,
    top color=white,
    bottom color=white
    }]
  \node {\( + \)}
    child {node { \( a_1 \) }}
    child {node { \( \cdots \) }}
    child {node { \( a_n \) }};
\end{tikzpicture}
\end{center}
Same with $\otimes$. Now we prove the statement.

\newtheorem*{LANGINCIRC}{Theorem~\ref{th-ml-to-circuits}}

\begin{LANGINCIRC}
  Let $e$ be a \langfor expression over a schema $\Sch$, and let $V_1,\ldots ,V_k$ be the variables of $e$ such that $\ttype(V_i)\in \{(\alpha,\alpha), (\alpha,1), (1,\alpha), (1,1)\}$. Then there exists a uniform arithmetic circuit family over matrices $\Phi_n(A_1,\ldots ,A_k)$ such that:
  \begin{itemize}
  \item For any instance $\I = (\dom,\conc)$ such that $\dom(\alpha) = n$ and $\conc(V_i) = A_i$ it holds that:
  \item $\sem{e}{\I} = \Phi_n(A_1,\ldots ,A_k)$.
  \end{itemize}
\end{LANGINCIRC}

\begin{proof}

Let $e$ be a \langfor expression. 

If $e=V$ then $\Phi_n^e:=\Phi(A)$, and we have that
\begin{itemize}
	\item If $\ttype(V)=(1,1)$ then $\ttype(\Phi^e_n)=(1,1)$ and $\Phi^e_n$ has the one input/output gate.
	\item If $\ttype(V)=(1,\alpha)$ then $\ttype(\Phi^e_n)=(1,n)$ and $\Phi^e_n$ has $n$ input/output gates.
  \item If $\ttype(V)=(\alpha,1)$ then $\ttype(\Phi^e_n)=(n,1)$ and $\Phi^e_n$ has $n$ input/output gates.
	\item If $\ttype(V)=(\alpha,\alpha)$ then $\ttype(\Phi^e_n)=(n,n)$ and $\Phi^e_n$ has $n^2$ input/output gates. 
\end{itemize}

If $e=e'^T$ then $\Phi^e_n=\Phi^{e'}_n$. 
\begin{itemize}
	\item If $\ttype(\Phi^{e'}_n)=(1, 1)$ then $\Phi^e_n=\Phi^{e'}_n$ and $\type(\Phi^e_n)=(1,1)$.
	\item If $\ttype(\Phi^{e'}_n)=(1, n)$ then $\type(\Phi^e_n)=(n,1)$ and $\Phi^e_n[i,1]:=\Phi^{e'}_n[1,i]$. 
  \item If $\ttype(\Phi^{e'}_n)=(n, 1)$ then $\type(\Phi^e_n)=(1,n)$ and $\Phi^e_n[1,j]:=\Phi^{e'}_n[j,1]$. 
  \item If $\ttype(\Phi^{e'}_n)=(n, n)$ then $\type(\Phi^e_n)=(n,n)$ and $\Phi^e_n[i,j]:=\Phi^{e'}_n[j,i]$. 
\end{itemize}

If $e={\ones}(e')$ where $\ttype(\Phi^{e'}_n)=(\alpha,\beta)$ then $\ttype(\Phi^{e}_n)=(\alpha,1)$ and $\Phi^e_n[i,1]:=1$.

If $e=e_1 + e_2$ we have

\begin{itemize}
	\item When $\ttype(\Phi^{e_1}_n)=\ttype(\Phi^{e_2}_n)=(1, 1)$  then $\ttype(\Phi^{e}_n)=(1, 1)$ and $\Phi^e_n:=\Phi^{e_1}_n \oplus \Phi^{e_2}_n$.
  \item When $\ttype(\Phi^{e_1}_n)=\ttype(\Phi^{e_2}_n)=(1, n)$  then $\ttype(\Phi^{e}_n)=(1, n)$ and $\Phi^e_n[1,j]:=\Phi^{e_1}_n[1,j] \oplus \Phi^{e_2}_n[1,j]$.
  \item When $\ttype(\Phi^{e_1}_n)=\ttype(\Phi^{e_2}_n)=(n, 1)$  then $\ttype(\Phi^{e}_n)=(n, 1)$ and $\Phi^e_n[i,1]:=\Phi^{e_1}_n[i,1] \oplus \Phi^{e_2}_n[i,1]$.
  \item When $\ttype(\Phi^{e_1}_n)=\ttype(\Phi^{e_2}_n)=(n, n)$  then $\ttype(\Phi^{e}_n)=(n, n)$ and $\Phi^e_n[i,j]:=\Phi^{e_1}_n[i,j] \oplus \Phi^{e_2}_n[i,j]$.
\end{itemize}

If $e=f(e_1, \ldots, e_k)$ we have two cases

\begin{itemize}
  \item When $f$ is the function $f_{\odot}$ (recall that this function is definable in $\langf{\emptyset}$ by Lemma \ref{lm-prod-sum}) then
  \begin{itemize}
    \item If $\ttype(\Phi^{e_1}_n)=\ldots =\ttype(\Phi^{e_k}_n)=(1, 1)$ then $\Phi^e_n:=\bigotimes_{l=1}^k \Phi^{e_l}_n$.
    \item If $\ttype(\Phi^{e_1}_n)=\ldots =\ttype(\Phi^{e_k}_n)=(1, n)$ then $\Phi^e_n[1,j]:=\bigotimes_{l=1}^k \Phi^{e_l}_n[1,j]$.
    \item If $\ttype(\Phi^{e_1}_n)=\ldots =\ttype(\Phi^{e_k}_n)=(n, 1)$ then $\Phi^e_n[i,1]:=\bigotimes_{l=1}^k \Phi^{e_l}_n[i,1]$.
    \item If $\ttype(\Phi^{e_1}_n)=\ldots =\ttype(\Phi^{e_k}_n)=(n, n)$ then $\Phi^e_n[i,j]:=\bigotimes_{l=1}^k \Phi^{e_l}_n[i,j]$.
  \end{itemize}
	\item When $f$ is any function, we prove the case when $\ttype(\Phi^{e_1}_n)=\ldots =\ttype(\Phi^{e_k}_n)=(1, 1)$
  (only case necessary, as discussed in Appendix \ref{app:simp}). Here $\Phi^e_n$ is 
	
\begin{center}
\begin{tikzpicture}[level distance=1.5cm,
  level 1/.style={sibling distance=1.5cm},
  every node/.style = {
  	shape=circle,
    draw,
    align=center,
    top color=white,
    bottom color=white
    }]
  \node {\( f \)}
    child {node { \( \Phi^{e_1}_n \) }}
    child {node { \( \cdots \) }}
    child {node { \( \Phi^{e_k}_n \) }};
\end{tikzpicture}
\end{center}

Note that since for the context of this result we only consider $\langfor = \langf{\emptyset}$, this case is not strictly necessary, modulo for $f_\odot,f_\oplus$ due to Lemma \ref{lm-prod-sum}. However, if we extend the circuits with the same functions allowed in \langfor, then our inductive construction still goes through, as just illustrated.

\end{itemize}

If $e=e_1\cdot e_2$ we have

\begin{itemize}
	\item When $\ttype(\Phi^{e_1}_n)=(1,1)$ and $\ttype(\Phi^{e_2}_n)=(1, 1)$ then $\ttype(\Phi^{e}_n)=(1, 1)$ and $\Phi^{e}_n:=\Phi^{e_1}_n \otimes \Phi^{e_2}_n$.
  \item When $\ttype(\Phi^{e_1}_n)=(1,1)$ and $\ttype(\Phi^{e_2}_n)=(1, n)$ then $\ttype(\Phi^{e}_n)=(1, n)$ and $\Phi^{e}_n[1,j]:=\Phi^{e_1}_n \otimes \Phi^{e_2}_n[1,j]$.
  \item When $\ttype(\Phi^{e_1}_n)=(n,1)$ and $\ttype(\Phi^{e_2}_n)=(1, 1)$ then $\ttype(\Phi^{e}_n)=(n, 1)$ and $\Phi^{e}_n[i,1]:=\Phi^{e_1}_n[i,1] \otimes \Phi^{e_2}_n$.
  \item When $\ttype(\Phi^{e_1}_n)=(n,1)$ and $\ttype(\Phi^{e_2}_n)=(1, n)$ then $\ttype(\Phi^{e}_n)=(n, n)$ and $\Phi^{e}_n[i,j]:=\Phi^{e_1}_n[i,1] \otimes \Phi^{e_2}_n[1,j]$.
  \item When $\ttype(\Phi^{e_1}_n)=(1,n)$ and $\ttype(\Phi^{e_2}_n)=(n, 1)$ then $\ttype(\Phi^{e}_n)=(1, 1)$ and $$\Phi^{e}_n:=\bigoplus_{k=1}^n \left( \Phi^{e_1}_n[1,k] \otimes \Phi^{e_2}_n[k,1] \right).$$
  \item When $\ttype(\Phi^{e_1}_n)=(1,n)$ and $\ttype(\Phi^{e_2}_n)=(n, n)$ then $\ttype(\Phi^{e}_n)=(1, n)$ and $$\Phi^{e}_n[1,j]:=\bigoplus_{k=1}^n \left( \Phi^{e_1}_n[1,k] \otimes \Phi^{e_2}_n[k,j] \right).$$
  \item When $\ttype(\Phi^{e_1}_n)=(n,n)$ and $\ttype(\Phi^{e_2}_n)=(n, 1)$ then $\ttype(\Phi^{e}_n)=(n, 1)$ and $$\Phi^{e}_n[i,1]:=\bigoplus_{k=1}^n \left( \Phi^{e_1}_n[i,k] \otimes \Phi^{e_2}_n[k,1] \right).$$
  \item When $\ttype(\Phi^{e_1}_n)=(n,n)$ and $\ttype(\Phi^{e_2}_n)=(n, n)$ then $\ttype(\Phi^{e}_n)=(n, n)$ and $$\Phi^{e}_n[i,j]:=\bigoplus_{k=1}^n \left( \Phi^{e_1}_n[i,k] \otimes \Phi^{e_2}_n[k,j] \right).$$
\end{itemize}

If $e=\ffor{X}{v}e'(X, v)$, then define $\Phi^{\mathbf{0}}$ 
as the zero matrix circuit $\ttype(\Phi^{\mathbf{0}})=(1,1)$ if $\ttype(\Phi^{e'}_n)=(1,1)$ and 
$\ttype(\Phi^{\mathbf{0}})=(n,n)$ if $\ttype(\Phi^{e'}_n)=(n,n)$. Also, $\Phi^{\mathbf{0}}=0$ and
$\Phi^{\mathbf{0}}[i,j]=0$ $\forall i,j$ for each case respectively. Now for $i=1,\ldots, n$, define
$\Phi^{v_i}$ as the circuit such that $\ttype(\Phi^{v_i})=(n,1)$ and $\Phi^{v_i}[i,1]:=1$ and zero otherwise.
Finally, define

$$\Phi^{e}_n=\Phi^{e'}_n\left( \Phi^{e'}_n \left( \cdots \left( \Phi^{e'}_n\left( \Phi^{\mathbf{0}}, \Phi^{v_1}\right), \Phi^{v_2}\right)\cdots, \Phi^{v_{n-1}} \right), \Phi^{v_n} \right).$$

Note that every circuit adds a constant number of layers except when $e=\ffor{X}{v}e'(X, v)$. 
This means that the depth still is polynomial. When $e=\ffor{X}{v}e'(X, v)$
we have that the depth of the circuit is $n\cdot p(n)$, where the depth of $e'(X, v)$ is $p(n)$, 
so it also remains polynomial.

Here, we do not need to translate scalar multiplication
because it can be simulated using the $\mathsf{ones}$ operator and $f_{\kprod}$ (see section \ref{app:simp}).

Finally, we remark that when composing the circuits the fact that uniformity is preserved (i.e. the resulting circuit can be generated by a \logspace\ machine) is proved analogously as when composing two \logspace\ transducers \cite{aroraB2009}. The only more involved case is treating for-loop construction, however, notice here that we only need to keep track of where we are in the evaluation (i.e. which $v_i$ we are processing), and not of all the previous results, given that they update the resulting matrix in a fixed order.

\end{proof}

%% file: sections/app-undec-result.tex
\newtheorem*{Undec}{Proposition~\ref{prop-undec}}
Let $e$ be a \langfor expression over a matrix schema $\mathcal{S}=(\mathcal{M},\textsf{size})$ and let $V_1,\ldots, V_k$ be
the variables of $e$, each of type $(\alpha,\alpha)$, $(1,\alpha)$, $(\alpha,1)$ or $(1,1)$. We know from Theorem~\ref{th-ml-to-circuits}
that there exists a uniform arithmetic circuit family $\{\Phi_n \mid n=1,2,\ldots\}$
such that $\sem{e}{\I}=\Phi_n(A_1,\ldots,A_k)$ for any instance $\I$ such that
$\mathcal{D}(\alpha)=n$ and $\conc(V_i)=A_i$ for $i=1,\ldots,k$. We are interested in deciding
whether there exists such a  uniform arithmetic circuit family $\{\Phi_n \mid n=1,2,\ldots\}$
of polynomial degree, i.e., such that $\mathsf{degree}(\Phi_n)=\mathcal{O}(p(n))$ for some polynomial $p(x)$. If such a circuit family exists, we call $e$ of polynomial degree.

\begin{Undec}
	Given a \langfor expression $e$ over a schema $\Sch$, it is undecidable to check whether $e$ is of polynomial degree.
\end{Undec}
\begin{proof}
We show undecidability based on the following undecidable language:
$$
\{ \langle M\rangle\mid \text{$M$ is a deterministic TM which halts on the empty input}\},
$$
where $\langle M\rangle$ is some string encoding of $M$.
Consider a TM $M$ described by $(Q,\Gamma=\{0,1\},q_0,q_m,\Delta)$
with $Q=\{q_1,\ldots,q_m\}$ its states, $q_1$ being the initial state and $q_m$ being
the halting state, $\Gamma$ is the tape alphabet, and $\Delta$ is a transition function
from $Q\times \Gamma\to Q\times\Gamma\times \{\leftarrow,\sqcup,\rightarrow\}$. The simulation
of linear space TM, as given in the proof of Proposition~\ref{prop:transducer} can be easily modified to
any TM $M$ provided that we limit the execution of $M$ to exactly $n$ steps. Let $e_M$ denote this expression. Similarly
as in the linear space TM simulation, we have vector variables $Q_1,\ldots,Q_m$ encoding the
states, a single relation $T$ encoding the tape and relation $H_T$ encoding the position
of the tape.  When an instance $\I$ assigns $n$ to $\alpha$, we have a tape of length $n$ at our disposal. This suffices if we let $M$ run for $n$ steps. We further observe that all vector variables can be assumed to be zero, initially.
This is because we do not have input. So, let $\I_n^0$ denote the instance which assigns vector variables to the $n$-dimensional zero vector.  Furthermore, by contrast to the linear space TM simulation, we use a single vector $v$ (instead of $k$ such vectors) to simulate $n$ steps of $M$. Finally, we modify the expression given in the proof of Proposition~\ref{prop:transducer} such $\sem{e_M}{\I_n^0}$  returns $1$ if $M$
halts in at most $n$ steps, and $0$ if $M$ did not halt yet after $n$ steps.

As a consequence, when $M$ does not halt, $\sem{e_M}{\I_n^0}=0$ for all $n\geq 0$. When $M$ halts, there will be an $n$ such that $\sem{e_M}{\I_n^0}=1$ It now suffices to consider the \langfor expression
$$
d_M:=e_M\cdot e_{\mathsf{exp}}
$$
where $e_{\texttt{exp}} = \ffor{v}{X=\ones(X)^T\cdot\ones(X)}{X\cdot X}$ such that
$e_{\texttt{exp}}(\I_n^0)=n^{2^n}$. Then, when $M$ does not halt we can clearly compute $d_M$ with a constant degree circuit ``0''
for any $n$, otherwise, the circuit needed will be of exponential degree
for at least one $n$, simply because no polynomial degree uniform  circuit family can compute $n^{2^n}$. In other words, deciding whether $d_M$ has polynomial degree coincides with deciding whether $M$ halts.
\end{proof}

%
%
%

%% file: sections/app-sum-to-ara.tex
\newtheorem*{SUMTOARA}{Proposition~\ref{prop:sum_to_ara}}

We prove proposition \ref{prop:sum_to_ara}.

\begin{SUMTOARA}
  For each \langsum expression $e$ over schema $\Sch$ such that $\Sch(e)=(\alpha,\beta)$ with $\alpha\neq 1\neq\beta$, there exists a \rak  expression $\arae(e)$ over relational schema $\text{Rel}(\Sch)$ such that $\text{Rel}(\Sch)(\arae(e))=\{\row_\alpha,\row_\beta\}$ and 
	such that for any instance $\I$ over~$\Sch$,
	$$
	\sem{e}{\I}_{i,j}=\ssem{\arae(e)}{\text{Rel}(\I)}(t)
	$$
	for tuple $t(\mathrm{row}_\alpha)=i$ and $t(\mathrm{col}_\beta)=j$. Similarly for when $e$ has schema $\Sch(e)=(\alpha,1)$, $\Sch(e)=(1,\beta)$ or $\Sch(e)=(1,1)$, then $\arae(e)$ has schema $\text{Rel}(\Sch)(\arae(e))=\{\mathrm{row}_\alpha\}$,
	$\text{Rel}(\Sch)(\arae(e))=\{\mathrm{col}_\alpha\}$, or
	$\text{Rel}(\Sch)(\arae(e))=\{\}$, respectively.
\end{SUMTOARA}

\begin{proof}
We start from a matrix schema $\Sch=(\Mnam,\size)$, where $\Mnam\subset \Mvar$ is a finite set of matrix variables, 
and $\size: \Mvar \mapsto \DD\times \DD$ is a function that maps each matrix variable to a pair of size symbols. 
On the relational side we have for each size symbol $\alpha\in\DD\setminus\{1\}$, attributes $\alpha$, $\row_\alpha$, 
and $\col_\alpha$ in $\att$. We also reserve some special attributes $\gamma_1,\gamma_2,\ldots$ whose role will become clear shortly.
For each $V\in\Mnam$ and $\alpha \in \DD$ we denote
by $R_V$ and $R_\alpha$ its corresponding relation name, respectively. 

Then, given $\Sch$ we define the relational 
schema $\text{Rel}(\Sch)$ such that $\fdom(\text{Rel}(\Sch)) =  \{R_\alpha \mid \alpha\in\DD\} \cup \{R_V \mid V \in \Mnam\}$
where $\text{Rel}(\Sch)(R_\alpha) = \{\alpha\}$ and for all $V\in\Mnam$:
\[
\text{Rel}(\Sch)(R_V) = \begin{cases}
\lbrace\row_\alpha,\col_\beta \rbrace & \text{ if $ \size(V)=(\alpha,\beta)$} \\
\lbrace\row_\alpha \rbrace & \text{ if $ \size(V)=(\alpha,1)$} \\
\lbrace\col_\beta \rbrace  &
\text{ if $ \size(V)=(1,\beta)$} \\
\lbrace\rbrace & \text{ if $\size(V)=(1,1)$}.
\end{cases}
\]

Next, for a matrix instance $\I = (\dom,\conc)$ over $\Sch$,
let $V\in\Mnam$ with $\size(V)=(\alpha,\beta)$ and let $\conc(V)$ be its corresponding $K$-matrix of dimension $\dom(\alpha)\times \dom(\beta)$.
The $K$-instance in $\mathsf{RA}_{K}^+$ according to $\I$ is $\text{Rel}(\I)$ with data domain $\ddom = \mathbb{N} \setminus \{0\}$. For each $V\in\Mnam$ we define 
$R_V^{\text{Rel}(\I)}(t):=\conc(V)_{ij}$ whenever $t(\row_\alpha) = i \leq \dom(\alpha)$ and $t(\col_\beta) = j \leq \dom(\beta)$, and $\kzero$ otherwise. 
Also, for each $\alpha \in \DD$ we define $R_\alpha^{\text{Rel}(\I)}(t):=\kone$ whenever $t(\alpha) \leq \dom(\alpha)$, and $\kzero$ otherwise.
If $\size(V)=(\alpha,1)$ then $R_V^{\text{Rel}(\I)}(t):=\conc(V)_{i1}$ whenever $t(\row_\alpha) = i \leq \dom(\alpha)$ and $\kzero$ otherwise.
Similarly, if $\size(V)=(1,\beta)$ then $R_V^{\text{Rel}(\I)}(t):=\conc(V)_{1j}$ whenever $t(\col_\beta) = j \leq \dom(\beta)$ and $\kzero$ otherwise.
If $\size(V)=(1,1)$ then $R_V^{\text{Rel}(\I)}(()):=\conc(V)_{11}$.


Let $e$ be a \langsum expression. In the following we need to distinguish between matrix variables $v$
that occur in $e$ as part of a sub-expression $\ssum v. (\cdot)$, i.e., those variables that will be used to iterate over by means of canonical vectors, and those that are not. To make this distinction clear, we use $v_1,v_2,\ldots$ for those ``iterator'' variables, and capital $V$ for the other variables occurring in $e$. For simplicity, we assume that each occurrence of $\ssum$ has its own iterator variable associated with it. 

We define free (iterator) variables, as follows.
$\mathsf{free}(V):=\emptyset$, $\mathsf{free}(v):=\{v\}$, $\mathsf{free}(e^T):=\mathsf{free}(e)$, $\mathsf{free}(e_1+e_2):=\mathsf{free}(e_1)\cup \mathsf{free}(e_2)$, $\mathsf{free}(e_1\cdot e_2):=\mathsf{free}(e_1)\cup \mathsf{free}(e_2)$,
 $\mathsf{free}(f_\odot(e_1,\ldots,e_k)):=\mathsf{free}(e_1)\cup\cdots \cup \mathsf{free}(e_k)$, and $\mathsf{free}(e=\ssum V. e_1)=\mathsf{free}(e_1)\setminus\{v\}$. We will explicitly denote the free variables in an expression $e$ by writing $e(v_1,\ldots,v_k)$.

We now use the following induction hypotheses:
\begin{itemize}
	\item If $e(v_1,\ldots,v_k)$ is of type $(\alpha,\beta)$ then there exists a
	\rak expression $\arae$ such that $\text{Rel}(\Sch)(\arae(e))=\{\row_\alpha,\col_\beta,\gamma_1,\ldots,\gamma_k\}$
	and such that 
	$$
	\ssem{\arae(e)}{\text{Rel}(\I)}(t)=\sem{e}{\I[v_1\gets b_{i_1},\ldots,v_k\gets b_{i_k}]}_{i,j}
	$$
	for tuple $t(\mathrm{row}_\alpha)=i$, $t(\mathrm{col}_\beta)=j$ and $t(\gamma_s)=i_s$ for $s=1,\ldots, k$.
	\item If $e(v_1,\ldots,v_k)$ is of type $(\alpha,1)$ then there exists a
	\rak expression $\arae$ such that $\text{Rel}(\Sch)(\arae(e))=\{\row_\alpha,\gamma_1,\ldots,\gamma_k\}$
	and such that 
	$$
	\ssem{\arae(e)}{\text{Rel}(\I)}(t)=\sem{e}{\I[v_1\gets b_{i_1},\ldots,v_k\gets b_{i_k}]}_{i,1}
	$$
	for tuple $t(\mathrm{row}_\alpha)=i$,  and $t(\gamma_s)=i_s$ for $s=1,\ldots, k$.
	And similarly for when $e$ is type $(1,\beta)$.
	\item If $e(v_1,\ldots,v_k)$ is of type $(1,1)$ then there exists a
	\rak expression $\arae$ such that $\text{Rel}(\Sch)(\arae(e))=\{\gamma_1,\ldots,\gamma_k\}$
	and such that 
	$$
	\ssem{\arae(e)}{\text{Rel}(\I)}(t)=\sem{e}{\I[v_1\gets b_{i_1},\ldots,v_k\gets b_{i_k}]}_{1,1}
	$$
	for tuple $t(\gamma_s)=i_s$ for $s=1,\ldots, k$.
\end{itemize}
Clearly, this suffices to show the proposition since we there consider expressions $e$ for which $\mathsf{free}(e)=\emptyset$, in which case the above statements reduce to the one given in the proposition.

The proof is by induction on the structure of \langsum expressions. In line with the simplifications in Section~\ref{app:simp}, it suffices to consider pointwise function application with $f_\odot$ instead of scalar multiplication. (We also note that we can express the one-vector operator in \langsum, so scalar multiplication can be expressed using $f_\odot$ in \langsum).

Let $e$ be a \langsum expression.
\begin{itemize}
  \item If $e=V$ then $\arae (e):=R_V$.
  \item If $e=v_p$ then $\arae (e):=\sigma_{\lbrace \row_\alpha,\gamma_p\rbrace}\bigl(\rho_{\row_\alpha\to \alpha}(R_\alpha)\bowtie \rho_{\gamma_p\to \alpha}(R_\alpha)\bigr)$ when  $v_p$ is of type $(\alpha,1)$. It is here that we introduce the attribute $\gamma_p$ associated with iterator variable $v_p$.
 We note that 
$$ \ssem{\arae(v_p)}{\text{Rel}(\I)}(t)=\sem{v_p}{\I[v_p\gets b_{j}]}_{i,1}=(b_j)_{i,1}
$$
for $t(\mathrm{row}_\alpha)=i$ and $t[\gamma_p]=j$. Indeed, $(b_j)_{i,1}=\kone$ if $j=i$
and this holds when $t(\mathrm{row}_\alpha)=t[\gamma_p]=j$, and $(b_j)_{i,1}=\kzero$ if $j\neq i$
and this also holds when $t(\mathrm{row}_\alpha)\neq t[\gamma_p]=j$.


  \item If $e(v_1,\ldots,v_k)=(e_1(v_1,\ldots,v_k))^T$ with $\Sch (e_1)=(\alpha, \beta)$ then \[
\arae(e) :=
\begin{cases}
\rho_{\mathrm{row}_\alpha \to \mathrm{col}_\alpha,\mathrm{col}_\beta \to \mathrm{row}_\beta}\bigl(\arae(e_1)\bigr) & \text{if } \alpha \neq 1 \neq \beta; \cr
\rho_{\mathrm{row}_\alpha \to \mathrm{col}_\alpha}\bigl(\arae(e_1)\bigr) & \text{if } \alpha \neq 1 = \beta; \cr
\rho_{\mathrm{col}_\beta \to \mathrm{row}_\beta}\bigl(\arae(e_1)\bigr) & \text{if } \alpha = 1 \neq \beta; \cr
\arae(e_1) & \text{if } \alpha = 1 = \beta.
\end{cases}
\]
\floris{There is an issue here since $e_1$ and $e_2$ can have different free iterator variables. I think we can ensure that both have the same by introducing them somehow, or alternative, since all operators are linear, by pushing $+$ all the way down? Not sure.}
\item If $e=e_1(v_1,\ldots,v_k)+e_2(v_1,\ldots,v_k)$ with $\Sch (e_1)=\Sch (e_2)=(\alpha, \beta)$ then $\arae (e):=\arae (e_1)\cup \arae (e_2)$. We assume here that $e_1$ and $e_2$ have the same free variables. This is without loss of generality. Indeed, as an example, suppose that we have $e_1(v_1,v_2)$
and $e_2(v_2,v_3)$. Then, we can replace $e_1$ by  $e_1(v_1,v_2,v_3)=(v_3^T\cdot v_3)\times e_1(v_1,v_2)$
and similarly, $e_2$ by $e_2(v_1,v_2,v_3)=(v_1^T\cdot v_1)\times e_2(v_2,v_3)$, where in addition we replace scalar multiplication with its simulation using $f_{\odot}$ and the ones vector, as mentioned earlier. 
	%
  \item If $e=f_\odot(e_1,\ldots, e_k)$ with $\Sch(e_i)=\Sch(e_j)$ for all $i,j\in[1,k]$, then $\arae(e):=\arae(e_1)\Join \cdots \Join\arae(e_k)$.

  \item If $e=e_1\cdot e_2$ with $\Sch (e_1)=(\alpha, \gamma)$ and $\Sch (e_2)=(\gamma, \beta)$, we have two cases. If $\gamma = 1$ then $\arae (e):=\arae (e_1)\Join \arae (e_2)$.
If $\gamma\neq 1$ then
$$
\arae (e) := \pi_{\lbrace \row_{\alpha},\col_{\beta}, \gamma_1,\ldots,\gamma_k \rbrace}\left(\rho_{C\to \col_\gamma}(\arae (e_1))\Join \rho_{C\to \row_\gamma}(\arae (e_2)) \right),
$$
when $\text{Rel}(\Sch)(\arae(e_1))=\{\row_\alpha,\col_\gamma,\gamma_1',\ldots,\gamma_{\ell}'\}$,
$\text{Rel}(\Sch)(\arae(e_2))=\{\row_\gamma,\col_\beta,\gamma_1'',\ldots,\gamma_{\ell}''\}$ and $\{\gamma_1,\ldots,\gamma_k\}=\{\gamma_1',\ldots,\gamma_k',\gamma_1'',\ldots,\gamma_m''\}$.

  \item If $e(v_1,\ldots,v_{p-1},v_{p+1},\ldots,v_k)=\ssum v_p. e_1(v_1,\ldots,v_k)$ where $\Sch(e_1)=(\alpha,\beta)$ and $\Sch(V)=(\gamma,1)$. Then we do 
  $$
  \arae (e):=\pi_{\text{Rel}(\Sch)(\arae(e_1))\setminus\{\gamma_p\}} \arae (e_1).
  $$
 Indeed, by induction we know that 
 $$
\ssem{\arae(e_1)}{\text{Rel}(\I)}(t)=\sem{e}{\I[v_1\gets b_{i_1},\ldots,v_k\gets b_{i_k}]}_{i,j}
$$
for tuple $t(\mathrm{row}_\alpha)=i$, $t(\mathrm{col}_\beta)=j$ and $t(\gamma_s)=i_s$ for $s=1,\ldots, k$.
Hence, for $t(\mathrm{row}_\alpha)=i$, $t(\mathrm{col}_\beta)=j$ and $t(\gamma_s)=i_s$ for $s=1,\ldots, k$ and $s\neq p$,
$$
\ssem{\arae(e_1)}{\text{Rel}(\I)}(t):=\bigksum_{i_p=1,\ldots,\dom(\gamma)} \sem{e_1}{\I[v_1\gets b_{i_1},\ldots,v_k\gets b_{i_k}]}_{i,j},$$
which precisely corresponds to 
$$
\sem{\ssum v_p. e_1(v_1,\ldots,v_k)}{\I[v_1\gets b_{i_1},\ldots,v_{p-1}\gets b_{p-1},v_{p+1}\gets b_{p+1},\ldots,v_k\gets b_k]}_{i,j}.
$$

\end{itemize}
All other cases, when expressions have type $(\alpha,1)$, $(1,\beta)$ or $(1,1)$ can be dealt with in a similar way.
\end{proof}

%% file: sections/app-ara-to-sum.tex
Let $\cR$ be binary relational schema. For each $R\in \cR$ we associate a matrix variable 
$V_R$ such that, if $R$ is a binary relational signature, then $V_R$ represents a (square) matrix, 
if $R$ is unary, then $V_R$ represents a vector and if $|R|=0$ then $V_R$ represents a constant. Formally, 
fix a symbol $\alpha \in \DD \setminus \{1\}$. Let $\text{Mat}(\cR)$ denote the \lang \ schema
$(\Mnam_\cR,\size_\cR)$ such that $\Mnam_\cR = \{ V_R \mid R \in \cR\}$ and $\size_\cR(V_R) = (\alpha, \alpha)$ 
whenever $|R| = 2$, $\size_\cR(V_R) = (\alpha, 1)$ whenever $|R|=1$ and $\size_\cR(V_R) = (1, 1)$ whenever $|R|=0$. 
Let $\cJ$ be the $K$-instance of $\cR$ and suppose that $\adom(\cJ) = \{d_1, \ldots, d_n\}$ is 
the active domain (with arbitrary order) of $\cJ$. 
Define the matrix instance $\text{Mat}(\cJ) = (\dom_\cJ,\conc_\cJ)$ such 
that $\dom_\cJ(\alpha) = n$, $\conc_\cJ(V_R)_{i,j} = R^{\cJ}((d_i, d_j))$ whenever $|R|=2$, $\conc_\cJ(V_R)_{i} = R^{\cJ}((d_i))$ 
whenever $|R|=1$, 
\floris{This case relates to nullary relations. What does $R^{\cJ}$ mean?}
and $\conc_\cJ(V_R)_{1,1} = R^{\cJ}$ whenever $|R|=0$. 
Note that we consider the active domain of the whole $K$-instance.

\newcommand{\earae}{e_{\arae}}

We next translate \rak expressions in to \langsum expressions over an extended schema. More specifically, for each attribute $A \in \att$ we define a vector variable $v_A$ of type $(\alpha,1)$. Then for each \rak expression $\arae$ with attributes $A_1, \ldots, A_k$ we define a \langsum expression $\earae(v_{A_1}, \ldots, v_{A_k})$ of type $(1,1)$ such that the following inductive hypothesis holds:
$$
\sem{\earae}{\text{Mat}(\cJ)[v_{A_1} \gets b_{i_1},\ldots, v_{A_k} \gets b_{i_k}]} = 
\ssem{\arae}{\cJ}(t) \ \ \ \ \ \ \  (*)
$$
where $t(A_s)=i_s$ for $s=1,\ldots, k$. The proof of this claim follows by induction on the structure of expressions:
\begin{itemize} \itemsep3mm
	\item If $\arae=R$, then $\earae:=v_{A_1}^T \cdot V_R \cdot v_{A_2}$ if $\mathcal{R}(R)=\{A_1,A_2\}$ with $A_1<A_2$; 
	$\earae:=V_R^T \cdot v_A$ if $\mathcal{R}(R)=\{A\}$; and 
	$\earae:=V_R$ if $\mathcal{R}(R)=\{\}$.
	\item If $\arae=\arae_1\cup \arae_2$ then
	$\earae:=e_{\arae_1} + e_{\arae_2}$.
	\item If $\arae=\pi_{Y}(\arae_1)$ for $Y\subseteq \mathcal{R}(\arae_1)$ and $\{B_1, \ldots, B_l\} = \mathcal{R}(\arae_1) \setminus Y$ then
	$$
	\earae:= \Sigma v_{B_1}. \ \Sigma v_{B_2}. \ \ldots \Sigma v_{B_l}. \ e_{\arae_1}
	$$
	\item If $\arae=\sigma_{Y}(\arae_1)$ with $Y\subseteq\mathcal{R}(\arae_1)$ then
	$$
	\earae:=e_{\arae_1}\cdot \prod_{A,B\in Y} (v_{A}^T \cdot v_{B}).
	$$
	Here $\Pi$ is the matrix multiplication of expressions of type $(1,1)$.
	\item If $\arae=\rho_{X\mapsto Y}(\arae_1)$ then
	$$\earae:=e_{\arae_1}[v_B\gets v_A\mid A\in X, B\in Y, A\mapsto B].$$
	In other words, we rename variable $v_B$ with variable $v_B$ in all the expression $e_{\arae_1}$. 
	\item If $\arae=\arae_1\bowtie \arae_2$ then
	$\earae:=e_{\arae_1} \cdot e_{\arae_1}$ where the product is over expression of type $(1,1)$.
\end{itemize}
One can check, by induction over the construction, that the inductive hypothesis $(*)$ holds in each case.
Now we can obtain proposition \ref{prop:ara_to_sum}.

\newtheorem*{ARATOSUM}{Proposition~\ref{prop:ara_to_sum}}

\begin{ARATOSUM}
  Let $\cR$ be a binary relational schema. For each $\mathsf{RA}_{K}^+$  expression $\arae$ over $\cR$  such that $|\cR(\arae)| = 2$, there exists a \langsum  expression $\Psi(\arae)$ over \lang \ schema $\text{Mat}(\cR)$ such that for any $K$-instance $\cJ$ with $\adom(\cJ) = \{d_1, \ldots, d_n\}$ over $\cR$,
	$$
	\ssem{\arae}{\cJ}((d_i, d_j))=\sem{\Psi(\arae)}{\text{Mat}(\cJ)}_{i,j}.
	$$
	Similarly for when $|\cR(\arae)| = 1$, or $|\cR(\arae)| = 0$ respectively.
\end{ARATOSUM}
\begin{proof}
As a consequence of the previous discussion above, when $\arae$ is a \rak expression 
such that $\mathcal{R}(\arae)=\{A_1,A_2\}$ with $A_1<A_2$ then we define
$$
\Psi(\arae) \ = \ \Sigma v_{A_1}. \ \Sigma v_{A_2}. \ \earae \cdot (v_{A_1} \cdot v_{A_2}^T). 
$$
Instead, when $\mathcal{R}(\arae)=\{A\}$ we have
$$
\Psi(\arae) \ = \ \Sigma v_{A}. \  (v_{A} \cdot \earae). 
$$
And when $\mathcal{R}(\arae)=\{\}$ we have
$$
\Psi(\arae) \ = \ \earae.
$$
By using the inductive hypothesis $(*)$ one can check that $\Psi(\arae)$ works in each case as expected. 
\end{proof}

%% file: sections/app-prod-and-wl.tex
\newtheorem*{WL}{Proposition~\ref{prop:wl}}

We prove proposition \ref{prop:wl}:

\begin{WL}
  Weighted logics over $\Gamma$ and \langprod over $\Sch$ have the same expressive power. More specifically,
  \begin{itemize}
  	\item for each \langprod expression $e$ over $\Sch$ such that $\Sch(e)=(1,1)$, there exists a WL-formula $\Phi(e)$ over $\text{WL}(\Sch)$ such that for every instance $\I$ of~$\Sch$, 
  	$
  	\sem{e}{\I} = \ssem{\Phi(e)}{\text{WL}(\I)}
  	$.
  	\item for each WL-formula $\varphi$ over $\Gamma$ without free variables, there exists a \langprod expression $\Psi(\varphi)$ such that for any structure $\cA$ over~$\text{Mat}(\Gamma)$,
  	$
  	\ssem{\varphi}{\cA}=\sem{\Psi(\varphi)}{\text{Mat}(\cA)}
  	$.
  \end{itemize}	
\end{WL}

\begin{proof}
Both directions are proved by induction on the structure of expressions.

\smallskip

\noindent \textbf{(\langprod to WL)} First, let $\Sch=(\Mnam,\size)$ be a schema of square matrices, that is, there exists an $\alpha$ such 
that $\size(V) \in \{1, \alpha\} \times \{1,\alpha\}$ for every $V \in \Mnam$.
Define the relational vocabulary $\text{WL}(\Sch) = \{R_V \mid V \in \Mnam\}$ such that $\arity(R_V) = 2$ 
if $\size(V) = (\alpha, \alpha)$, $\arity(R_V) = 1$ if $\size(V) \in \{(\alpha,1), (1,\alpha)\}$, and 
$\arity(R_V) = 0$ otherwise.
Then given a matrix instance $\I = (\dom,\conc)$ over $\Sch$ with  $\dom(\alpha) = n$ define the structure 
$\text{WL}(\I) = (\{1, \ldots, n\}, \{R_V^{\I}\} )$ such that 
$R_V^{\I}(i, j) = \conc(V)_{i,j}$ if $\size(V) = (\alpha, \alpha)$, $R_V^{\I}(i) = \conc(V)_{i}$ 
if $\size(V) \in \{(\alpha,1), (1,\alpha)\}$, and $R_V^{\I} = \conc(V)$ if $\size(V) = (1,1)$.

Similar to the proof of Proposition~\ref{prop:sum_to_ara}, for each expression $e(v_1, \ldots, v_k)$ of type $(\alpha, \alpha)$ we must encode in WL the $\alpha$ and the vector variables $v_1, \ldots, v_k$. For this, we use variables $x_{\alpha}^\row$, $x_{\alpha}^\col$, and $x_{v_i}$ for each variable $v_1, \ldots, v_k$. Then we use the following inductive hypothesis (similar to Proposition~\ref{prop:sum_to_ara}):

\newcommand{\varphie}{\varphi_e}
\newcommand{\xr}{x_{\alpha}^\row}
\newcommand{\xc}{x_{\alpha}^\col}

\begin{itemize}
	\item If $e(v_1,\ldots,v_k)$ is of type $(\alpha,\alpha)$ then there exists a WL formula $\varphie(x_{\alpha}^\row,x_{\alpha}^\col, x_{v_1}, \ldots, x_{v_k})$ such that
	$$
	\ssem{\varphie}{\text{WL}(\I)}(\sigma) \ = \ \sem{e}{\I[v_1\gets b_{i_1},\ldots,v_k\gets b_{i_k}]}_{i,j}
	$$
	for assignment $\sigma$ with $\sigma(\xr)=i$, $\sigma(\xc)=j$ and $\sigma(x_{v_s})=i_s$ for $s=1,\ldots, k$.
	
	\item If $e(v_1,\ldots,v_k)$ is of type $(\alpha,1)$ then there exists a WL formula $\varphie(x_{\alpha}^\row, x_{v_1}, \ldots, x_{v_k})$ such that
	$$
	\ssem{\varphie}{\text{WL}(\I)}(\sigma) \ = \ \sem{e}{\I[v_1\gets b_{i_1},\ldots,v_k\gets b_{i_k}]}_{i}
	$$
	for assignment $\sigma$ with $\sigma(\xr)=i$ and $\sigma(x_{v_s})=i_s$ for $s=1,\ldots, k$.
	And similarly for when $e$ is type $(1,\alpha)$.
	
	\item If $e(v_1,\ldots,v_k)$ is of type $(1,1)$ then there exists a WL formula $\varphie( x_{v_1}, \ldots, x_{v_k})$ such that
	$$
	\ssem{\varphie}{\text{WL}(\I)}(\sigma) \ = \ \sem{e}{\I[v_1\gets b_{i_1},\ldots,v_k\gets b_{i_k}]}
	$$
	for assignment $\sigma$ with $\sigma(x_{v_s})=i_s$ for $s=1,\ldots, k$.
\end{itemize}
If we prove the previous statement we are done, because the last bullet is what we want to show when $e$ has no free vector variables. 
Then rest of the proof is to go by induction on the structure of \langprod expressions.
For a WL-formula $\varphi$ and FO-variables $x,y$, we will write  $\varphi[x \mapsto y]$ the formula $\varphi$ when $x$ is replaced with $y$ all over the formula (syntactically).
Let $e$ be a \langprod expression.
\begin{itemize} \itemsep3mm
  \item If $e:=V$ and $\Sch(e)= (\alpha, \alpha)$ then $\varphie:=R_V(\xr, \xc)$. Similarly, if $\Sch(e)$ is of type $(\alpha,1)$, $(1, \alpha)$, or $(1,1)$, then $\varphie:=R_V(\xr)$, $\varphie:=R_V(\xc)$, and $\varphie:=R_V$, respectively.
  
  \item If $e:=v$, for $v\in \{v_1,\ldots ,v_k\}$, and $\Sch(v)= (\alpha,1)$ then $\varphie := \xr = x_v$. Similarly, if $\Sch(v)= (1,\alpha)$ then $\varphie := \xc = x_v$.
  
  \item if $e:= e_1^T$ and $\Sch(e)=(\alpha,\alpha)$ then
  $$
  \varphie:= \varphi_{e_1}[\xr \mapsto \xc, \xc \mapsto \xr].
  $$
  Similarly, if $\Sch(e)$ is equal to $(\alpha,1)$ or $(1,\alpha)$ then $\varphie:=\varphi_{e_1}[\xr \mapsto \xc]$ and $\varphie:=\varphi_{e_1}[\xc \mapsto \xr]$, respectively.

	\item If $e=e_1+e_2$ with $\Sch (e_1)=\Sch (e_2)$, then $\varphie:= \varphi_{e_1} \ksum \varphi_{e_2}$.
	
	\item If $e=f_\odot(e_1,\ldots, e_k)$ with $\Sch(e_i)=\Sch(e_j)$ for all $i,j\in[1,k]$, then $\varphie:= \varphi_{e_1} \kprod \varphi_{e_2} \cdots \kprod \varphi_{e_k}$.
	
	\item If $e=e_1\cdot e_2$ with $\Sch (e_1)=\Sch (e_2)=(\alpha, \alpha)$,  then $\varphie:= \Sigma y. \  \varphi_{e_1}[\xc \mapsto y] \kprod \varphi_{e_2}[\xr \mapsto y]$ where $y$ is a fresh variable not mentioned in $\varphi_{e_1}$ or $\varphi_{e_2}$. Instead, if $\Sch (e_1)= (\alpha', 1)$ and $\Sch (e_2)=(1, \alpha'')$ with $\alpha', \alpha'' \in \{\alpha, 1\}$, then $\varphie := \varphi_{e_1} \kprod \varphi_{e_2}$.
	
	\item If $e=\ssum v. e_1(v)$, then we define $\varphie := \Sigma x_{v}. \  \varphi_{e_1}(x_v)$.

  \item If $e=\qhadprod v. e_1(v)$, then $\varphie := \sprod x_{v}.\  \varphi_{e_1}(x_v)$.
\end{itemize}
From the construction it is now straightforward to check that the inductive hypothesis holds for all cases. To conclude this direction, we have to define $\Phi(e) := \varphie$ for every expression $e$ and we are done.

\medskip
\noindent \textbf{(WL to \langprod)} We now encode weighted structures into matrices and vectors. Let $\Gamma$ be a relational vocabulary 
where $\arity(R) \leq 2$. 
Define $\text{Mat}(\Gamma) = (\Mnam_\Gamma,\size_\Gamma)$ such 
that $\Mnam_\Gamma = \{ V_{R} \mid R \in \Gamma\}$ and $\size_\Gamma(V_{R})$ is equal to 
$(\alpha, \alpha), (\alpha, 1)$, or $(1,1)$ if $\arity(R)=2$, $\arity(R)=1$, or $\arity(R)=0$, 
respectively, for some $\alpha \in \DD$. Similarly, let $\cA = (A, \{R^{\cA}\}_{R \in \Gamma})$ 
be a structure with $A = \{a_1, \ldots, a_n\}$, ordered arbitrarily.
Then we define the matrix instance $\text{Mat}(\cA) = (\dom,\conc)$ such that $\dom(\alpha) = n$, 
$\conc(V_{R})_{i,j} = R^{\cA}(a_i, a_j)$ if $\arity(R)=2$, $\conc(V_{R})_{i,1} = R^{\cA}(a_i)$ if $\arity(R)=1$, 
and $\conc(V_{R})_{1,1} = R^{\cA}$ otherwise.

\newcommand{\evarphi}{e_\varphi}

Similar to the above direction, we have to encode the FO variables of a formula $\varphi$ with vector variables in the equivalent \langprod expression $\evarphi$. For this, for each FO variable $x$ we define a vector variable $v_x$ of type $(\alpha, 1)$. Then for each formula $\varphi(x_1, \ldots, x_k)$ we define an expression $\evarphi(v_{x_1}, \ldots, v_{x_k})$ of type $(1,1)$ such that for every assignment $\sigma$ of $x_1, \ldots, x_k$ we have:
$$
\sem{\evarphi}{\text{Mat}(\cA)[v_{x_1} \gets b_{i_1},\ldots,v_{x_1}\gets b_{i_k}]} \ = \ \ssem{\varphi}{\cA}(\sigma) 
$$
such that $\sigma(x_{s}) = i_s$ for every $s \leq k$. Note that when the formula has no free variables, the proof of the proposition is shown. Finally, we proceed by induction over the formula $\varphi$ over $\Gamma$.
\begin{itemize} \itemsep3mm
  \item If $\varphi:=x=y$, then $\evarphi:= v_x^T \cdot v_{y}$.
  \item If $\varphi:=R(x,y)$, then $\evarphi:=v_x^T \cdot V_R \cdot v_{y}$. Similarly, if $\varphi:=R(x)$ or $\varphi:=R$, then $\evarphi:= V_R^T \cdot v_{x}$  and $\evarphi:= V_R$, respectively. 
  \item If $\varphi = \varphi_1 \ksum \varphi_2$, then $\evarphi:= e_{\varphi_1} + e_{\varphi_2}$.
  \item If $\varphi = \varphi_1 \kprod \varphi_2$, then $\evarphi:= f_\odot(e_{\varphi_1},e_{\varphi_2})$.
  \item If $\varphi = \ssum x.\  \varphi_1$, then $\evarphi :=\ssum v_x.\ e_{\varphi_1}$.
  \item If $\varphi = \qhadprod x. \varphi_1$, then $\evarphi := \sprod v_x.\ e_{\varphi_1}$.
\end{itemize}
The inductive hypothesis can be proved following the above construction. To finish the proof, we define $\Psi(\varphi) := \evarphi$ and the proposition is shown.

\end{proof}

%% file: sections/app-asset-order.tex

We conclude by verifying that the fragment defined in Section \ref{subsec:langlinear}, i.e, 
\langmprod extended with order and $f_{>0}$, 
can perform matrix inversion and compute the determinant. To this aim, we verify that all order
predicates in Section \ref{app:order} can be derived using $\ssum$, $\sprod$, $f_{>0}$ and 
$e_{S_{<}}$. Given this, it suffices to observe that Csanky's algorithm, as shown in Section~\ref{app:inverse}, only relies on expressions using $\ssum$ and $\sprod$ and order information on canonical vectors and $f_/$.
As consequence, our fragment can perform matrix inversion and compute the determinant.

It remains to show that if we have $e_{S_{<}}$, using $\ssum$ and $\sprod$ and $f_{>0}$ we can
can define all order predicates from Section~\ref{app:order}. We note that due to the restricted for-loops
in $\ssum$ and $\sprod$, we do not have access to the intermediate
result in the iterations and as such, it is unclear whether order information can be computed. This is why
we assume access to $e_{S_<}$.

We first remark that if we have $e_{S_{<}}$, we can also obtain
 $e_{S_{\leq}}$ by adding $e_{\mathsf{Id}}$. Hence,
%
%
we can compute $\mathsf{succ}$ and $\mathsf{succ}^+$ as well. Furthermore, 
\begin{align*}
  e_{\mathsf{min}}&:=\ssum v. \left[ \sprod w. \mathsf{succ}(w,v)\right] \times v. \\
  e_{\mathsf{max}}&:=\ssum v. \left[ \sprod w. \left( 1-\mathsf{succ}(w,v) \right) \right] \times v.
\end{align*}
Both expressions are only using $\ssum$ and $\sprod$ and $\mathsf{succ}$, so are in our fragment.
Furthermore, if we have $f_{>0}$ then we can define
$$
e_{\mathsf{Pred}}:= e_{S_{<}}- f_{>0}(e_{S_{<}}^2)
$$
Also, recall that  $e_{\mathsf{Next}}:=e_{\mathsf{Pred}}^T$. As a consequence, 
we can now define $\mathsf{prev}(v)$ and $\mathsf{next}(v)$ as in \ref{app:order}. Similarly,
it is readily verified that also $e_{\mathsf{getPrevMatrix}}(V)$,
$e_{\mathsf{getNextMatrix}}(V)$, $e_{\mathsf{min}+i}$ and $e_{\mathsf{max}+i}$ can be expressed
in our fragment.